\documentclass[twoside]{article}

\PassOptionsToPackage{round}{natbib}
\usepackage[preprint]{aistats2027}
\usepackage[utf8]{inputenc}
\usepackage[T1]{fontenc}
\usepackage{url}
\usepackage{booktabs}
\usepackage{nicefrac}
\usepackage{microtype}

\usepackage{amsmath}
\usepackage{amssymb}
\usepackage{amsthm}
\usepackage{bbm}
\usepackage{graphicx}
\usepackage{placeins}
\graphicspath{{figures/}}
\usepackage{subcaption}
\usepackage{algorithm}
\usepackage{algpseudocode}
\usepackage{hyperref}
\usepackage{natbib}
\usepackage{dsfont}
\usepackage{xcolor}
\usepackage{enumitem}
\usepackage{tikz}

\algrenewcommand\algorithmicindent{1.0em}%
\newtheorem{theorem}{Theorem}[section]
\newtheorem{lemma}[theorem]{Lemma}
\newtheorem{corollary}[theorem]{Corollary}

\newtheorem{remark}[theorem]{Remark}

\definecolor{airforceblue}{rgb}{0.36, 0.54, 0.66}
\definecolor{aliceblue}{rgb}{0.94, 0.97, 1.0}
\definecolor{bluegray}{rgb}{0.4, 0.6, 0.8}
\definecolor{bluefill}{rgb}{0.5, 0.7, 0.9}
\definecolor{cornflowerblue}{rgb}{0.39, 0.58, 0.93}
\definecolor{coolblack}{rgb}{0.0, 0.18, 0.39}
\definecolor{coolblack2}{rgb}{0.05, 0.23, 0.49}
\definecolor{beaublue}{rgb}{0.74, 0.83, 0.9}
\hypersetup{
  colorlinks,
  linkcolor={blue},
  citecolor={blue},
  urlcolor={blue}
}

\allowdisplaybreaks[1]

\newcommand{\argmax}{\mathop{\rm arg~max}\limits}
\newcommand{\argmin}{\mathop{\rm arg~min}\limits}
\newcommand{\R}{\mathbb{R}}
\newcommand{\E}{\mathbb{E}}
\newcommand{\Pp}{\mathbb{P}}
\newcommand{\1}{\mathbf{1}}

\newcommand{\cS}{\mathcal S}
\newcommand{\simplex}{\Delta}
\newcommand{\Gap}{\operatorname{Gap}}
\newcommand{\NE}{\operatorname{NE}}
\newcommand{\normmax}[1]{\left\lVert #1\right\rVert_{\max}}
\newcommand{\normone}[1]{\left\lVert #1\right\rVert_1}
\newcommand{\trans}{\mathsf T}

\newcommand{\supp}{\operatorname{supp}}

\title{Thompson Sampling for Learning in Zero-Sum Games}

\makeatletter
\let\papertitle\@title
\makeatother

\begin{document}

\twocolumn[
  \aistatstitle{\papertitle}
  \aistatsauthor{
  Kenshi Abe$^1$ \And
  Kaito Ariu$^{1,2}$ \And
  Po-An Wang$^3$
}
\aistatsaddress{
  $^1$CyberAgent \And
  $^2$Tohoku University \And
  $^3$National Tsing Hua University
}

]

\begin{abstract}
Thompson sampling (TS) achieves logarithmic and asymptotically optimal regret in classical stochastic bandits, but its behavior in multi-agent learning remains less understood.
We study TS in an unknown finite two-player zero-sum matrix game with Gaussian bandit feedback.
Under centralized self-play, TS samples a single payoff matrix from the posterior and plays a Nash equilibrium of the sampled game.
For every fixed zero-sum game with a unique Nash equilibrium, we show that the expected cumulative duality gap is $O(\sqrt T)$.
The rate improves to $O(\log T)$ for pure equilibria.
We then identify the exact coefficient of the $\sqrt T$ term, which depends only on the equilibrium strategies and the noise level.
We also establish a minimax lower bound matching the square-root dependence on $T$ and an instance-dependent lower bound matching the logarithmic dependence on $T$ for a class of games with a unique pure equilibrium.
Experiments support the predicted rates and coefficient.
\end{abstract}

\section{Introduction}
\label{sec:intro}

Thompson sampling (TS) samples a model from the posterior and selects an optimal action for that model \citep{thompson1933likelihood,russo2018tutorial}.
Its logarithmic regret in Bernoulli bandits \citep{agrawal2012analysis} and asymptotic optimality in one-dimensional exponential-family bandits \citep{kaufmann2012thompson,korda2013thompson} motivate applying posterior sampling to equilibrium learning in unknown games.

Learning in multiplayer games under bandit feedback has been extensively studied \citep{odonoghue2021matrix,ito2026adversarial}.
Recent methods provide guarantees for learning approximate Nash equilibria using algorithms that explicitly encourage exploration of uncertain payoffs \citep{yang2025incentivize,nayak2026logarithmic,xiang2026policy}.
Yet the behavior of plain TS, which simply samples a payoff matrix from the posterior and plays a Nash equilibrium of the sampled game, remains poorly understood.
Therefore, we ask:
\begin{quote}
  \emph{Can plain Thompson sampling learn equilibrium play in an unknown zero-sum game?}
\end{quote}

We answer this question affirmatively for games with a unique Nash equilibrium under centralized self-play, where a single posterior draw jointly determines both players’ strategies.
The main technical challenge is to establish that posterior sampling alone generates enough information for equilibrium learning.
While playing an arm in a single-player bandit provides information about its reward, playing a sampled equilibrium in a game need not provide the information needed to correct inaccurate payoff estimates.

We evaluate the played strategies through the cumulative duality gap, which sums both players' gains from unilateral best responses over all rounds.
For every fixed game with a unique Nash equilibrium, we show that the expected cumulative duality gap is $O(\sqrt T)$ and improves to $O(\log T)$ when the equilibrium is pure.
These rates reflect two contributions, a square-root term for estimating payoffs within the equilibrium supports and a logarithmic term for excluding actions outside them.
The former vanishes for pure equilibria, while the latter vanishes for fully mixed equilibria.
Beyond these rates, we identify the exact coefficient of the $\sqrt T$ term.
This coefficient depends only on the equilibrium strategies and the noise level.
Consequently, payoff entries outside the equilibrium supports contribute only lower-order terms.

We complement these results with two lower bounds.
The square-root dependence on $T$ is unavoidable in the worst case, even for games with a unique fully mixed Nash equilibrium.
For a subclass of games with a unique pure equilibrium, we also show that the logarithmic dependence is unavoidable on each fixed instance for uniformly efficient policies.

Our contributions are summarized as follows.
First, we show that plain TS achieves an $O(\sqrt T)$ expected cumulative duality gap for every fixed game with a unique Nash equilibrium, improving to $O(\log T)$ for pure equilibria (Theorem~\ref{thm:unique-rate}).
Second, we identify the exact coefficient of the $\sqrt T$ term, which depends only on the equilibrium strategies and the noise level (Theorem~\ref{thm:exact-core-asymptotics}).
Finally, we establish that the square-root dependence on $T$ is unavoidable in the worst case, even for games with a unique fully mixed Nash equilibrium (Theorem~\ref{thm:unique-minimax-lower}).
For uniformly efficient policies, the logarithmic dependence on $T$ is also unavoidable on fixed instances in a subclass of games with a unique pure equilibrium (Theorem~\ref{thm:pure-separable-lower}).

\section{Related Literature}
\label{sec:related}

\paragraph{Thompson sampling in stochastic bandits.}
\citet{russo2014learning} establish Bayesian regret bounds for posterior sampling in bandit problems with dependent rewards across actions.
For Bernoulli bandits, TS admits logarithmic and problem-independent regret bounds \citep{agrawal2012analysis,agrawal2013further}.
It is also asymptotically optimal for Bernoulli and one-dimensional exponential-family bandits \citep{kaufmann2012thompson,korda2013thompson}.
Our work studies TS in unknown games, where the learner plays a Nash equilibrium of a sampled payoff matrix.

\paragraph{Learning in games with bandit feedback.}
Learning in zero-sum matrix games with bandit feedback has been studied both against arbitrary opponents \citep{odonoghue2021matrix,odonoghue2021variational,maiti2026limitations,ito2026adversarial} and in self-play \citep{ito2025instance}.
In self-play, \citet{ito2025instance} establish instance-dependent regret bounds, achieving logarithmic regret when a pure Nash equilibrium exists.
Other studies address equilibrium identification through active entry sampling \citep{maiti2023instance,maiti2024pure} or last-iterate convergence without opponent-action observations \citep{cai2023uncoupled,fiegel2025harder,fiegel2026optimal}.
More closely related to our criterion, \citet{yang2025incentivize,nayak2026logarithmic} establish a $\widetilde O(\sqrt T)$ cumulative duality gap bound, where $\widetilde O$ hides logarithmic factors.
The last-iterate guarantee of \citet{hait2026near} also yields a high-probability $\widetilde O(\sqrt T)$ cumulative duality gap bound by summation.

\paragraph{Posterior sampling for games.}
In matrix games, \citet{odonoghue2021matrix} show that TS can suffer linear regret against an opponent whose play prevents learning the unknown payoffs.
Posterior sampling has also been studied in zero-sum Markov games \citep{xiong2022selfplay,qiu2023posterior}, including methods with Bayesian regret guarantees \citep{jafarnia2024bayesian,yueh2026bayesian}.
Most recently, the TS variant of \citet{xiang2026policy} for matrix games combines regularized equilibrium computation with continued explicit exploration and admits a cumulative duality gap bound of order $\widetilde O(T^{3/4})$ in probability, allowing multiple equilibria.
In contrast, Nash-TS uses neither regularization nor continued forced exploration after initialization.
We jointly select both strategies from one posterior sample and, for a fixed game with a unique equilibrium, obtain $O(\sqrt T)$ expected cumulative duality gap and $O(\log T)$ for pure equilibria.
Furthermore, we identify the exact leading coefficient of the square-root term for Nash-TS.

\section{Preliminaries}
\label{sec:setting}

\paragraph{Two-player zero-sum matrix games.}
In this study, we consider a two-player zero-sum matrix game with a fixed but unknown payoff matrix $A^\star\in\R^{m\times n}$.
The row player and the column player choose their mixed strategies $x\in\simplex_m$ and $y\in\simplex_n$ over their $m$ and $n$ actions, respectively, where $\simplex_d:=\{z\in\R_+^d:\sum_{i=1}^d z_i=1\}$ denotes the probability simplex.
We write $[d]:=\{1,\ldots,d\}$ for the set of action indices.
The expected payoffs to the row player and the column player are $x^\trans A^\star y$ and $-x^\trans A^\star y$, respectively.
Thus, the row player aims to maximize $x^\trans A^\star y$, while the column player aims to minimize it, giving the max-min problem
\[
  \max_{x\in\simplex_m}\min_{y\in\simplex_n} x^\trans A^\star y.
\]
We refer to $(x,y)\in\simplex_m\times\simplex_n$ as a \emph{strategy profile}.
A strategy profile is called \emph{pure} if each player assigns probability one to a single action, and \emph{fully mixed} if each player assigns positive probability to every action.

\paragraph{Nash equilibrium and duality gap.}
A \emph{Nash equilibrium} \citep{nash1951non} is a strategy profile at which neither player can improve its expected payoff by changing its own strategy.
Formally, for a payoff matrix $A\in\R^{m\times n}$, a strategy profile $(x^\star,y^\star)\in\simplex_m\times\simplex_n$ is a Nash equilibrium if and only if
\[
  \begin{aligned}
    x^\trans A y^\star
     & \le
    (x^\star)^\trans A y^\star
    \le
    (x^\star)^\trans A y,                                               \\
     & \hspace{2em}\text{for every }(x,y)\in\simplex_m\times\simplex_n.
  \end{aligned}
\]
We denote the set of Nash equilibria of $A$ by $\NE(A)$.

To quantify the proximity of a strategy profile $(x,y)$ to a Nash equilibrium, we use the \emph{duality gap} (also called the Nash gap), defined by
\begin{equation*}
  \Gap_A(x,y)
  :=
  \max_{i\in[m]} e_i^\trans A y
  -
  \min_{j\in[n]} x^\trans A e_j,
\end{equation*}
where $e_i$ denotes the standard basis vector whose dimension is determined by context.
The duality gap is the sum of the two players' gains from switching to a best response against the opponent's current strategy.
It is nonnegative for every $(x,y)\in\simplex_m\times\simplex_n$ and equals zero if and only if $(x,y)\in\NE(A)$.

\paragraph{Problem setting.}
As in \citet{yang2025incentivize,nayak2026logarithmic,hait2026near,xiang2026policy}, we consider online learning with \emph{bandit feedback} and observed actions of both players.
We study the \emph{centralized} setting, where a single learner chooses the strategies of both players.
At each round $t=1,\ldots,T$, the learner chooses both $x_t\in\simplex_m$ and $y_t\in\simplex_n$ based on previously observed actions and payoffs.
Then, the learner draws actions $I_t\sim x_t$ and $J_t\sim y_t$ independently.
The learner observes the selected action pair $(I_t,J_t)$ and receives only its noisy payoff
\begin{equation}
  Y_t=A^\star_{I_tJ_t}+\xi_t,
  \qquad
  \xi_t\stackrel{\mathrm{i.i.d.}}{\sim}\mathcal N(0,\sigma^2),
  \label{eq:observation}
\end{equation}
where the noise variance $\sigma^2>0$ is known.
The payoffs of all unselected action pairs remain unobserved.

We use the \emph{cumulative duality gap} to evaluate the quality of the strategy profiles played throughout learning.
For any payoff matrix $A\in\R^{m\times n}$, we define the cumulative duality gap by
\begin{equation}
  \mathcal R_T(A)
  :=
  \sum_{t=1}^T \Gap_A(x_t,y_t).
  \label{eq:regret}
\end{equation}
For the true payoff matrix $A^\star$, we write $\mathcal R_T:=\mathcal R_T(A^\star)$.
Our objective is to minimize $\E[\mathcal R_T]$, where the expectation is taken over the observation noise and the learner's randomization, with $A^\star$ fixed.
A sublinear bound on $\E[\mathcal R_T]$ ensures that the duality gap of the played strategy profiles vanishes on average in expectation.
It also guarantees convergence of the averaged strategies: since $\Gap_{A^\star}(x,y)$ is jointly convex in $(x,y)$, they satisfy
\[
  \E\!
  \left[
    \Gap_{A^\star}\!\left(
    \frac1T\sum_{t=1}^T x_t,
    \frac1T\sum_{t=1}^T y_t
    \right)
    \right]
  \le
  \frac{\E[\mathcal R_T]}{T}.
\]

The cumulative duality gap is also related to standard regret measures.
It equals the sum of the two players' individual dynamic regrets, where each player compares its strategy with a best response to the opponent's current strategy \citep{cai2023doubly}.
In the special case where either player has only one action, it reduces to the stochastic-bandit pseudo-regret of the other player.

\paragraph{Other notation.}
For a matrix $B$, we write $\normmax{B}:=\max_{i,j}|B_{ij}|$.
For $B\in\R^{m\times n}$ and index sets $I\subseteq[m]$ and $J\subseteq[n]$, we write $B_{I,J}:=(B_{ij})_{i\in I,j\in J}$ for the corresponding submatrix, with rows and columns in increasing index order.
Similarly, $z_I:=(z_i)_{i\in I}$ denotes the restriction of a vector $z$ to $I$ in increasing index order.
Every indexed tuple $(z_i)_{i\in I}$ is regarded as a column vector.

\section{Nash Thompson Sampling}
\label{sec:algorithm}

This section describes Nash Thompson Sampling (Nash-TS), the plain Thompson-sampling algorithm analyzed in this paper.
At each round, the learner samples a payoff matrix from the posterior and selects a Nash equilibrium of the sampled game.
It then draws the two players' actions independently from the selected strategies and observes the noisy payoff at the realized action pair.
Using this observation, it updates only the posterior distribution of the observed entry.

\subsection{Posterior Sampling}

We place an independent Jeffreys prior on each payoff entry.
Under the Gaussian observation model in \eqref{eq:observation}, the prior on $A_{ij}$ is flat for each $(i,j)\in[m]\times[n]$ and satisfies
\begin{equation}
  \pi_J(A_{ij})\propto 1,
  \qquad A_{ij}\in\R.
  \label{eq:jeffreys}
\end{equation}
To obtain a proper posterior under \eqref{eq:jeffreys}, the learner plays each pure strategy pair $(e_i,e_j)$ once in a fixed order during the first $mn$ rounds.
For each entry $a=(i,j)\in[m]\times[n]$, let $\widehat A_{a,1}$ denote the noisy payoff observed in these rounds.
After these $mn$ rounds, the learner repeats the following procedure for $s=1,2,\ldots$.
We write $Y_s=A^\star_{I_sJ_s}+\xi_s$ for the noisy payoff observed at action pair $(I_s,J_s)$ in round $s$, with the noise specified in \eqref{eq:observation}.

To sample a payoff matrix at round $s$, the learner uses the observations collected so far.
The observation count and empirical mean of entry $a$ at this round are
\[
  \begin{aligned}
    N_a(s) & :=1+\sum_{r=1}^{s-1}\1_{\{(I_r,J_r)=a\}},              \\
    \widehat A_{a,N_a(s)}
           & :=\frac{\widehat A_{a,1}
                 +\sum_{r=1}^{s-1}Y_r\1_{\{(I_r,J_r)=a\}}}{N_a(s)}.
  \end{aligned}
\]
Using these empirical means, the learner samples the entries of $\widetilde A_s$ independently from their posteriors:
\begin{equation}
  \widetilde A_{a,s}
  \sim
  \mathcal N\!
  \left(\widehat A_{a,N_a(s)},\frac{\sigma^2}{N_a(s)}\right).
  \label{eq:posterior}
\end{equation}

\subsection{Equilibrium Selection and Posterior Update}

Given the sampled payoff matrix $\widetilde A_s$, the learner selects $(x_s,y_s)\in\NE(\widetilde A_s)$ using an arbitrary but fixed equilibrium selection rule.
It then draws $I_s\sim x_s$ and $J_s\sim y_s$ independently and observes $Y_s$.
Because the two actions are drawn independently, their conditional distribution satisfies
\[
  \Pp\big((I_s,J_s)=(i,j)\mid\widetilde A_s\big)=x_{s,i}\,y_{s,j}.
\]
The learner uses $Y_s$ to update only the posterior distribution of entry $(I_s,J_s)$.
Nash-TS is an anytime algorithm and does not require knowledge of the horizon.
Algorithm~\ref{alg:ts-ne} summarizes the complete procedure.

\begin{algorithm}[t]
  \caption{Nash-TS}
  \label{alg:ts-ne}
  \begin{algorithmic}[1]
    \Require noise variance $\sigma^2$
    \Statex \textbf{Phase I: observe each entry once}
    \State observe every entry $(i,j)\in[m]\times[n]$ exactly once, in an arbitrary fixed order
    \Statex \textbf{Phase II: Thompson sampling}
    \For{$s=1,2,\ldots$}
    \State $N_{ij}(s)\gets$ observation count of entry $(i,j)$ before this round
    \State sample $\widetilde A_{ij,s}\sim\mathcal N\!\big(\widehat A_{ij,N_{ij}(s)},\,\sigma^2/N_{ij}(s)\big)$ independently over all $(i,j)$
    \State compute a Nash equilibrium $(x_s,y_s)\in\NE(\widetilde A_s)$
    \State draw $I_s\sim x_s$ and $J_s\sim y_s$ independently
    \State observe $Y_s=A^\star_{I_sJ_s}+\xi_s$
    \State update only the posterior of entry $(I_s,J_s)$
    \EndFor
  \end{algorithmic}
\end{algorithm}

\section{Cumulative Gap Upper Bounds}
\label{sec:upper-bounds}
In this section, we establish upper bounds on the expected cumulative duality gap of Nash-TS in Algorithm~\ref{alg:ts-ne}.
Throughout this section, we assume that $A^\star$ has a unique Nash equilibrium $(x^\star,y^\star)$.
We write $I^\star:=\supp(x^\star)$ and $J^\star:=\supp(y^\star)$.

\subsection{Finite-Horizon Bound}
\label{sec:finite-horizon-regret}

We first show that, under the uniqueness assumption, Nash-TS achieves an expected cumulative duality gap of $\mathcal O(\sqrt T)$ in general and $\mathcal O(\log T)$ when the equilibrium is pure.

\begin{theorem}
  \label{thm:unique-rate}
  Suppose that $m,n\ge2$, $T>mn$, and $A^\star$ has a unique Nash equilibrium.
  Then, there exist finite nonnegative constants $C_{\rm mix},C_{\rm out},C_0$, depending only on $(A^\star,\sigma,m,n)$, such that Nash-TS satisfies
  \begin{equation*}
    \E[\mathcal R_T]
    \le C_{\rm mix}\sqrt T+C_{\rm out}\log(eT)+C_0.
  \end{equation*}
  Moreover, $C_{\rm mix}=0$ when the unique equilibrium is pure, and $C_{\rm out}=0$ when it is fully mixed.
\end{theorem}

For games with a unique pure Nash equilibrium, our bound achieves the same logarithmic dependence on $T$ as the external regret bound of \citet{ito2025instance}.
The cumulative duality gap upper-bounds the sum of the two players' external regrets and thus provides a stronger performance guarantee.
Notably, Nash-TS achieves this bound without explicit regularization or exploration after initialization.

The different rates for mixed and pure equilibria arise from how payoff estimation errors affect the chosen strategies.
For a mixed equilibrium, these errors can change the probabilities assigned to its actions even when the sampled equilibrium uses the same supports, and the square-root term accounts for this discrepancy.
For a pure equilibrium, if the sampled equilibrium uses the same supports as the true equilibrium, both players necessarily choose their equilibrium actions with probability one, so the square-root term vanishes.
Thus, for a pure equilibrium, a positive duality gap arises only when the sampled equilibrium uses an action outside the true supports.
The logarithmic term accounts for such use of actions outside the supports.

This separation also clarifies how the constants scale with the number of actions.
Write $d:=|I^\star|=|J^\star|$ for the equilibrium support size.
\footnote{The equality of the support sizes follows from uniqueness by Lemma~\ref{proof:lem:local-geometry-from-uniqueness}.}
With the instance-specific geometric quantities and $\sigma$ fixed, the explicit constants satisfy
\begin{equation*}
  \begin{aligned}
    C_{\rm mix} & =\mathcal O\!
    \left(d\sqrt{\log(ed)}\right),                                   \\
    C_{\rm out} & =\mathcal O\!\left(d(m+n-2d)(1+\log(mn))^3\right).
  \end{aligned}
\end{equation*}
Thus, the square-root coefficient depends on the support size, while the logarithmic coefficient also accounts for actions outside the supports.
Appendix~\ref{proof:sec:proof-unique-rate} gives the proof and explicit constants in Equations~\eqref{proof:eq:mix-constant}--\eqref{proof:eq:additive-constant}.

\subsection{Exact Asymptotics}
\label{sec:exact-asymptotics}

Theorem~\ref{thm:unique-rate} bounds the square-root contribution for non-pure equilibria but does not identify its exact leading coefficient.
We next show that Nash-TS actually incurs this order of cumulative duality gap and identify its leading coefficient.
Let $\mathsf W^{\rm core}$ have independent entries indexed by
$I^\star\times J^\star$ with
\begin{equation}
  \mathsf W_{ij}^{\rm core}
  \sim
  \mathcal N\!\left(0,\frac{2\sigma^2}{x_i^\star y_j^\star}\right).
  \label{eq:core-limit-matrix}
\end{equation}

The matrix $\mathsf W^{\rm core}$ describes the limiting $\sqrt{s}$-scaled errors in the sampled support submatrix.
Each entry $(i,j)$ is observed asymptotically at frequency $x_i^\star y_j^\star$, so less frequently observed entries have larger limiting error variances.
The factor $2$ combines the equal contributions of empirical-mean error and fresh posterior-sampling noise.
Define
\begin{equation}
  \Lambda_{\rm core}
  :=
  \E\!\left[
    \max_{j\in J^\star}(x^\star)_{I^\star}^\trans
    \mathsf W^{\rm core}e_j
    -\min_{i\in I^\star}e_i^\trans
    \mathsf W^{\rm core}(y^\star)_{J^\star}
    \right].
  \label{eq:core-asymptotic-constant}
\end{equation}
The quantity $\Lambda_{\rm core}$ is the expected first-order duality gap arising when the equilibrium of a locally perturbed support submatrix is evaluated in the true game.
The following theorem identifies $2\Lambda_{\rm core}$ as the exact leading coefficient of $\E[\mathcal R_T]$ on the square-root scale.

\begin{theorem}
  \label{thm:exact-core-asymptotics}
  Suppose that $m,n\ge2$ and $A^\star$ has a unique Nash equilibrium.
  Then, Nash-TS satisfies
  \begin{equation*}
    \lim_{T\to\infty}
    \frac{\E[\mathcal R_T]}{\sqrt T}
    =2\Lambda_{\rm core}.
  \end{equation*}
  Moreover, $\Lambda_{\rm core}=0$ when the unique Nash equilibrium is pure, whereas $\Lambda_{\rm core}>0$ otherwise.
\end{theorem}

Thus, for every fixed game with a unique non-pure equilibrium, Nash-TS has an expected cumulative duality gap of $\Theta(\sqrt T)$.
With $\sigma$ fixed, $\Lambda_{\rm core}$ depends on $A^\star$ only through its equilibrium strategies and supports.
The following corollary bounds $\Lambda_{\rm core}$ in terms of the equilibrium probabilities and support size $d$ and evaluates it for equilibria that are uniform on their supports.

\begin{corollary}
  \label{cor:core-coefficient-bound}
  Under the assumptions of Theorem~\ref{thm:exact-core-asymptotics}, when the unique Nash equilibrium is not pure,
  \begin{equation*}
    \Lambda_{\rm core}
    \le
    2\sigma\sqrt{\log d}
    \left[
      \Big(\min_{j\in J^\star}y_j^\star\Big)^{-1/2}
      +\Big(\min_{i\in I^\star}x_i^\star\Big)^{-1/2}
      \right].
  \end{equation*}
  Moreover, if the equilibrium is uniform on its supports, so that $x_i^\star=y_j^\star=1/d$ for all $i\in I^\star$ and $j\in J^\star$,
  \begin{equation*}
    \Lambda_{\rm core}
    =
    2\sigma\sqrt{2d}\,
    \E\!\left[\max_{1\le k\le d}
      G_k\right],
  \end{equation*}
  where $G_1,\ldots,G_d$ are independent standard normal variables, and
  \begin{equation*}
    \frac{3\sqrt2}{8}\sigma\sqrt{d\log d}
    \le\Lambda_{\rm core}\le4\sigma\sqrt{d\log d}.
  \end{equation*}
\end{corollary}

For a uniform equilibrium on the supports, the leading coefficient therefore scales as $\sigma\sqrt{d\log d}$, independently of the number of actions outside the supports.
Appendix~\ref{proof:sec:proof-exact-core-asymptotics} gives the proofs.

\subsection{Proof Sketch}
\label{sec:proof-overview}

We sketch the proofs of Theorems~\ref{thm:unique-rate} and \ref{thm:exact-core-asymptotics}.
The complete proofs are given in Appendices~\ref{proof:sec:proof-unique-rate} and \ref{proof:sec:proof-exact-core-asymptotics}.

Uniqueness of the true equilibrium implies that there exists $\rho>0$ such that every matrix $B\in\R^{d\times d}$ satisfying
\[
  \normmax{B-A^\star_{I^\star,J^\star}}\le\rho
\]
has a unique fully mixed equilibrium, as shown in Lemma~\ref{proof:lem:local-geometry-from-uniqueness}.
We use this fact to separate rounds according to the accuracy of the sampled support submatrix.
Let $\mathsf{Near}_s$ denote the event that the above condition holds for $B=\widetilde A_{I^\star,J^\star,s}$, that is,
\[
  \mathsf{Near}_s
  :=\left\{\normmax{\widetilde A_{I^\star,J^\star,s}-A^\star_{I^\star,J^\star}}\le\rho\right\}.
\]
On $\mathsf{Near}_s$, write $(\bar x_s,\bar y_s)$ for the corresponding equilibrium extended by zeros outside $I^\star$ and $J^\star$.
However, in the sampled game $\widetilde A_s$, an action outside $I^\star$ or $J^\star$ may be profitable against $(\bar x_s,\bar y_s)$.
To distinguish whether the learner plays this pair in the full game, we define
\[
  \mathsf{Supp}_s
  :=\mathsf{Near}_s\cap\{(x_s,y_s)=(\bar x_s,\bar y_s)\}.
\]
Splitting online rounds according to these events gives
\begin{align}
   & \E[\mathcal R_T]\le2mn\normmax{A^\star}
  +2\normmax{A^\star}
  \sum_{s=1}^{T-mn}\Pp(\mathsf{Near}_s^c)
  \notag                                                                     \\
   & \quad+\E\!\left[
                 \sum_{s=1}^{T-mn}\1_{\mathsf{Supp}_s}\Gap_{A^\star}(x_s,y_s)
                 \right]
  \notag                                                                     \\
   & \quad+2\normmax{A^\star}\sum_{s=1}^{T-mn}
  \Pp(\mathsf{Near}_s\cap\{(x_s,y_s)\ne(\bar x_s,\bar y_s)\}),
  \label{eq:proof-sketch-decomposition}
\end{align}
where the first term bounds the cumulative duality gap of the initial $mn$ rounds.
The following three steps bound the remaining terms in order.

\paragraph{(1) Rounds with inaccurate sampled matrices.}
If the observation counts $N_{ij}(s)$ grow for all $(i,j)\in I^\star\times J^\star$, Gaussian concentration makes the probability of $\mathsf{Near}_s^c$ tend to zero.
However, Nash-TS has no explicit exploration mechanism after initialization, so we cannot assume that these counts grow.
Instead, we show that, conditional on the observation-noise sequences, posterior draws satisfying
\begin{align*}
  |\widetilde A_{ij,s}-A^\star_{ij}|\le\rho
  ~~\text{for all }(i,j)\in([m]\times J^\star)\cup(I^\star\times[n])
\end{align*}
occur with probability at least a positive constant that does not depend on the observation counts or $T$.
Such draws ensure that the learner plays $(\bar x_s,\bar y_s)$ with $\bar x_{s,i},\bar y_{s,j}\ge p>0$, so each entry in $I^\star\times J^\star$ is observed with probability at least $p^2$.
Combining these observation probability bounds with Gaussian error estimates gives
\begin{equation}
  \sum_{s=1}^{T-mn}\Pp(\mathsf{Near}_s^c)=\mathcal O(1).
  \label{eq:proof-sketch-localization-bound}
\end{equation}

\paragraph{(2) Contribution from the equilibrium supports.}
Using $\bar x_{s,i},\bar y_{s,j}\ge p$ and the observation probabilities $\bar x_{s,i}\bar y_{s,j}$ on $\mathsf{Supp}_s$, we bound the contribution of these rounds to the expected cumulative duality gap as
\begin{equation*}
  \begin{aligned}
     & \E\!\left[\sum_{s=1}^{T-mn}
             \1_{\mathsf{Supp}_s}\Gap_{A^\star}(x_s,y_s)\right]                \\
     & \quad\le\frac2p\E\!\left[\sum_{s=1}^{T-mn}
                            \1_{\{(I_s,J_s)\in I^\star\times J^\star\}}
                            |\widetilde A_{I_sJ_s,s}-A^\star_{I_sJ_s}|\right].
  \end{aligned}
\end{equation*}
Thus, it is enough to bound the errors $|\widetilde A_{ij,s}-A^\star_{ij}|$ for observed entries only within the equilibrium supports, $(i,j)\in I^\star\times J^\star$.
Each time entry $(i,j)$ is observed, its observation count $N_{ij}(s)$ increases by one.
As this count increases, the typical size of $|\widetilde A_{ij,s}-A^\star_{ij}|$ decreases at rate $N_{ij}(s)^{-1/2}$.
Summing this $N^{-1/2}$ rate over successive observation counts and using the above inequality gives
\begin{equation}
  \E\!
  \left[\sum_{s=1}^{T-mn}
    \1_{\mathsf{Supp}_s}\Gap_{A^\star}(x_s,y_s)\right]
  =\mathcal O(\sqrt T).
  \label{eq:proof-sketch-core-bound}
\end{equation}
Note that, for $d=1$, the learner always plays the true equilibrium $(x^\star,y^\star)$ when $\mathsf{Supp}_s$ holds, so $\sum_{s=1}^{T-mn}\1_{\mathsf{Supp}_s}\Gap_{A^\star}(x_s,y_s)=0$.

\paragraph{(3) Contribution from outside the equilibrium supports.}
When $\mathsf{Near}_s$ holds but $(x_s,y_s)\ne(\bar x_s,\bar y_s)$, at least one entry in $\cS_\times:=(I^{\star c}\times J^\star)\cup(I^\star\times J^{\star c})$ has not been observed sufficiently often, except in rounds with rare large posterior-sampling errors.
After $b_T=\mathcal O(\log T)$ observations of an entry, Gaussian error bounds make such large errors rare enough to contribute only $\mathcal O(1)$ rounds in expectation.
Thus, apart from these rare errors, every outside action that appears profitable still has an entry with $N_{ij}(s)<b_T$.
It therefore remains to bound $\E[\sum_{s=1}^{T-mn}\1_{\mathsf D_s(b_T)}]$, where $\mathsf D_s(b_T)$ denotes $\mathsf{Near}_s\cap\{(x_s,y_s)\ne(\bar x_s,\bar y_s)\}$ with this observation-count condition.
The technical difficulty is that $\mathsf D_s(b_T)$ does not guarantee that an entry in $\cS_\times$ with $N_{ij}(s)<b_T$ is observed.
Nevertheless, we bound the expected number of these rounds using the number of observations of such entries, as shown in Lemma~\ref{proof:lem:outside-action-occupation}.
Since each entry can be observed while $N_{ij}(s)<b_T$ at most $b_T-1$ times, we obtain
\[
  \E\!
  \left[
    \sum_{s=1}^{T-mn}\1_{\mathsf D_s(b_T)}
    \right]
  =\mathcal O(b_T)
  =\mathcal O(\log T).
\]
Adding the $\mathcal O(1)$ contribution from rounds with large posterior-sampling errors, we obtain
\begin{equation}
  \sum_{s=1}^{T-mn}\Pp(\mathsf{Near}_s\cap\{(x_s,y_s)\ne(\bar x_s,\bar y_s)\})
  =\mathcal O(\log T).
  \label{eq:proof-sketch-outside-bound}
\end{equation}

\paragraph{(4) Cumulative gap bound and exact leading coefficient.}
Substituting \eqref{eq:proof-sketch-localization-bound}, \eqref{eq:proof-sketch-core-bound}, and \eqref{eq:proof-sketch-outside-bound} into \eqref{eq:proof-sketch-decomposition} proves Theorem~\ref{thm:unique-rate}.
Note that the square-root contribution vanishes for $d=1$ by step~(2), while the logarithmic contribution vanishes for full supports because $\mathsf{Supp}_s=\mathsf{Near}_s$.

Steps (1) and (3) give
\begin{equation*}
  \begin{aligned}
    \E[\mathcal R_T]
    =
    \sum_{s=1}^{T-mn}
    \E\!\left[
          \mathbf 1_{\mathrm{Near}_s}
          \operatorname{Gap}_{A^\star}(\bar x_s,\bar y_s)
          \right]
    +O(\log(eT)).
  \end{aligned}
\end{equation*}
It therefore suffices to identify the leading expected gap of
the restricted equilibrium on $\mathrm{Near}_s$.

To derive the exact leading coefficient, we refine the analysis of the expected gap of $(\bar x_s,\bar y_s)$ on $\mathsf{Near}_s$.
On $\mathsf{Near}_s$, we can express the gap of $(\bar x_s,\bar y_s)$ in terms of the sampled matrix errors.
To characterize these errors, we use the convergence of $(\bar x_s,\bar y_s)$ to $(x^\star,y^\star)$, which gives the limiting observation frequency $x_i^\star y_j^\star$ for each support entry $(i,j)$.
Consequently, the scaled errors $\sqrt{s}(\widetilde A_{I^\star,J^\star,s}-A^\star_{I^\star,J^\star})$ converge in distribution to $\mathsf W^{\rm core}$.
By using the error and equilibrium limits in the gap expression, we obtain
\[
  \sqrt{s}\,\E\!
  \left[
    \1_{\mathsf{Near}_s}
    \Gap_{A^\star}(\bar x_s,\bar y_s)\right]
  \longrightarrow\Lambda_{\rm core}.
\]
Hence, summing the expected gaps and using $\sum_{s=1}^{T-mn} s^{-1/2}=2\sqrt T+o(\sqrt T)$ proves Theorem~\ref{thm:exact-core-asymptotics}:
\[
  \E[\mathcal R_T]=2\Lambda_{\rm core}\sqrt T+o(\sqrt T).
\]

\begin{remark}[Technical difficulty compared to single-player stochastic bandits]
  \normalfont
  In a stochastic bandit, selecting a suboptimal arm provides an observation that helps correct its estimate.
  In contrast, the corresponding bad event $\mathsf D_s(b_T)$ in our two-player setting does not ensure an observation of an insufficiently observed payoff entry.
  To recover this link, we compare the bad-event probability with the probability that only one outside action appears slightly profitable while the other outside actions remain unprofitable.
  On this alternative event, the equilibrium conditions force the learner to play the profitable outside action and actions in the opponent's equilibrium support with probabilities bounded below, ensuring a chance to observe an insufficiently observed entry.
\end{remark}

\section{Cumulative Gap Lower Bounds}
\label{sec:lower-bounds}

In this section, we establish lower bounds on the expected cumulative duality gap.
These bounds match the square-root and logarithmic dependence on $T$ in Section~\ref{sec:upper-bounds}, respectively, in the minimax sense for fully mixed games and on fixed separable instances with a pure equilibrium.

\subsection{Minimax Lower Bound}
We first consider games with a unique fully mixed Nash equilibrium.
Fix a payoff radius $R>0$ and suppose that $m=n\ge4$.
Consider any possibly randomized, nonanticipating policy $\pi$ for the learning protocol in Section~\ref{sec:setting}.
The policy may depend on $(m,R,\sigma,T)$ but not on the unknown payoff matrix.
We write $\E_A^\pi$ for expectation under policy $\pi$ and payoff matrix $A$, taken over the observation noise and all randomization in the learning protocol.

\begin{theorem}
  \label{thm:unique-minimax-lower}
  For every policy $\pi$ described above and every $T\ge1$, there exists a matrix $A_{\pi,T}\in[-R,R]^{m\times m}$ with a unique fully mixed Nash equilibrium such that
  \begin{equation*}
    \E_{A_{\pi,T}}^\pi[\mathcal R_T(A_{\pi,T})]
    \ge
    \frac1{8192}
    \min\left\{
    \frac{RT}{m},
    \sigma\sqrt{mT\log m}
    \right\}.
  \end{equation*}
\end{theorem}

For fixed $(m,R,\sigma)$, Theorem~\ref{thm:unique-minimax-lower} shows that an $\Omega(\sqrt T)$ cumulative duality gap is unavoidable in the worst case as $T\to\infty$.
This matches the dependence on $T$ in Theorem~\ref{thm:unique-rate} for fully mixed equilibria.
In the proof, we relate the duality gap to the expected maximum of Gaussian row perturbations.
The full proof is given in Appendix~\ref{proof:sec:unique-minimax-lower}.

\subsection{Instance-Dependent Lower Bound}
We next consider a fixed separable game with a unique pure Nash equilibrium.
Fix a separable matrix $A^\star$ with $A^\star_{ij}=u_i+v_j$, where $u_1>u_i$ for every $i>1$ and $v_1<v_j$ for every $j>1$.
Then, $(e_1,e_1)$ is the unique Nash equilibrium.

Consider any online policy $\pi$ after the one-observation-per-cell initialization of Nash-TS (Algorithm~\ref{alg:ts-ne}).
We call a policy \emph{uniformly efficient} on a family of games if, for each fixed payoff matrix $A$ in the family, its expected cumulative duality gap satisfies $\E_A^\pi[\mathcal R_T(A)]=o(T^\alpha)$ as $T\to\infty$ for every $\alpha>0$.
Suppose that the policy is uniformly efficient on $A^\star$ and, for every sufficiently small $\varepsilon>0$, on each separable alternative obtained by replacing $u_i$ with $u_1+\varepsilon$ for one $i>1$ or replacing $v_j$ with $v_1-\varepsilon$ for one $j>1$.

\begin{theorem}
  \label{thm:pure-separable-lower}
  Every policy $\pi$ described above satisfies
  \begin{equation*}
    \begin{aligned}
      \!\!\liminf_{T\to\infty}
      \frac{\E_{A^\star}^\pi[\mathcal R_T(A^\star)]}{\log T}
      \!\ge\!
      2\sigma^2\left(\!
      \sum_{i=2}^m\frac1{u_1-u_i}
      \!+\!\sum_{j=2}^n\frac1{v_j-v_1}
      \!\!\right)\!\!.
    \end{aligned}
  \end{equation*}
\end{theorem}

Theorem~\ref{thm:unique-rate} ensures that Nash-TS satisfies the efficiency condition when $m,n\ge2$, since the base and alternative games have unique pure Nash equilibria.
Thus, Theorem~\ref{thm:pure-separable-lower} matches its $O(\log T)$ upper bound in the dependence on $T$ over this separable family.

To prove Theorem~\ref{thm:pure-separable-lower}, we construct alternative games in which one suboptimal row or column becomes uniquely optimal.
Uniform efficiency then requires logarithmically many plays of every suboptimal action by a Gaussian change-of-measure argument \citep{lai1985asymptotically}.
The full proof is given in Appendix~\ref{proof:sec:pure-separable-lower-proof}.

For separable games with a pure Nash equilibrium, \citet{ito2025instance} prove a related lower bound on the sum of the players' external regrets under Bernoulli feedback.
Our result gives an explicit asymptotic lower bound for cumulative duality gap under Gaussian feedback.

\section{Experiments}
\label{sec:experiments}

In this section, we empirically evaluate Nash-TS on games with a unique pure equilibrium, a unique fully mixed equilibrium, a unique mixed equilibrium on the boundary, and multiple equilibria.
Although our analysis assumes a unique equilibrium, we include the game with multiple equilibria to examine the behavior of Nash-TS beyond this assumption.
We also use random payoff games with $25$ and $50$ actions per player, where each payoff entry is drawn independently and uniformly from $[-1,1]$.

\begin{figure*}[t]
  \centering
  \includegraphics[width=\textwidth]{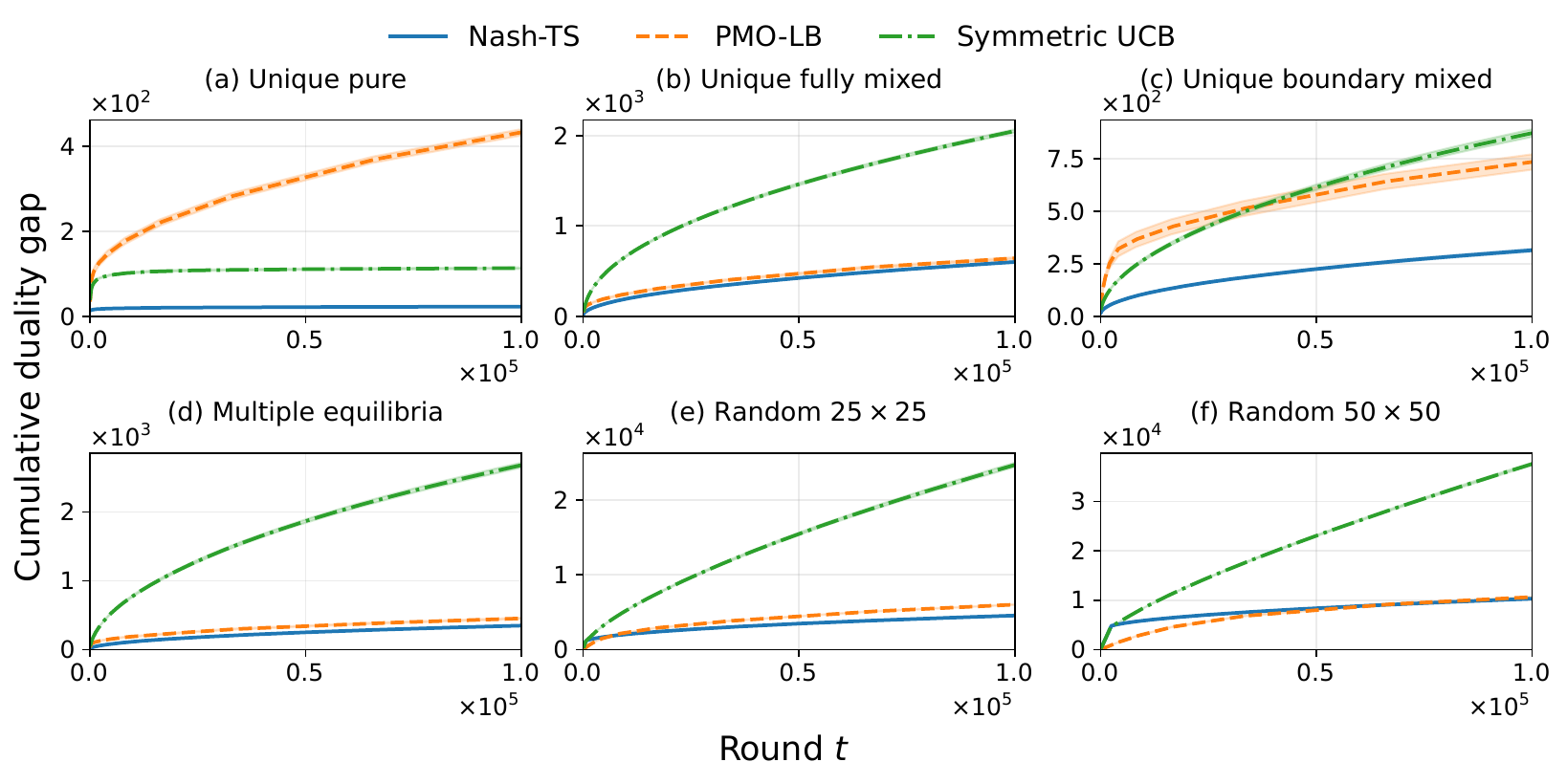}
  \caption{Cumulative duality gap under Gaussian bandit feedback with $\sigma=1/5$.
    Each curve is the mean over $100$ trials, and the shaded areas represent standard errors.}
  \label{fig:regret}
\end{figure*}

\begin{figure*}[t]
  \centering
  \includegraphics[width=\textwidth]{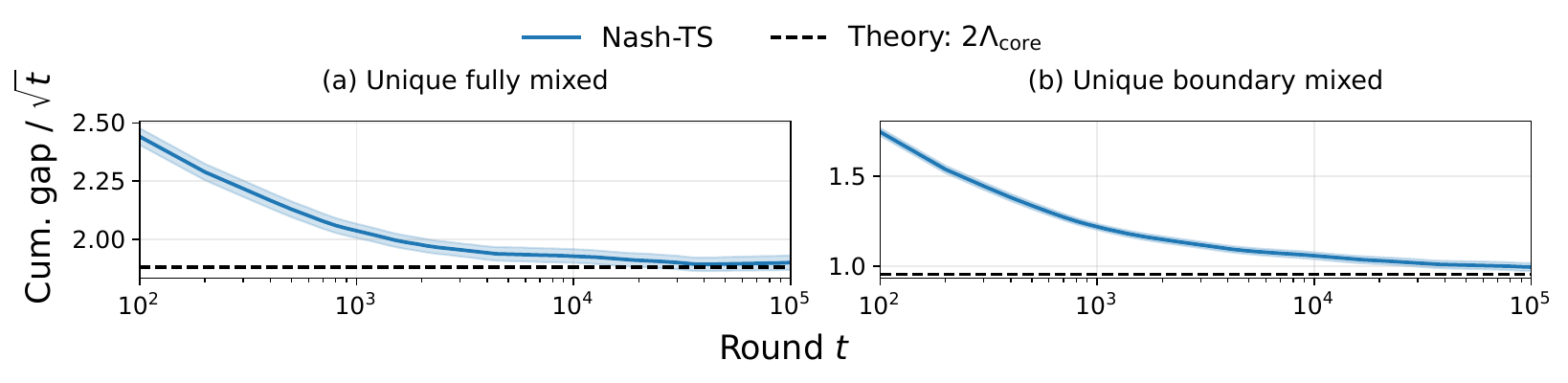}
  \caption{Cumulative duality gap of Nash-TS divided by $\sqrt t$.
    Each curve is the mean over $100$ trials, and the shaded areas represent standard errors.
    Dashed lines show $2\Lambda_{\rm core}$ from Theorem~\ref{thm:exact-core-asymptotics}.}
  \label{fig:scaled-regret}
\end{figure*}

We compare Nash-TS with PMO-LB \citep{hait2026near} and the matrix-game UCB of \citet{odonoghue2021matrix} applied to both players, which we call symmetric UCB.
\footnote{\citet{yang2025incentivize} also analyze cumulative duality gap under bandit feedback, but the strategy pairs whose gaps are evaluated are not those actually played.
  Their regret is therefore not directly comparable to cumulative duality gap under our definition.
  Appendix~\ref{app:experimental-details} reports an additional comparison in which their cumulative gap is computed over the strategy pairs evaluated in their analysis.}
All runs use $T=10^5$ payoff observations with Gaussian noise of standard deviation $\sigma=1/5$.
We evaluate the cumulative duality gap in \eqref{eq:regret}, including initialization.
For each game, we average the results over $100$ trials.
The payoff matrices, implementation details, and hyperparameter tuning for PMO-LB are provided in Appendix~\ref{app:experimental-details}.

Figure~\ref{fig:regret} illustrates the cumulative duality gap for the three methods.
Nash-TS attains the smallest final mean gap on all six benchmarks.
On the games with unique pure and boundary mixed equilibria, it performs substantially better than PMO-LB.
In contrast, the two methods perform similarly on the game with a unique fully mixed equilibrium.
The results on the game with multiple equilibria suggest that Nash-TS can also perform well when the uniqueness assumption does not hold.
We also observe favorable performance in the larger random games with $25$ and $50$ actions per player.

As an additional experiment, Figure~\ref{fig:scaled-regret} shows the cumulative gap of Nash-TS divided by $\sqrt t$ on the two mixed-equilibrium games.
Both curves approach the coefficient predicted by Theorem~\ref{thm:exact-core-asymptotics}.
These results suggest that the asymptotic expression in Theorem~\ref{thm:exact-core-asymptotics} provides a useful approximation to the expected cumulative duality gap even at finite horizons.

\section{Conclusion}
\label{sec:conclusion}

This paper analyzes Nash-TS for equilibrium learning in unknown two-player zero-sum matrix games under Gaussian bandit feedback.
For every fixed game with a unique Nash equilibrium, we establish an $O(\sqrt T)$ bound on the expected cumulative duality gap, which improves to $O(\log T)$ when the equilibrium is pure.
We further identify the exact coefficient of the $\sqrt T$ term, which depends only on the equilibrium strategies and the noise level.
Our lower bounds show that the dependence on $T$ is optimal in the minimax sense for fully mixed games and on fixed separable instances with a pure equilibrium.
Experiments support the predicted rates and coefficient.
Extending the analysis to games with multiple equilibria and narrowing the gap between the finite-horizon upper and lower bounds are important directions for future work.

\bibliographystyle{icml2026}
\bibliography{references}

\clearpage
\onecolumn
\appendix
\section{Equilibrium geometry under uniqueness}
\label{proof:sec:equilibrium-geometry}
Uniqueness ensures that nearby support submatrices have fully mixed equilibria and strictly negative outside margins.
We first record these properties and their consequences for the duality gap.
We then give conditions for extending these equilibria to the full game.

\begin{lemma}
  \label{proof:lem:local-geometry-from-uniqueness}
  Suppose that $A^\star$ has a unique Nash equilibrium $(x^\star,y^\star)$.
  Then, $|I^\star|=|J^\star|$.
  Moreover, there exist constants $\rho,p>0$ and $\Delta_{\rm row},\Delta_{\rm col}\in(0,+\infty]$ with the following properties.
  Every $C\in\R^{d\times d}$ with $\normmax{C-A^\star_{I^\star,J^\star}}\le\rho$ has a unique Nash equilibrium $(x^C,y^C)$, this equilibrium is fully mixed, and the map $C\mapsto(x^C,y^C)$ is continuous on the closed $\rho$-ball around $A^\star_{I^\star,J^\star}$.
  We index the coordinates of $x^C$ by $I^\star$ and those of $y^C$ by $J^\star$, and we write $v^C:=(x^C)^\trans Cy^C$ for the value of $C$.
  For every such $C$,
  \[
    \min\Bigl\{\min_{i'\in I^\star}x_{i'}^C,\ \min_{j'\in J^\star}y_{j'}^C\Bigr\}\ge p.
  \]
  The outside margins satisfy
  \begin{align*}
    v^C-\sum_{j\in J^\star}
    A^\star_{ij}y_j^C
     & \ge\Delta_{\rm row} &  & (i\notin I^\star), \\
    \sum_{i\in I^\star}x_i^CA^\star_{ij}-v^C
     & \ge\Delta_{\rm col} &  & (j\notin J^\star).
  \end{align*}
  Here, $\Delta_{\rm row}=+\infty$ if $I^\star=[m]$, $\Delta_{\rm col}=+\infty$ if $J^\star=[n]$, and
  \begin{equation*}
    \rho\le\frac{\Delta_{\rm row}}2,
    \qquad
    \rho\le\frac{\Delta_{\rm col}}2.
  \end{equation*}
\end{lemma}

\begin{lemma}
  \label{proof:lem:core-gap-bound}
  With the constants of Lemma~\ref{proof:lem:local-geometry-from-uniqueness}, let $C\in\R^{d\times d}$ satisfy $\normmax{C-A^\star_{I^\star,J^\star}}\le\rho$.
  Let $(\bar x^C,\bar y^C)\in\simplex_m\times\simplex_n$ denote its unique equilibrium $(x^C,y^C)$ extended by zeros outside $I^\star$ and $J^\star$.
  Then,
  \[
    \Gap_{A^\star}(\bar x^C,\bar y^C)
    =\max_{j\in J^\star}(x^C)^\trans(C-A^\star_{I^\star,J^\star})e_j
    -\min_{i\in I^\star}e_i^\trans(C-A^\star_{I^\star,J^\star})y^C.
  \]
  Consequently, there is $\chi_{\rm core}\in[0,2]$ such that every such $C$ satisfies
  \begin{equation*}
    \Gap_{A^\star}(\bar x^C,\bar y^C)
    \le
    \chi_{\rm core}\normmax{C-A^\star_{I^\star,J^\star}}.
  \end{equation*}
  Moreover, the probabilities $x_i^Cy_j^C$ give
  \[
    \Gap_{A^\star}(\bar x^C,\bar y^C)
    \le\min\{\chi_{\rm core}/p^2,2/p\}
    \sum_{(i,j)\in I^\star\times J^\star}x_i^Cy_j^C|C_{ij}-A^\star_{ij}|.
  \]
  Moreover, $\chi_{\rm core}=0$ when $d=1$.
\end{lemma}

\begin{lemma}
  \label{proof:lem:gradient-margins}
  With the constants of Lemma~\ref{proof:lem:local-geometry-from-uniqueness}, there are $\kappa_x,\kappa_y>0$ such that every $C$ with $\normmax{C-A^\star_{I^\star,J^\star}}\le\rho$ and every $u\in\R^d$ with $\1^\trans u=0$ satisfy
  \begin{equation*}
    \max_{i'\in I^\star}(Cu)_{i'}\ge\kappa_y\normone{u},\qquad
    \max_{j'\in J^\star}(-u^\trans C)_{j'}\ge\kappa_x\normone{u}.
  \end{equation*}
\end{lemma}

The next statement concerns an arbitrary restricted game.
It distinguishes equilibrium validity from uniqueness when an outside margin is zero.

\begin{lemma}
  \label{proof:lem:margin-embedding}
  Let $B\in\R^{m\times n}$, let $\varnothing\ne I\subseteq[m]$ and $\varnothing\ne J\subseteq[n]$, and suppose that $C=B_{I,J}$ has a unique Nash equilibrium $(x^C,y^C)$ with value $v^C$.
  Let $(\bar x^C,\bar y^C)$ be its zero extension, and define
  \[
    U_i(B)=e_i^\trans B\bar y^C-v^C\quad(i\notin I),\qquad
    V_j(B)=v^C-(\bar x^C)^\trans Be_j\quad(j\notin J).
  \]
  Then,
  \[
    (\bar x^C,\bar y^C)\in\NE(B)
    \quad\Longleftrightarrow\quad
    U_i(B)\le0\ (i\notin I),\quad V_j(B)\le0\ (j\notin J).
  \]
  If all existing outside margins are strictly negative, then
  \[
    \NE(B)=\{(\bar x^C,\bar y^C)\}.
  \]
  Conditions over empty outside index sets are vacuous.
  Nonpositive margins alone do not imply uniqueness when an outside margin is zero.
\end{lemma}

The next statement requires accuracy only on entries with at least one index in the equilibrium supports.
No bound is required on entries whose two indices both lie outside the equilibrium supports.

\begin{lemma}
  \label{proof:lem:relevant-equilibrium-embedding}
  With the constants of Lemma~\ref{proof:lem:local-geometry-from-uniqueness}, let $B\in\R^{m\times n}$ satisfy $|B_{ij}-A^\star_{ij}|\le\rho$ for all $(i,j)\in([m]\times J^\star)\cup(I^\star\times[n])$.
  Set $C=B_{I^\star,J^\star}$.
  Let $(\bar x^C,\bar y^C)\in\simplex_m\times\simplex_n$ denote its unique equilibrium $(x^C,y^C)$ extended by zeros outside $I^\star$ and $J^\star$.
  Then, $B$ has the unique equilibrium $(\bar x^C,\bar y^C)$.
  Moreover, $\bar x_i^C\bar y_j^C\ge p^2$ for every $(i,j)\in I^\star\times J^\star$.
\end{lemma}

\subsection{Proof of Lemma~\ref{proof:lem:local-geometry-from-uniqueness}}
\begin{proof}
  We use the following consequence of \citet[Lemma~2]{wu2024minimally}.
  Suppose that the two-player zero-sum game $A^\star$ has a unique Nash equilibrium $(x^\star,y^\star)$.
  Let $I^\star=\operatorname{supp}(x^\star)$, $J^\star=\operatorname{supp}(y^\star)$, and $v^\star=(x^\star)^\trans A^\star y^\star$.
  Then,
  \begin{align*}
    \sum_{j\in J^\star}A^\star_{ij}y_j^\star  & <v^\star\quad(i\notin I^\star), \\
    \sum_{i\in I^\star}x_i^\star A^\star_{ij} & >v^\star\quad(j\notin J^\star),
  \end{align*}
  and the matrix
  \[
    \begin{pmatrix}
      A^\star_{I^\star,J^\star} & -\1 \\
      \1^\trans                 & 0
    \end{pmatrix}
  \]
  is invertible.

  Its numbers of rows and columns are equal, so $|I^\star|=|J^\star|=:d$.
  Transposing this matrix and multiplying its last row and column by $-1$ gives
  \[
    \begin{pmatrix}
      (A^\star_{I^\star,J^\star})^\trans & -\1 \\
      \1^\trans                          & 0
    \end{pmatrix},
  \]
  which is also invertible.
  By continuity of the determinants, both matrices remain invertible for $\normmax{C-A^\star_{I^\star,J^\star}}\le\rho_0$ with some $\rho_0>0$.
  Consider the linear systems for $(y^C,v_y^C)$ and $(x^C,v_x^C)$:
  \[
    \begin{aligned}
      \sum_{j\in J^\star}
      C_{ij}y_j^C                     & =v_y^C\quad(i\in I^\star),
                                      & \sum_{j\in J^\star}y_j^C   & =1, \\
      \sum_{i\in I^\star}x_i^C C_{ij} & =v_x^C\quad(j\in J^\star),
                                      & \sum_{i\in I^\star}x_i^C   & =1.
    \end{aligned}
  \]
  Equivalently,
  \begin{equation}
    \begin{aligned}
      \begin{pmatrix}C&-\1\\\1^\trans&0\end{pmatrix}
      \begin{pmatrix}y^C\\v_y^C\end{pmatrix}
       & =\begin{pmatrix}0\\1\end{pmatrix}, \\
      \begin{pmatrix}C^\trans&-\1\\\1^\trans&0\end{pmatrix}
      \begin{pmatrix}x^C\\v_x^C\end{pmatrix}
       & =\begin{pmatrix}0\\1\end{pmatrix}.
    \end{aligned}
    \label{proof:eq:core-equilibrium-systems}
  \end{equation}
  Since both matrices are invertible, the two systems have unique solutions:
  \begin{align*}
    \begin{pmatrix}y^C\\v_y^C\end{pmatrix}
     & =\begin{pmatrix}C&-\1\\\1^\trans&0\end{pmatrix}^{-1}
    \begin{pmatrix}0\\1\end{pmatrix},                              \\
    \begin{pmatrix}x^C\\v_x^C\end{pmatrix}
     & =\begin{pmatrix}C^\trans&-\1\\\1^\trans&0\end{pmatrix}^{-1}
    \begin{pmatrix}0\\1\end{pmatrix}.
  \end{align*}
  By \eqref{proof:eq:core-equilibrium-systems}, the two scalar components agree:
  \begin{align*}
    v_y^C & =(x^C)^\trans(v_y^C\1)
    =(x^C)^\trans Cy^C
    =(v_x^C\1)^\trans y^C=v_x^C=:v^C.
  \end{align*}
  The inverses, and hence the solutions, are continuous in $C$.
  At $C=A^\star_{I^\star,J^\star}$, the original equilibrium satisfies these systems, so uniqueness gives $x^C=(x^\star)_{I^\star}>0$ and $y^C=(y^\star)_{J^\star}>0$.
  By continuity and the strict inequalities from the cited result, we can choose $0<\rho\le\rho_0$ such that every $C$ with $\normmax{C-A^\star_{I^\star,J^\star}}\le\rho$ satisfies
  \begin{align*}
    x_i^C & >0\quad(i\in I^\star),
    \qquad y_j^C>0\quad(j\in J^\star),                                               \\
    v^C-\sum_{j\in J^\star}
    A^\star_{ij}y_j^C
          & \ge\frac12\left(v^\star-\sum_{j\in J^\star}A^\star_{ij}y_j^\star\right)
    \ge2\rho\quad(i\notin I^\star),                                                  \\
    \sum_{i\in I^\star}x_i^C A^\star_{ij}-v^C
          & \ge\frac12\left(\sum_{i\in I^\star}x_i^\star A^\star_{ij}-v^\star\right)
    \ge2\rho\quad(j\notin J^\star).
  \end{align*}
  Together with $\1^\trans x^C=\1^\trans y^C=1$, positivity gives $x^C,y^C\in\operatorname{ri}\simplex_d$.
  Moreover, \eqref{proof:eq:core-equilibrium-systems} gives, for every $x,y\in\simplex_d$,
  \[
    x^\trans Cy^C
    \le (x^C)^\trans Cy^C
    =v^C
    \le (x^C)^\trans Cy,
  \]
  so $(x^C,y^C)$ is a fully mixed Nash equilibrium of $C$ with value $v^C=(x^C)^\trans Cy^C$.

  To establish uniqueness, let $(x',y')$ be any Nash equilibrium of $C$.
  Optimality of $y'$ and the equal-payoff identity for $x^C$ give
  \begin{align*}
    v^C\1-Cy' & \ge0,                    \\
    \sum_{i\in I^\star}x_i^C\bigl(v^C-(Cy')_i\bigr)
              & =v^C-(x^C)^\trans Cy'=0.
  \end{align*}
  Nonnegativity of each term gives
  \[
    x_i^C\bigl(v^C-(Cy')_i\bigr)=0
    \quad\Longrightarrow\quad
    (Cy')_i=v^C
    \qquad(i\in I^\star),
  \]
  where the implication follows from $x_i^C>0$.
  Therefore,
  \[
    \begin{pmatrix}C&-\1\\\1^\trans&0\end{pmatrix}
    \begin{pmatrix}y'\\v^C\end{pmatrix}
    =\begin{pmatrix}0\\1\end{pmatrix}.
  \]
  Thus, $(y',v^C)$ solves the first system in \eqref{proof:eq:core-equilibrium-systems}, giving $y'=y^C$ by uniqueness.
  The second system similarly gives $x'=x^C$.

  If $I^\star\ne[m]$, define
  \[
    \Delta_{\rm row}
    =\min_{\substack{\normmax{C-A^\star_{I^\star,J^\star}}\le\rho\\i\notin I^\star}}
    \left\{v^C-\sum_{j\in J^\star}A^\star_{ij}y_j^C\right\}
    \ge2\rho>0.
  \]
  Similarly, if $J^\star\ne[n]$, define
  \[
    \Delta_{\rm col}
    =\min_{\substack{\normmax{C-A^\star_{I^\star,J^\star}}\le\rho\\j\notin J^\star}}
    \left\{\sum_{i\in I^\star}x_i^C A^\star_{ij}-v^C\right\}
    \ge2\rho>0,
  \]
  where compactness and continuity ensure that the minima exist.
  Set $\Delta_{\rm row}=+\infty$ if $I^\star=[m]$ and $\Delta_{\rm col}=+\infty$ if $J^\star=[n]$.
  Thus, $\rho\le\Delta_{\rm row}/2$ and $\rho\le\Delta_{\rm col}/2$.
  Compactness, positivity and continuity also give
  \[
    p=\min_{\normmax{C-A^\star_{I^\star,J^\star}}\le\rho}
    \min\Bigl\{\min_{i\in I^\star}x_i^C,\min_{j\in J^\star}y_j^C\Bigr\}>0.
  \]
\end{proof}

\subsection{Proof of Lemma~\ref{proof:lem:core-gap-bound}}
\begin{proof}
  Fix the constants of Lemma~\ref{proof:lem:local-geometry-from-uniqueness} and let $C$ satisfy $\normmax{C-A^\star_{I^\star,J^\star}}\le\rho$.
  Its fully mixed equilibrium satisfies $Cy^C=v^C\1$ and $(x^C)^\trans C=v^C\1^\trans$.
  With $F=C-A^\star_{I^\star,J^\star}$, these identities give
  \begin{equation}
    e_i^\trans A^\star\bar y^C=v^C-e_i^\trans Fy^C\ (i\in I^\star),
    \qquad
    (\bar x^C)^\trans A^\star e_j=v^C-(x^C)^\trans Fe_j\ (j\in J^\star).
    \label{proof:eq:core-gap-identities}
  \end{equation}
  Since $y^C$ is a probability vector and $\normmax F\le\rho$, \eqref{proof:eq:core-gap-identities} gives $e_i^\trans A^\star\bar y^C\ge v^C-\rho$ for $i\in I^\star$.
  For $i\notin I^\star$, the row margin in Lemma~\ref{proof:lem:local-geometry-from-uniqueness} ($\Delta_{\rm row}\ge2\rho$) gives $e_i^\trans A^\star\bar y^C=\sum_{j\in J^\star}A^\star_{ij}y_j^C\le v^C-\Delta_{\rm row}\le v^C-2\rho$.
  The analogous bounds for columns follow by symmetry.
  Combining these bounds, we obtain
  \begin{equation}
    \max_{i\notin I^\star}e_i^\trans A^\star\bar y^C\le v^C-2\rho\le\min_{i\in I^\star}e_i^\trans A^\star\bar y^C,
    \qquad
    \min_{j\notin J^\star}(\bar x^C)^\trans A^\star e_j\ge v^C+2\rho\ge\max_{j\in J^\star}(\bar x^C)^\trans A^\star e_j,
    \label{proof:eq:core-gap-localization}
  \end{equation}
  where a bound over an empty outside index set ($I^\star=[m]$ or $J^\star=[n]$) holds trivially.
  By \eqref{proof:eq:core-gap-localization}, the maximizing row and minimizing column lie in the core, so the duality gap reduces to
  \[
    \Gap_{A^\star}(\bar x^C,\bar y^C)
    =\max_{i\in I^\star}e_i^\trans A^\star\bar y^C-\min_{j\in J^\star}(\bar x^C)^\trans A^\star e_j.
  \]
  Substituting \eqref{proof:eq:core-gap-identities} cancels $v^C$ and leaves
  \[
    \Gap_{A^\star}(\bar x^C,\bar y^C)
    =\max_{j\in J^\star}(x^C)^\trans Fe_j-\min_{i\in I^\star}e_i^\trans Fy^C
    \le2\normmax F.
  \]
  This establishes $\chi_{\rm core}=2$ for $d\ge2$.
  For $d=1$, $x^C=y^C=1$ and the two terms equal the single entry of $F$, so they cancel and $\chi_{\rm core}=0$.

  It remains to express these bounds using the probabilities $x_i^Cy_j^C$.
  The inequalities $x_i^C,y_j^C\ge p$ hold by Lemma~\ref{proof:lem:local-geometry-from-uniqueness}.
  Applying the maximum-norm bound just established and using $x_i^Cy_j^C\ge p^2$, we obtain
  \begin{equation}
    \begin{aligned}
      \Gap_{A^\star}(\bar x^C,\bar y^C)
       & \le\chi_{\rm core}\max_{(i,j)\in I^\star\times J^\star}|F_{ij}| \\
       & \le\frac{\chi_{\rm core}}{p^2}
      \max_{(i,j)\in I^\star\times J^\star}x_i^Cy_j^C|F_{ij}|
      \le\frac{\chi_{\rm core}}{p^2}
      \sum_{(i,j)\in I^\star\times J^\star}x_i^Cy_j^C|F_{ij}|.
    \end{aligned}
    \label{proof:eq:geometry-core-action-uniform-bound}
  \end{equation}

  For the second bound, the triangle inequality and $y_j^C\ge p$ give, for every $j\in J^\star$,
  \begin{align*}
    (x^C)^\trans Fe_j
     & \le\sum_{i\in I^\star} x_i^C|F_{ij}|
    \le\frac1p\sum_{i\in I^\star} x_i^Cy_j^C|F_{ij}|
    \le\frac1p\sum_{(i',j')\in I^\star\times J^\star}x_{i'}^Cy_{j'}^C|F_{i'j'}|.
  \end{align*}
  Similarly, using $x_i^C\ge p$, we obtain, for every $i\in I^\star$,
  \begin{align*}
    -e_i^\trans Fy^C
     & \le\sum_{j\in J^\star} y_j^C|F_{ij}|
    \le\frac1p\sum_{j\in J^\star} x_i^Cy_j^C|F_{ij}|
    \le\frac1p\sum_{(i',j')\in I^\star\times J^\star}x_{i'}^Cy_{j'}^C|F_{i'j'}|.
  \end{align*}
  Substituting these bounds into the gap identity above, we obtain
  \begin{equation}
    \begin{aligned}
      \Gap_{A^\star}(\bar x^C,\bar y^C)
       & =\max_{j\in J^\star}(x^C)^\trans Fe_j
      +\max_{i\in I^\star}(-e_i^\trans Fy^C)
      \le\frac2p\sum_{(i,j)\in I^\star\times J^\star}x_i^Cy_j^C|F_{ij}|.
    \end{aligned}
    \label{proof:eq:geometry-core-action-weighted-bound}
  \end{equation}
  Combining \eqref{proof:eq:geometry-core-action-uniform-bound} and~\eqref{proof:eq:geometry-core-action-weighted-bound} gives the claimed coefficient $\min\{\chi_{\rm core}/p^2,2/p\}$.
\end{proof}

\subsection{Proof of Lemma~\ref{proof:lem:gradient-margins}}
\begin{proof}
  Fix the constants of Lemma~\ref{proof:lem:local-geometry-from-uniqueness} and a $C$ with $\normmax{C-A^\star_{I^\star,J^\star}}\le\rho$.
  Its unique fully mixed equilibrium satisfies $(x^C)^\trans C=v^C\1^\trans$ and $Cy^C=v^C\1$.
  For $d\ge2$, let $\mathcal U=\{u:\1^\trans u=0,\normone{u}=1\}$ and fix $u\in\mathcal U$.
  Suppose $\max_{i'\in I^\star}(Cu)_{i'}\le0$.
  Then, $(x^C)^\trans Cu=v^C\1^\trans u=0$.
  Since $x^C>0$ and every $(Cu)_{i'}\le0$, we obtain $Cu=0$.
  For small $\varepsilon>0$, the vector $y^C+\varepsilon u$ lies in $\simplex_d$ and satisfies $C(y^C+\varepsilon u)=v^C\1$.
  Together with $(x^C)^\trans C=v^C\1^\trans$ and $Cu=0$, this gives, for every $x,y\in\simplex_d$,
  \[
    x^\trans C(y^C+\varepsilon u)
    = x^\trans Cy^C
    \le (x^C)^\trans Cy^C
    = (x^C)^\trans C(y^C+\varepsilon u)
    =v^C
    \le (x^C)^\trans Cy.
  \]
  Hence, $(x^C,y^C+\varepsilon u)$ is a Nash equilibrium of $C$ distinct from $(x^C,y^C)$, contradicting uniqueness.
  Therefore, $\max_{i'\in I^\star}(Cu)_{i'}>0$, and the transposed argument with $x^C+\varepsilon u$ gives $\max_{j'\in J^\star}(-u^\trans C)_{j'}>0$.
  Hence, the compact minima
  \[
    \kappa_y=\min_{\substack{\normmax{C-A^\star_{I^\star,J^\star}}\le\rho\\u\in\mathcal U}}\max_{i'\in I^\star}(Cu)_{i'}>0,\qquad
    \kappa_x=\min_{\substack{\normmax{C-A^\star_{I^\star,J^\star}}\le\rho\\u\in\mathcal U}}\max_{j'\in J^\star}(-u^\trans C)_{j'}>0
  \]
  are positive.
  Homogeneity extends both inequalities to all zero-sum $u$.
  For $d=1$, the only such $u$ is zero, and we take $\kappa_x=\kappa_y=1$.
\end{proof}

\subsection{Proof of Lemma~\ref{proof:lem:margin-embedding}}
\begin{proof}
  The restricted equilibrium inequalities give
  \[
    e_i^\trans B\bar y^C\le v^C\quad(i\in I),\qquad
    (\bar x^C)^\trans Be_j\ge v^C\quad(j\in J).
  \]
  Thus, the zero extension satisfies the full-game equilibrium inequalities exactly when all outside margins are nonpositive.
  Its value is then $v^C$.
  Suppose that all existing outside margins are strictly negative, and let $(x,y)\in\NE(B)$.
  Optimality in the full game and the restricted equilibrium inequalities imply
  \[
    \begin{aligned}
      v^C\le x^\trans B\bar y^C
       & \le v^C+\sum_{i\notin I}x_iU_i(B)\le v^C, \\
      v^C\ge(\bar x^C)^\trans By
       & \ge v^C-\sum_{j\notin J}y_jV_j(B)\ge v^C.
    \end{aligned}
  \]
  Strict negativity therefore forces $x_i=0$ for $i\notin I$ and $y_j=0$ for $j\notin J$.
  The restricted pair $(x_I,y_J)$ is an equilibrium of $C$, so its uniqueness gives $(x,y)=(\bar x^C,\bar y^C)$.
  The same argument applies when an outside index set is empty.
\end{proof}

\subsection{Proof of Lemma~\ref{proof:lem:relevant-equilibrium-embedding}}
\begin{proof}
  Since $C=B_{I^\star,J^\star}$ satisfies $\normmax{C-A^\star_{I^\star,J^\star}}\le\rho$, Lemma~\ref{proof:lem:local-geometry-from-uniqueness} gives its unique fully mixed equilibrium $(x^C,y^C)$ with value $v^C$.
  On the core, $Cy^C=v^C\1$ and $(x^C)^\trans C=v^C\1^\trans$.
  For outside actions, the entrywise accuracy and the margins of that lemma give
  \begin{align*}
    e_i^\trans B\bar y^C-v^C
     & \le\sum_{j\in J^\star}
    A^\star_{ij}y_j^C+\rho-v^C
    \le-\Delta_{\rm row}+\rho<0,\quad i\notin I^\star,   \\
    v^C-(\bar x^C)^\trans Be_j
     & \le v^C-\sum_{i\in I^\star}x_i^CA^\star_{ij}+\rho
    \le-\Delta_{\rm col}+\rho<0,\quad j\notin J^\star.
  \end{align*}
  Lemma~\ref{proof:lem:margin-embedding} now gives the unique equilibrium $(\bar x^C,\bar y^C)$.
  The product bound follows from $x_i^C\ge p$ and $y_j^C\ge p$.
\end{proof}

\section{Gaussian estimates}
\label{proof:sec:gaussian-estimates}
We first bound finite Gaussian maxima and sample means, then interval probabilities and conditional expectations.
Throughout this section, $\phi$ and $\Phi$ denote the density and distribution function of $\mathcal N(0,1)$, and $u_+=\max\{u,0\}$.

The following bounds control maxima of finitely many Gaussian variables.
For independent $\mathcal N(0,1)$ variables, a lower bound of the same order also holds.
\begin{lemma}
  \label{proof:lem:gaussian-maxima}
  Let $D\ge1$ be an integer and $v>0$.
  Let $Z_1,\ldots,Z_D$ be random variables, not necessarily independent, with $Z_j\sim\mathcal N(0,v_j^2)$ and $0\le v_j\le v$ for every $j\in[D]$.
  Then,
  \[
    \E\!
    \left[\max_{1\le j\le D}
      Z_j\right]
    \le v\sqrt{2\log D},
    \qquad
    \E\!
    \left[\max_{1\le j\le D}|Z_j|^2\right]
    \le 2v^2\{\log(2D)+1\}.
  \]
  For independent $G_j\sim\mathcal N(0,1)$, $j\in[k]$, put $m_k=\E[\max_{1\le j\le k}
      G_j]$.
  Then $m_1=0$ and, for every $k\ge2$,
  \[
    \frac{3}{16}\sqrt{\log k}
    \le m_k\le\sqrt{2\log k}.
  \]
\end{lemma}

The next result bounds the absolute deviations in decreasing order, without requiring independence.
\begin{lemma}
  \label{proof:lem:gaussian-order-statistics}
  Let $M\ge1$ be an integer and let $G_1,\ldots,G_M$ be random variables, not necessarily independent, with $G_j\sim\mathcal N(0,1)$ for every $j\in[M]$.
  Let $W_{(1)}\ge\cdots\ge W_{(M)}$ be the decreasing
  rearrangement of $|G_1|,\ldots,|G_M|$.

  For $1\le j\le M$ and $u\ge0$,
  \[
    \Pp(W_{(j)}>u)
    \le\min\left\{1,\frac{2M}{j}e^{-u^2/2}\right\},
    \qquad
    \E[W_{(j)}]\le\sqrt{2\log(2M/j)}+1.
  \]
  Consequently, for every integer $1\le L\le M$,
  \[
    \E\!\left[\sum_{j=1}^{L}\frac{W_{(j)}}{\sqrt j}\right]
    \le10\sqrt{L\log(eM/L)}.
  \]
\end{lemma}

If deterministic shifts sum to zero, Gaussian variation in a subset of coordinates still gives a lower bound on the expected maximum.
\begin{lemma}
  \label{proof:lem:main-shifted-gaussian-maximum}
  Let $m\ge2$ be an integer, let $\theta\in\R^m$ satisfy
  $\sum_{i=1}^m\theta_i=0$, and let $Z_1,\ldots,Z_m$
  be independent with $Z_i\sim\mathcal N(0,s_i^2)$.

  Suppose that $S\subseteq[m]$, $1\le|S|<m$, and
  $0\le s\le s_i$ for every $i\in S$.
  With $m_k$ defined in Lemma~\ref{proof:lem:gaussian-maxima},
  \[
    \E\!
    \left[\max_{i\in[m]}(\theta_i+Z_i)\right]
    \ge\frac{|S|}{m}\,s\,m_{|S|}.
  \]
\end{lemma}

The following estimate controls sample means simultaneously over all sample sizes beyond a given minimum.
\begin{lemma}
  \label{proof:lem:complete-table-envelope}
  Let $D\ge1$ be an integer.
  Let $(\xi_{a,r})_{r\ge1}$, $a\in[D]$, be sequences, not necessarily independent of one another, with $\xi_{a,r}\stackrel{\mathrm{i.i.d.}}{\sim}\mathcal N(0,1)$ over $r$ for each $a$.
  Set $S_{a,N}=\sum_{r=1}^N\xi_{a,r}$.

  For every integer $k\ge1$ and $u\ge0$,
  \[
    \begin{gathered}
      \Pp\!\left(\max_{a\in[D]}\sup_{N\ge k}\frac{|S_{a,N}|}{N}>u\right)
      \le\min\{1,2D\exp(-ku^2/2)\},\\
      \E\!\left[\max_{a\in[D]}\sup_{N\ge k}\frac{|S_{a,N}|}{N}\right]
      \le2\sqrt{2\log(2D)/k},
      \qquad
      k\E\!\left[\max_{a\in[D]}\sup_{N\ge k}\frac{|S_{a,N}|^2}{N^2}\right]
      \le2\{\log(2D)+1\}.
    \end{gathered}
  \]
  For every integer $r\ge1$, there is a finite constant $C_r$,
  depending only on $r$, such that
  \[
    \E\!\left[\left(1+\max_{a\in[D]}\sup_{N\ge1}\frac{|S_{a,N}|}{N}\right)^r\right]
    \le C_r(1+\log(2D))^{r/2}
    \qquad(D\ge1).
  \]
\end{lemma}

We next control how far Gaussian partial sums exceed a linear boundary.
The exponential moment is needed when taking expectations of reciprocal interval probabilities.
\begin{lemma}
  \label{proof:lem:gaussian-drift-moments}
  Let $\xi_r\sim\mathcal N(0,1)$, $r\ge1$, be independent, and let $S_N=\sum_{r=1}^N\xi_r$.
  For every $\alpha>0$ and $u\ge0$,
  \[
    \Pp\!
    \left(\sup_{N\ge1}\frac{(|S_N|-\alpha N)_+}{\sqrt N}>u\right)
    \le\frac{2\exp(-u^2/2-\alpha u)}{\exp(\alpha^2/2)-1}.
  \]
  Moreover, for every $\alpha>0$ and nonnegative integer $r$,
  \[
    \E\!\left[
      \left(1+\sup_{N\ge1}\frac{(|S_N|-\alpha N)_+}{\sqrt N}\right)^r
      \exp\!\left\{\frac12\sup_{N\ge1}\frac{(|S_N|-\alpha N)_+^2}{N}\right\}
      \right]<\infty.
  \]
\end{lemma}

Data-dependent sample sizes have the usual Gaussian limit when their ratios converge to positive deterministic limits.
\begin{lemma}
  \label{proof:lem:gaussian-random-index-limit}
  Let $D\ge1$ be an integer.
  Let $\xi_{a,r}\sim\mathcal N(0,1)$, $a\in[D]$, $r\ge1$, be independent.
  Set $S_{a,N}=\sum_{r=1}^N\xi_{a,r}$.

  For each integer $t\ge1$, let $N_a(t)$ be positive
  integer-valued random variables such that,
  for deterministic $\nu_a>0$,
  \[
    N_a(t)/t\longrightarrow\nu_a
    \quad\text{in probability},\qquad a\in[D].
  \]
  Independence between the indices and the partial sums
  is not required.
  Then,
  \[
    \left(\sqrt t\,\frac{S_{a,N_a(t)}}{N_a(t)}\right)_{a\in[D]}
    \Rightarrow
    \mathcal N\!
    \left(
    0,\operatorname{diag}(\nu_a^{-1})_{a\in[D]}
    \right).
  \]

  For each $t$, let
  $G_t=(G_{a,t})_{a\in[D]}\sim\mathcal N(0,I_D)$.
  Assume that $G_t$ is independent of the joint collection
  of all $\xi_{a,r}$ and the current index vector
  $(N_a(t))_{a\in[D]}$.
  Then, we also have
  \[
    \left[
      \sqrt t\left\{
      \frac{S_{a,N_a(t)}}{N_a(t)}
      +\frac{G_{a,t}}{\sqrt{N_a(t)}}
      \right\}
      \right]_{a\in[D]}
    \Rightarrow
    \mathcal N\!
    \left(
    0,\operatorname{diag}(2\nu_a^{-1})_{a\in[D]}
    \right).
  \]
\end{lemma}

The probability that a scalar Gaussian lies in an interval can be bounded using the interval width and its distance from the mean.
When the mean is nonpositive, a positive interval also captures a fixed proportion of the positive-tail probability.
\begin{lemma}
  \label{proof:lem:gaussian-interval-probabilities}
  Let $Z\sim\mathcal N(\mu,\tau^2)$ with $\tau>0$.
  If $a<b$ and $\ell>0$ satisfy $b-a\ge\ell\tau$, then
  \[
    \begin{aligned}
      \Pp(a<Z\le b)
       & \ge\frac{\min\{1,\ell\}}{\sqrt{2\pi}(1+\operatorname{dist}(\mu,[a,b])/\tau)} \\
       & \qquad\cdot\exp\!\left\{
      -\frac{\operatorname{dist}(\mu,[a,b])^2}{2\tau^2}
      -\min\{1,\ell\}-\frac{\min\{1,\ell\}^2}{2}
      \right\}.
    \end{aligned}
  \]

  In addition, if $\mu\le0$, $0<\tau\le v$, and $L>0$, then
  \[
    \Pp(0<Z\le L)
    \ge
    \left(1-e^{-\min\{1,L/v\}^2/2}\right)\Pp(Z>0).
  \]
\end{lemma}

The next result bounds the size of a Gaussian vector when a positive weighted average is conditioned to lie in an interval.
\begin{lemma}
  \label{proof:lem:window-regression}
  Let
  $\mathbf Z\sim
    \mathcal N(\mu,\operatorname{diag}(v_1^2,\ldots,v_d^2))$,
  where $\sigma>0$ and $0<v_j\le\sigma$ for every $j$.
  Suppose that $w\in\simplex_d$, $w_j\ge p>0$ for every $j$,
  and $|h|\le M$ with $M\ge0$.

  For every $L>0$,
  \[
    \E[\|\mathbf Z\|_\infty\mid0<w^\trans\mathbf Z-h\le L]
    \le
    \frac{M+L}{p}
    +(1+p^{-1})\|\mu\|_\infty
    +\sigma\sqrt{2\log(2d)}.
  \]
\end{lemma}

For a fixed realization of a Gaussian sequence, we now consider interval probabilities over all sample sizes.
The minimum probability stays positive, and its reciprocal remains integrable even after multiplication by the sum of failure probabilities.
\begin{lemma}
  \label{proof:lem:all-count-localization}
  Let $\xi_r\sim\mathcal N(0,1)$, $r\ge1$, be independent, put $S_N=\sum_{r=1}^N\xi_r$, and let $G\sim\mathcal N(0,1)$ be independent of this sequence.
  In the definition below, probability is taken over $G$, with $(\xi_r)_{r\ge1}$ held fixed.

  For $c>0$, define
  \[
    q_N(c)=\Pp\!
    \left(
    \left|\frac{S_N}{N}+\frac{G}{\sqrt N}\right|\le c
    \right).
  \]
  Then, almost surely,
  \[
    \inf_{N\ge1}q_N(c)>0,
    \qquad\sum_{N\ge1}(1-q_N(c))<\infty.
  \]
  Moreover,
  \[
    \E\!\left[\frac{1+\sum_{N\ge1}(1-q_N(c))}{\inf_{N\ge1}q_N(c)}\right]<\infty.
  \]
\end{lemma}

Finally, suppose that a weighted mean lies below a threshold by a fixed amount.
If the conditional Gaussian distribution does not yet put high probability below that threshold, both the minimum sample size and the probability of a small positive deviation can be controlled uniformly.
\begin{lemma}
  \label{proof:lem:gaussian-negative-mean-interval}
  Fix an integer $d\ge1$, $\mu\in\R^d$, $\sigma,\Delta>0$,
  $0<\varepsilon<1/2$, and $0<L\le\Delta/2$.
  Let $\xi_{j,r}\sim\mathcal N(0,1)$, $j\in[d]$, $r\ge1$, be independent, and put $S_{j,N}=\sum_{r=1}^N\xi_{j,r}$.

  Define
  \[
    H_\Delta=\left\{
    \sum_{j=1}^d\left(
    \sup_{N\ge1}\frac{(|S_{j,N}|-3\Delta N/(8\sigma))_+}{\sqrt N}
    \right)^2
    \right\}^{1/2}.
  \]
  Let $G_j\sim\mathcal N(0,1)$, $j\in[d]$, be independent of each other and of all $\xi_{j,r}$.
  For $\mathbf n\in\{1,2,\ldots\}^d$, define
  \[
    \mathbf Z_{\mathbf n}
    =\left(
    \mu_j+\frac{\sigma S_{j,n_j}}{n_j}
    +\frac{\sigma G_j}{\sqrt{n_j}}
    \right)_{j=1}^d.
  \]

  All probabilities involving $\mathbf Z_{\mathbf n}$ below are taken over $G_1,\ldots,G_d$, with $(\xi_{j,r})_{j\in[d],\,r\ge1}$ held fixed.
  There exists $c_0>0$, depending only on $(\sigma,L,\varepsilon)$, such that the following holds almost surely, simultaneously for every $\mathbf n\in\{1,2,\ldots\}^d$ and every $(w,h)\in\simplex_d\times\R$ satisfying $w^\trans\mu-h\le-\Delta$.
  If $\Pp(w^\trans\mathbf Z_{\mathbf n}-h\le-\Delta/4)<1-\varepsilon$, then
  \[
    \begin{gathered}
      \frac{3\Delta}{8\sigma}\sqrt{\min_j n_j}<H_\Delta+\Phi^{-1}(1-\varepsilon),\\
      \Pp(0<w^\trans\mathbf Z_{\mathbf n}-h\le L)
      \ge\frac{c_0}{(1+H_\Delta)e^{H_\Delta^2/2}}.
    \end{gathered}
  \]
  Moreover, for every nonnegative integer $r$,
  \[
    \E\!\left[(1+H_\Delta)^r e^{H_\Delta^2/2}\right]<\infty.
  \]
\end{lemma}

\subsection{Proof of Lemma~\ref{proof:lem:gaussian-maxima}}
\begin{proof}
  For $D=1$, the expectation of $Z_1$ is zero.
  For $D\ge2$ and $\lambda>0$, the Gaussian moment generating function and Jensen's inequality give
  \[
    \exp\!
    \left\{\lambda\E\max_{1\le j\le D}
    Z_j\right\}
    \le\E\exp\!
    \left\{\lambda\max_{1\le j\le D}
    Z_j\right\}
    \le\sum_{j=1}^D\E e^{\lambda Z_j}
    \le D e^{\lambda^2v^2/2}.
  \]
  Taking logarithms and setting $\lambda=\sqrt{2\log D}/v$ proves the mean bound.
  The Gaussian tail bound and a union bound yield, for $u\ge0$,
  \[
    \Pp\!
    \left(\max_{1\le j\le D}|Z_j|>v\sqrt u\right)
    \le\min\{1,2D e^{-u/2}\}.
  \]
  Integrating gives
  \[
    \E\!\left[\max_{1\le j\le D}|Z_j|^2\right]
    \le v^2\left[2\log(2D)+\int_{2\log(2D)}^\infty2D e^{-u/2}\,du\right]
    =2v^2\{\log(2D)+1\}.
  \]
  Applying the mean bound to the independent variables $G_j\sim\mathcal N(0,1)$ gives $m_k\le\sqrt{2\log k}$, and $m_1=0$.
  Integration by parts gives the Gaussian tail lower bound
  \[
    1-\Phi(x)=\frac{\phi(x)}x-\int_x^\infty\frac{\phi(u)}{u^2}\,du
    \ge\frac{x\phi(x)}{1+x^2}\qquad(x>0).
  \]
  For $k\ge16$, it implies
  \[
    k\bigl[1-\Phi(\sqrt{\log k}/2)\bigr]
    \ge\frac{2k^{7/8}\sqrt{\log k}}{(\log k+4)\sqrt{2\pi}}\ge1.
  \]
  The last expression is increasing: its logarithmic derivative with respect to $\log k$ is $7/8+1/(2\log k)-1/(\log k+4)>0$, and its value at $k=16$ is $8\sqrt{\log2}/[(1+\log2)\sqrt\pi]\ge\sqrt2$.
  Integrating the nonnegative maximum-tail integrand over $[0,\sqrt{\log k}/2]$ now gives
  \[
    \begin{aligned}
      m_k & =\int_0^\infty[1-\Phi(x)^k-(1-\Phi(x))^k]\,dx \\
          & \ge\tfrac12\sqrt{\log k}\,(1-e^{-1}-2^{-k})
      >\tfrac3{16}\sqrt{\log k}.
    \end{aligned}
  \]
  For $2\le k<16$, monotonicity gives $m_k\ge m_2=1/\sqrt\pi\ge1/2>3\sqrt{\log k}/16$.
\end{proof}

\subsection{Proof of Lemma~\ref{proof:lem:gaussian-order-statistics}}
\begin{proof}
  For $1\le j\le M$ and $u\ge0$, Markov's inequality and the Gaussian tail bound give
  \[
    \begin{aligned}
      \Pp(W_{(j)}>u)
       & =\Pp\!\left(\sum_{a=1}^M\1_{\{|G_a|>u\}}\ge j\right)    \\
       & \le\min\left\{1,\frac1j\sum_{a=1}^M\Pp(|G_a|>u)\right\}
      \le\min\left\{1,\frac{2M}{j}e^{-u^2/2}\right\}.
    \end{aligned}
  \]
  Integrating this bound and using $\int_v^\infty e^{-u^2/2}\,du\le v^{-1}e^{-v^2/2}$ for $v>0$ yields
  \[
    \begin{aligned}
      \E[W_{(j)}]
       & \le\sqrt{2\log(2M/j)}
      +\frac{2M}{j}\int_{\sqrt{2\log(2M/j)}}^\infty e^{-u^2/2}\,du \\
       & \le\sqrt{2\log(2M/j)}+\frac1{\sqrt{2\log(2M/j)}}
      \le\sqrt{2\log(2M/j)}+1.
    \end{aligned}
  \]
  Here, $\sqrt{2\log(2M/j)}\ge\sqrt{2\log2}>1$.
  Consequently, for $1\le L\le M$, we have
  \[
    \begin{aligned}
      \E\!\left[\sum_{j=1}^L\frac{W_{(j)}}{\sqrt j}\right]
       & \le\sum_{j=1}^L\frac{\sqrt{2\log(2M/j)}+1}{\sqrt j} \\
       & \le(\sqrt{2\log(2M/L)}+1)\sum_{j=1}^Lj^{-1/2}
      +\sqrt2\sum_{j=1}^L\frac{\sqrt{\log(L/j)}}{\sqrt j}    \\
       & \le2(\sqrt{2\log(2M/L)}+1)\sqrt L
      +\sqrt2\int_0^L\frac{\sqrt{\log(L/x)}}{\sqrt x}\,dx    \\
       & =(2\sqrt2\sqrt{\log(2M/L)}+2\sqrt\pi+2)\sqrt L      \\
       & \le(2\sqrt2+2\sqrt\pi+2)\sqrt{L\log(eM/L)}
      <10\sqrt{L\log(eM/L)}.
    \end{aligned}
  \]
  The second inequality uses $\sqrt{a+b}\le\sqrt a+\sqrt b$ for $a,b\ge0$.
  The third uses integral comparison, and the equality follows from
  $\int_0^Lx^{-1/2}\sqrt{\log(L/x)}\,dx=\sqrt{2\pi L}$.
\end{proof}

\subsection{Proof of Lemma~\ref{proof:lem:main-shifted-gaussian-maximum}}
\begin{proof}
  \label{proof:sec:main-shifted-gaussian-maximum-proof}
  For $i\in S$, decompose $Z_i$ into independent variables with distributions $\mathcal N(0,s^2)$ and $\mathcal N(0,s_i^2-s^2)$.
  Jensen's inequality removes the latter variables.
  Averaging all permutations of the shifts and using convexity then gives
  \[
    \E\max_{i\in[m]}(\theta_i+Z_i)
    \ge\frac1{|S|}\sum_{i\in S}\theta_i+sm_{|S|}.
  \]
  The expectation is also at least $\max_{i\notin S}\theta_i\ge-\sum_{i\in S}\theta_i/(m-|S|)$.
  Combining the two bounds with weights $|S|/m$ and $(m-|S|)/m$ cancels the shifts and proves the claim.
\end{proof}

\subsection{Proof of Lemma~\ref{proof:lem:complete-table-envelope}}
\begin{proof}
  For every integer $k\ge1$, we write $U_{D,k}=\max_{a\in[D]}\sup_{N\ge k}|S_{a,N}|/N$.
  Put $S_{a,0}=0$.
  For $u>0$, $\exp(uS_{a,N}-u^2N/2)$ is a nonnegative martingale with respect to $(\xi_{a,r})_{r\ge1}$, with initial value one.
  By \citet[Theorem~3.9]{lattimore2020bandit},
  \[
    \Pp(\exists N\ge k:S_{a,N}>uN)
    \le\Pp\left(\sup_{N\ge0}e^{uS_{a,N}-u^2N/2}\ge e^{ku^2/2}\right)
    \le e^{-ku^2/2}.
  \]
  A union bound over both signs and all $D$ sequences proves the tail bound, also trivially valid at $u=0$.
  Integrating gives
  \[
    \begin{aligned}
      \E[U_{D,k}]    & \le\sqrt{\frac{2\log(2D)}k}+\frac1{\sqrt{2k\log(2D)}}
      \le2\sqrt{\frac{2\log(2D)}k},                                          \\
      \E[kU_{D,k}^2] & \le2\log(2D)+\int_{2\log(2D)}^\infty2De^{-v/2}\,dv
      =2\{\log(2D)+1\}.
    \end{aligned}
  \]
  The same tail bound gives $\Pp(U_{D,1}>\sqrt{2\log(2D)}+v)\le e^{-v^2/2}$ and hence
  \[
    \E[(U_{D,1}-\sqrt{2\log(2D)})_+^r]
    \le\int_0^\infty rv^{r-1}e^{-v^2/2}\,dv
      =2^{r/2}\Gamma(1+r/2).
  \]
  Combining $(a+b)^r\le2^{r-1}(a^r+b^r)$ and $(1+\sqrt{2\log(2D)})^2\le4(1+\log(2D))$ gives
  \[
    C_r=2^{r-1}\{2^r+2^{r/2}\Gamma(1+r/2)\},\qquad C_4=192.
  \]
\end{proof}

\subsection{Proof of Lemma~\ref{proof:lem:gaussian-drift-moments}}
\begin{proof}
  For $\alpha>0$, we define $D_\alpha=\sup_{N\ge1}(|S_N|-\alpha N)_+/\sqrt N$.
  For $u\ge0$, the Gaussian tail bound and a union bound give
  \[
    \begin{aligned}
      \Pp(D_\alpha>u)
       & \le2\sum_{N\ge1}e^{-(\alpha\sqrt N+u)^2/2}          \\
       & \le2e^{-u^2/2-\alpha u}\sum_{N\ge1}e^{-\alpha^2N/2}
      =\frac{2e^{-u^2/2-\alpha u}}{e^{\alpha^2/2}-1}.
    \end{aligned}
  \]
  Integrating this tail bound against the derivative of $(1+u)^r e^{u^2/2}$ gives
  \[
    \E[(1+D_\alpha)^r e^{D_\alpha^2/2}]
    \le1+\frac2{e^{\alpha^2/2}-1}\int_0^\infty
    [r(1+u)^{r-1}+u(1+u)^r]e^{-\alpha u}\,du.
  \]
  The first term in the integrand is zero when $r=0$.
  The integral is finite because its integrand is a polynomial times $e^{-\alpha u}$.
\end{proof}

\subsection{Proof of Lemma~\ref{proof:lem:gaussian-random-index-limit}}
\begin{proof}
  For each $a\in[D]$, put $n_a(t)=\max\{1,\lfloor\nu_a t\rfloor\}$ and fix $0<\delta<\nu_a/2$.
  In either direction from $n_a(t)$, the first $\lceil\delta t\rceil$ increments are independent and centered, with total variance $\lceil\delta t\rceil$.
  Kolmogorov's maximal inequality and a union bound give, for sufficiently large $t$ and $\varepsilon>0$,
  \[
    \Pp(|S_{a,N_a(t)}-S_{a,n_a(t)}|>\varepsilon\sqrt t)
    \le\Pp(|N_a(t)-n_a(t)|>\delta t)
    +\frac{2\lceil\delta t\rceil}{\varepsilon^2t}.
  \]
  Letting $t\to\infty$ and then $\delta\downarrow0$ gives
  $(S_{a,N_a(t)}-S_{a,n_a(t)})/\sqrt t\to0$ in probability.
  Since $N_a(t)/t\to\nu_a>0$ and $n_a(t)/t\to\nu_a$, it follows jointly over the finite set $[D]$ that
  \[
    \frac{S_{a,N_a(t)}}{\sqrt{N_a(t)}}
    -\frac{S_{a,n_a(t)}}{\sqrt{n_a(t)}}
    \longrightarrow0\quad\text{in probability}.
  \]
  For every $t$, $S_{a,n_a(t)}/\sqrt{n_a(t)}\sim\mathcal N(0,1)$ independently across $a$, because the indices are deterministic.
  Thus, Slutsky's theorem applied to
  \[
    \sqrt t\,\frac{S_{a,N_a(t)}}{N_a(t)}
    =\sqrt{t/N_a(t)}\,
    \frac{S_{a,n_a(t)}}{\sqrt{n_a(t)}}+o_{\Pp}(1)
  \]
  proves the first limit.
  For the second, the same replacement gives
  \[
    \sqrt t\left\{\frac{S_{a,N_a(t)}}{N_a(t)}
    +\frac{G_{a,t}}{\sqrt{N_a(t)}}\right\}
    =\sqrt{t/N_a(t)}
    \left[\frac{S_{a,n_a(t)}}{\sqrt{n_a(t)}}+G_{a,t}\right]+o_{\Pp}(1).
  \]
  The bracketed variables are independent $\mathcal N(0,2)$ across $a$ at every $t$.
  Another application of Slutsky's theorem proves the second limit.
\end{proof}

\subsection{Proof of Lemma~\ref{proof:lem:gaussian-interval-probabilities}}
\begin{proof}
  Standardize the interval and write $z=\operatorname{dist}(\mu,[a,b])/\tau$.
  If the standardized interval avoids zero, symmetry and the monotonicity of $\phi$ on $[0,\infty)$ give
  $\Pp(a<Z\le b)\ge\int_z^{z+\ell}\phi(v)\,dv$.
  If it contains zero, let $u_1,u_2\ge0$ be its endpoint distances.
  Then, $u_1+u_2=(b-a)/\tau\ge\ell$, $z=0$, and
  \[
    \Pp(a<Z\le b)
    =\int_0^{u_1}\phi(v)\,dv+\int_0^{u_2}\phi(v)\,dv
    \ge\int_0^{u_1+u_2}\phi(v)\,dv
    \ge\int_0^\ell\phi(v)\,dv.
  \]
  In both cases, with $\delta_\ell=\min\{1,\ell\}$, we obtain
  \[
    \begin{aligned}
      \Pp(a<Z\le b)
       & \ge\int_z^{z+\delta_\ell/(1+z)}\phi(v)\,dv                    \\
       & \ge\frac{\delta_\ell}{\sqrt{2\pi}(1+z)}
      \exp\!\left\{-\frac{z^2}{2}
      -\frac{z\delta_\ell}{1+z}-\frac{\delta_\ell^2}{2(1+z)^2}\right\} \\
       & \ge\frac{\delta_\ell}{\sqrt{2\pi}(1+z)}
      \exp\!\left\{-\frac{z^2}{2}-\delta_\ell-\frac{\delta_\ell^2}{2}\right\}.
    \end{aligned}
  \]
  For the second claim, put $z=-\mu/\tau\ge0$ and $b_0=\min\{1,L/v\}>0$.
  Since $L/\tau\ge b_0$, we have
  \[
    1-\Phi(z+L/\tau)
    \le\int_z^\infty\phi(u+b_0)\,du
    \le e^{-b_0^2/2}[1-\Phi(z)].
  \]
  Subtracting from $\Pp(Z>0)=1-\Phi(z)$ proves the claim.
\end{proof}

\subsection{Proof of Lemma~\ref{proof:lem:window-regression}}
\begin{proof}
  We write $U=w^\trans\mathbf Z-h$.
  Set $\Sigma=\operatorname{diag}(v_1^2,\ldots,v_d^2)$, $\tau^2=w^\trans\Sigma w>0$ and $W=\mathbf Z-\Sigma w(w^\trans\mathbf Z)/\tau^2$.
  Direct calculation gives
  \[
    \begin{gathered}
      \operatorname{Cov}(W,w^\trans\mathbf Z)=0,
      0\le\frac{(\Sigma w)_j}{\tau^2}\le\frac1{w_j}\le\frac1p,
      \qquad\operatorname{Var}(W_j)\le\sigma^2,\\
      \qquad\|\E W\|_\infty\le(1+p^{-1})\|\mu\|_\infty.
    \end{gathered}
  \]
  Thus, $W$ and $U$ are independent Gaussians.
  Apply Lemma~\ref{proof:lem:gaussian-maxima} to the $2d$ centered coordinates $W_j-\E W_j$ and their negatives to obtain
  \[
    \E\|W-\E W\|_\infty\le\sigma\sqrt{2\log(2d)}.
  \]
  On $0<U\le L$, we have $\|\mathbf Z\|_\infty\le\|W\|_\infty+(M+L)/p$.
  Independence and the bound on $\|\E W\|_\infty$ consequently give
  \[
    \begin{aligned}
      \E[\|\mathbf Z\|_\infty\mid0<U\le L]
       & \le\|\E W\|_\infty+\E\|W-\E W\|_\infty+\frac{M+L}{p} \\
       & \le\frac{M+L}{p}
      +(1+p^{-1})\|\mu\|_\infty+\sigma\sqrt{2\log(2d)}.
    \end{aligned}
  \]
\end{proof}

\subsection{Proof of Lemma~\ref{proof:lem:all-count-localization}}
\begin{proof}
  We fix $c>0$.
  We write $q_*(c)=\inf_{N\ge1}q_N(c)$ and $R(c)=\sum_{N\ge1}(1-q_N(c))$.
  For $\alpha>0$, we write $D_\alpha=\sup_{N\ge1}(|S_N|-\alpha N)_+/\sqrt N$.
  For each fixed realization of $(\xi_r)_{r\ge1}$, the variable $S_N/N+G/\sqrt N$ is Gaussian with mean $S_N/N$ and standard deviation $N^{-1/2}$.
  Apply Lemma~\ref{proof:lem:gaussian-interval-probabilities} to $[-c,c]$ with $\ell=2c$.
  Its standardized distance from the mean is $(|S_N|-cN)_+/\sqrt N\le D_c$, so
  \[
    \frac1{q_*(c)}
    \le\frac{\sqrt{2\pi}}{\min\{1,2c\}}
    \exp\!
    \left\{\min\{1,2c\}+\frac{\min\{1,2c\}^2}{2}\right\}
    (1+D_c)e^{D_c^2/2}.
  \]
  In particular, Lemma~\ref{proof:lem:gaussian-drift-moments} gives $\E[1/q_*(c)]<\infty$.
  For the sum of failure probabilities, the Gaussian tail bound gives
  \[
    1-q_N(c)\le\1_{\{|S_N|>cN/2\}}+e^{-c^2N/8}.
  \]
  If $|S_N|>cN/2$, then $D_{c/4}>c\sqrt N/4$.
  Consequently,
  \[
    R(c)\le\frac{16}{c^2}D_{c/4}^2+\frac1{e^{c^2/8}-1}.
  \]
  Since $D_c\le D_{c/4}$ and $u^2(1+u)\le(1+u)^3$ for $u\ge0$, these bounds give
  \[
    \begin{aligned}
      \E\!\left[\frac{R(c)}{q_*(c)}\right]
       & \le\frac{\sqrt{2\pi}}{\min\{1,2c\}}
      \exp\!\left\{\min\{1,2c\}+\frac{\min\{1,2c\}^2}{2}\right\}    \\
       & \qquad\cdot\left(\frac{16}{c^2}+\frac1{e^{c^2/8}-1}\right)
      \E[(1+D_{c/4})^3 e^{D_{c/4}^2/2}]<\infty.
    \end{aligned}
  \]
  Finiteness follows from Lemma~\ref{proof:lem:gaussian-drift-moments} with $(\alpha,r)=(c/4,3)$.
  Combining the two expectations proves the stated bound.
  It also implies $q_*(c)>0$ and $R(c)<\infty$ almost surely, since $q_*(c)\le1$.
\end{proof}

\subsection{Proof of Lemma~\ref{proof:lem:gaussian-negative-mean-interval}}
\begin{proof}
  Set $\alpha=3\Delta/(8\sigma)$ and write $D_{\alpha,j}=\sup_{N\ge1}(|S_{j,N}|-\alpha N)_+/\sqrt N$ and $z_\varepsilon=\Phi^{-1}(1-\varepsilon)$.
  We fix a realization of $(\xi_{j,r})_{j\in[d],\,r\ge1}$ for which every $D_{\alpha,j}$ is finite.
  We also fix an arbitrary index vector $\mathbf n$ and $(w,h)$ satisfying $w^\trans\mu-h\le-\Delta$.
  With these sequences held fixed, $w^\trans\mathbf Z_{\mathbf n}-h$ is Gaussian with mean and variance
  \[
    \theta=w^\trans\mu-h+\sigma\sum_{j=1}^d w_j S_{j,n_j}/n_j,
    \qquad
    \tau^2=\sigma^2\sum_{j=1}^d w_j^2/n_j>0.
  \]
  By the definition of $D_{\alpha,j}$ and Cauchy--Schwarz,
  \[
    \theta\le-\Delta+\sigma\alpha
    +\sigma\sum_{j=1}^d\frac{w_jD_{\alpha,j}}{\sqrt{n_j}}
    \le-5\Delta/8+\tau H_\Delta,
    \qquad
    \tau^2\le\frac{\sigma^2}{\min_j n_j}.
  \]
  If $\Pp(w^\trans\mathbf Z_{\mathbf n}-h\le-\Delta/4)<1-\varepsilon$, then
  $\theta>-\Delta/4-\tau z_\varepsilon$, so
  \[
    \alpha\sqrt{\min_j n_j}
    \le\frac{3\Delta}{8\tau}<H_\Delta+z_\varepsilon.
  \]
  The standardized distance from $\theta$ to $[0,L]$ therefore satisfies
  \[
    \frac{\operatorname{dist}(\theta,[0,L])}{\tau}
    \le\begin{cases}
      2H_\Delta/3+5z_\varepsilon/3, & \theta<0,   \\
      H_\Delta,                     & \theta\ge0.
    \end{cases}
  \]
  For the second case, use $\theta\le-5\Delta/8+\tau H_\Delta$.
  For the first, use $-\theta/\tau<\Delta/(4\tau)+z_\varepsilon$ and the preceding bound on $3\Delta/(8\tau)$.
  The identity
  \[
    (2u/3+5v/3)^2-u^2=-5(u-2v)^2/9+5v^2
  \]
  shows that the squared distance is at most $H_\Delta^2+5z_\varepsilon^2$.
  Its linear factor is at most
  $(1+5z_\varepsilon/3)(1+H_\Delta)$.
  Apply Lemma~\ref{proof:lem:gaussian-interval-probabilities} with $\ell=L/(2\sigma)$, noting that $\tau\le\sigma$.
  With $\delta_\ell=\min\{1,\ell\}$, this gives the desired lower bound with
  \[
    c_0=\frac{\delta_\ell}{\sqrt{2\pi}(1+5z_\varepsilon/3)}
    \exp\!
    \left\{-\delta_\ell-\frac{\delta_\ell^2}{2}
    -\frac{5z_\varepsilon^2}{2}\right\}>0.
  \]
  These bounds hold for each fixed realization of the sequences and therefore simultaneously for all the stated index vectors and weights.
  For the moment claim, we use $1+(\sum_j h_j^2)^{1/2}\le\prod_j(1+h_j)$ for $h_j\ge0$ and independence across the sequences to obtain
  \[
    \E[(1+H_\Delta)^r e^{H_\Delta^2/2}]
    \le\prod_{j=1}^d\E[(1+D_{\alpha,j})^r e^{D_{\alpha,j}^2/2}]<\infty.
  \]
  Lemma~\ref{proof:lem:gaussian-drift-moments} proves finiteness of each factor.
\end{proof}

\section{Proofs for the Finite-Horizon Bound}
\label{proof:sec:proof-unique-rate}
\subsection{Notation}
Let $\xi_{a,r}\sim\mathcal N(0,1)$, $a\in[m]\times[n]$, $r\ge1$, be independent.
For $r,N,t\ge1$, including initialization at $r=1$, set
\[
  \begin{gathered}
    Y_{a,r}=A_a^\star+\sigma\xi_{a,r},\qquad
    S_{a,N}=\sum_{r=1}^N\xi_{a,r},\\
    N_a(t)=1+\sum_{r=1}^{t-1}\1_{\{(I_r,J_r)=a\}},\qquad
    \widehat A_{a,N}=\frac1N\sum_{r=1}^N Y_{a,r}.
  \end{gathered}
\]

Throughout this appendix, fix the unique-equilibrium instance of Theorem~\ref{thm:unique-rate}.
Write $I^\star=\supp(x^\star)$, $J^\star=\supp(y^\star)$, and $v^\star=(x^\star)^\trans A^\star y^\star$.
Fix the constants of Lemmas~\ref{proof:lem:local-geometry-from-uniqueness}, \ref{proof:lem:core-gap-bound}, and~\ref{proof:lem:gradient-margins}.
Define
\[
  \mathcal L_\rho=\{C:\normmax{C-A^\star_{I^\star,J^\star}}\le\rho\}.
\]
For $C\in\mathcal L_\rho$, let $(x^C,y^C)$ be its unique equilibrium and $(\bar x^C,\bar y^C)$ its zero extension.
Set
\begin{gather*}
  C_t=(\widetilde A_t)_{I^\star,J^\star},\qquad
  \mathsf{Near}_t=\{C_t\in\mathcal L_\rho\},\qquad
  \mathsf{Supp}_t=\mathsf{Near}_t\cap\{(x_t,y_t)=(\bar x^{C_t},\bar y^{C_t})\}.
\end{gather*}
An indicator multiplied by a quantity defined only on its event denotes zero on the complement.
For a matrix $B$ with $B_{I^\star,J^\star}\in\mathcal L_\rho$, define the outside-action margins by
\begin{align*}
  U_i(B) & =\sum_{j\in J^\star}B_{ij}y_j^{B_{I^\star,J^\star}}-v^{B_{I^\star,J^\star}}\quad(i\notin I^\star), \\
  V_j(B) & =v^{B_{I^\star,J^\star}}-\sum_{i\in I^\star}x_i^{B_{I^\star,J^\star}}B_{ij}\quad(j\notin J^\star).
\end{align*}
For an integer $k\ge1$, define
\begin{align*}
  \mathsf{D}_t(k)=\mathsf{Near}_t\cap\mathsf{Supp}_t^c
   & \cap\bigcap_{i\notin I^\star}\bigl(\{U_i(\widetilde A_t)\le0\}\cup\{\min_{j\in J^\star}
  N_{ij}(t)<k\}\bigr)                                                                                            \\
   & \cap\bigcap_{j\notin J^\star}\bigl(\{V_j(\widetilde A_t)\le0\}\cup\{\min_{i\in I^\star}N_{ij}(t)<k\}\bigr).
\end{align*}
Empty intersections impose no condition.

For an integer $k\ge1$, define
\[
  \mathcal V_t^\times(k)
  =\{(i,j)\in
  (I^{\star c}\times J^\star)\cup(I^\star\times J^{\star c}):N_{ij}(t)<k\},
\]
where the complements are relative to $[m]$ and $[n]$.
Abbreviate the fixed-core conditional probability and expectation by
\begin{align*}
  \Pp_t^C(\cdot) & =\Pp(\cdot\mid(\xi_{a,r})_{a\in[m]\times[n],\,r\ge1},(I_r,J_r)_{r<t},C_t=C), \\
  \E_t^C[\cdot]  & =\E[\cdot\mid(\xi_{a,r})_{a\in[m]\times[n],\,r\ge1},(I_r,J_r)_{r<t},C_t=C],
\end{align*}
and write $\Pp_t^{C_t}$ for $\Pp_t^C$ evaluated at the realized core $C_t$.

\subsection{Proof of Theorem~\ref{thm:unique-rate}}
\begin{proof}
  \label{proof:sec:notation}
  Let $H=T-mn$.
  The initial decomposition is
  \[
    \mathcal R_T=\mathcal R_{\rm init}+\sum_{t=1}^H\Gap_{A^\star}(x_t,y_t),\qquad
    0\le\mathcal R_{\rm init}\le2\normmax{A^\star} mn.
  \]
  Since $\mathsf{Supp}_t\subseteq\mathsf{Near}_t$ and $\Gap_{A^\star}\le2\normmax{A^\star}$, the partition of online rounds gives
  \begin{align*}
    \mathcal R_T
     & \le 2\normmax{A^\star} mn+2\normmax{A^\star}\sum_{t=1}^H\1_{\mathsf{Near}_t^c} \\
     & \quad+\sum_{t=1}^H\1_{\mathsf{Supp}_t}\Gap_{A^\star}(x_t,y_t)
    +2\normmax{A^\star}\sum_{t=1}^H\1_{\mathsf{Near}_t\cap\mathsf{Supp}_t^c}.
  \end{align*}
  Taking expectations, we obtain
  \begin{equation}
    \begin{aligned}
      \E[\mathcal R_T]
       & \le2\normmax{A^\star} mn+2\normmax{A^\star}\sum_{t=1}^H\Pp(\mathsf{Near}_t^c) \\
       & \quad+\E\left[\sum_{t=1}^H\1_{\mathsf{Supp}_t}\Gap_{A^\star}(x_t,y_t)\right]
      +2\normmax{A^\star}\sum_{t=1}^H\Pp(\mathsf{Near}_t\cap\mathsf{Supp}_t^c).
    \end{aligned}
    \label{proof:eq:expected-regret-decomposition}
  \end{equation}
  We bound the three sums in \eqref{proof:eq:expected-regret-decomposition} in order, starting with the nonlocal rounds.
  \begin{lemma}
    \label{proof:lem:general-support-nonlocal-occupation}
    There is a finite constant $B_{\rm loc}$, depending only on $(A^\star,\sigma,m,n)$, such that, for every $H\ge1$,
    \[
      \sum_{t=1}^H\Pp(\mathsf{Near}_t^c)\le B_{\rm loc}.
    \]
  \end{lemma}
  Lemma~\ref{proof:lem:general-support-nonlocal-occupation} applied to the first sum in \eqref{proof:eq:expected-regret-decomposition} gives
  \begin{equation}
    2\normmax{A^\star}\sum_{t=1}^H\Pp(\mathsf{Near}_t^c)\le2\normmax{A^\star}
    B_{\rm loc}.
    \label{proof:eq:regret-term-nonlocal}
  \end{equation}

  For the second sum in \eqref{proof:eq:expected-regret-decomposition}, we convert the core-gap bound into posterior errors on observed entries using the following conditional law.
  \begin{lemma}
    \label{proof:lem:conditional-posterior-law}
    For every $t\ge1$, the sampled entries are conditionally independent given $(\xi_{a,r})_{a\in[m]\times[n],\,r\ge1},(I_r,J_r)_{r<t}$, with the following conditional distributions almost surely:
    \[
      \widetilde A_{a,t}\mid(\xi_{a,r})_{a\in[m]\times[n],\,r\ge1},(I_r,J_r)_{r<t}
      \sim\mathcal N\!
      \left(\widehat A_{a,N_a(t)},\frac{\sigma^2}{N_a(t)}\right),
      \qquad a\in[m]\times[n].
    \]
    Under the fixed equilibrium selection rule, for every $(i,j)\in[m]\times[n]$,
    \[
      \Pp((I_t,J_t)=(i,j)\mid(\xi_{a,r})_{a\in[m]\times[n],\,r\ge1},(I_r,J_r)_{r<t},\widetilde A_t)
      =x_{t,i}y_{t,j}.
    \]
  \end{lemma}

  Combining Lemmas~\ref{proof:lem:conditional-posterior-law} and~\ref{proof:lem:core-gap-bound}, and using $(x_t,y_t)=(\bar x^{C_t},\bar y^{C_t})$ on $\mathsf{Supp}_t$, we obtain
  \begin{equation}
    \begin{aligned}
      \1_{\mathsf{Supp}_t}\Gap_{A^\star}(x_t,y_t)
       & \le\min\{\chi_{\rm core}/p^2,2/p\}\1_{\mathsf{Supp}_t}
      \sum_{(i,j)\in I^\star\times J^\star}x_i^{C_t}y_j^{C_t}|\widetilde A_{ij,t}-A_{ij}^\star| \\
       & \le\min\{\chi_{\rm core}/p^2,2/p\}\sum_{(i,j)\in I^\star\times J^\star}
      \E\!\left[\1_{\{(I_t,J_t)=(i,j)\}}
            |\widetilde A_{ij,t}-A_{ij}^\star|\,\middle|\,\widetilde A_t\right].
    \end{aligned}
    \label{proof:eq:observed-core-gap}
  \end{equation}
  The following lemma bounds expected cumulative posterior error on observed entries.
  \begin{lemma}
    \label{proof:lem:selected-posterior-deviation}
    Let $\mathcal A\subseteq[m]\times[n]$ be a fixed set of $D\ge1$ entries.
    For every $a\in\mathcal A$ and every online round $t$, let $\omega_{a,t}\in\{0,1\}$ be a random variable with $\omega_{a,t}\le\1_{\{(I_t,J_t)=a\}}$.
    Then, for every $H\ge1$,
    \begin{equation*}
      \E\!\left[
        \sum_{a\in\mathcal A}\sum_{t=1}^H
        \omega_{a,t}|\widetilde A_{a,t}-A_a^\star|
        \right]
      \le
      20\sigma\sqrt{DH\log(eD)}.
    \end{equation*}
  \end{lemma}

  Taking expectations in \eqref{proof:eq:observed-core-gap} and summing over $t$, apply Lemma~\ref{proof:lem:selected-posterior-deviation} with $\mathcal A=I^\star\times J^\star$, $D=d^2$, and $\omega_{(i,j),t}=\1_{\{(I_t,J_t)=(i,j)\}}$ to obtain
  \begin{align}
     & \E\left[\sum_{t=1}^H\1_{\mathsf{Supp}_t}\Gap_{A^\star}(x_t,y_t)\right] \notag                                       \\
     & \quad\le\min\{\chi_{\rm core}/p^2,2/p\}\E\left[\sum_{(i,j)\in I^\star\times J^\star}\sum_{t=1}^H
                                                  \1_{\{(I_t,J_t)=(i,j)\}}|\widetilde A_{ij,t}-A_{ij}^\star|\right] \notag \\
     & \quad\le20\sigma d\min\{\chi_{\rm core}/p^2,2/p\}\sqrt{T\log(ed^2)}.
    \label{proof:eq:regret-term-marked}
  \end{align}
  For the third sum in \eqref{proof:eq:expected-regret-decomposition}, we bound $\sum_{t=1}^H\Pp(\mathsf{Near}_t\cap\mathsf{Supp}_t^c)$.
  Let $\cS_{\rm rel}=([m]\times J^\star)\cup(I^\star\times[n])$.
  For $0<\delta<1$, define
  \[
    \mathcal E_T(\delta)=\bigcap_{a\in\cS_{\rm rel}}\bigcap_{t=1}^H
    \left\{|\widetilde A_{a,t}-A_a^\star|
    \le2\sigma\sqrt{\frac{2\log(4d(m+n-d)T/\delta)}{N_a(t)}}\right\}.
  \]
  Decomposing according to whether $\mathcal E_T(\delta)$ holds, we obtain
  \begin{equation}
    \sum_{t=1}^H\Pp(\mathsf{Near}_t\cap\mathsf{Supp}_t^c)
    =\sum_{t=1}^H\Pp(\mathsf{Near}_t\cap\mathsf{Supp}_t^c\cap\mathcal E_T(\delta))
    +\sum_{t=1}^H\Pp(\mathsf{Near}_t\cap\mathsf{Supp}_t^c\cap\mathcal E_T(\delta)^c).
    \label{proof:eq:outside-confidence-decomposition}
  \end{equation}
  The following lemma bounds the failure probability and shows that positive outside margins require entries with low counts.
  \begin{lemma}
    \label{proof:lem:finite-horizon-confidence}
    For $T>mn$, $H=T-mn$, and $0<\delta<1$, define
    \[
      n_T^\star(\delta)=\max\{1,\lceil(8\sigma^2/\rho^2)\log(4d(m+n-d)T/\delta)\rceil\}.
    \]
    For the event $\mathcal E_T(\delta)$ defined above, we have
    \[
      \Pp(\mathcal E_T(\delta)^c)\le\delta,\qquad
      \mathcal E_T(\delta)\cap\mathsf{Near}_t\cap\mathsf{Supp}_t^c
      \subseteq\mathsf{D}_t(n_T^\star(\delta)).
    \]
    When $\delta=T^{-1}$,
    \begin{equation*}
      n_T^\star(T^{-1})-1\le\frac{8\sigma^2\{2+\log(4d(m+n-d))\}}{\rho^2}\log(eT).
    \end{equation*}
  \end{lemma}
  Apply Lemma~\ref{proof:lem:finite-horizon-confidence} with $\delta=T^{-1}$ and set $n_T^\star=n_T^\star(T^{-1})$.
  For the second term in \eqref{proof:eq:outside-confidence-decomposition}, the failure-probability bound gives
  \begin{equation}
    \sum_{t=1}^H\Pp(\mathsf{Near}_t\cap\mathsf{Supp}_t^c\cap\mathcal E_T(T^{-1})^c)
    \le H\Pp(\mathcal E_T(T^{-1})^c)\le1,
    \label{proof:eq:outside-confidence-failure-bound}
  \end{equation}
  where the last inequality follows from $H\le T$ and $\Pp(\mathcal E_T(T^{-1})^c)\le T^{-1}$.
  For the first term in \eqref{proof:eq:outside-confidence-decomposition}, use the inclusion in Lemma~\ref{proof:lem:finite-horizon-confidence} and the following occupation bound.
  \begin{lemma}
    \label{proof:lem:outside-action-occupation}
    Suppose that at least one player has actions outside its equilibrium support, that is, $m>d$ or $n>d$.
    There exist finite nonnegative constants $C_{\rm occ}$ and $B_{\rm occ}$, depending only on $(A^\star,\sigma,m,n)$, such that, for every $H,k\ge1$,
    \[
      \sum_{t=1}^H\Pp(\mathsf{D}_t(k))
      \le d(m+n-2d)(k-1)C_{\rm occ}+B_{\rm occ}.
    \]
  \end{lemma}
  If neither player has actions outside its equilibrium support, then $\mathsf{Supp}_t=\mathsf{Near}_t$ almost surely, and we set $C_{\rm occ}=B_{\rm occ}=0$.
  Otherwise, let $C_{\rm occ},B_{\rm occ}$ be the constants in Lemma~\ref{proof:lem:outside-action-occupation}.
  Combining Lemmas~\ref{proof:lem:finite-horizon-confidence} and~\ref{proof:lem:outside-action-occupation} with $k=n_T^\star$, we obtain
  \begin{equation}
    \begin{aligned}
       & \sum_{t=1}^H\Pp(\mathsf{Near}_t\cap\mathsf{Supp}_t^c\cap\mathcal E_T(T^{-1}))
      \le\sum_{t=1}^H\Pp(\mathsf{D}_t(n_T^\star)) \le B_{\rm occ}+d(m+n-2d)(n_T^\star-1)C_{\rm occ} \\
       & \le B_{\rm occ}
      +\frac{8\sigma^2\{2+\log(4d(m+n-d))\}}{\rho^2}d(m+n-2d)C_{\rm occ}\log(eT).
    \end{aligned}
    \label{proof:eq:outside-confidence-success-bound}
  \end{equation}
  Substituting \eqref{proof:eq:outside-confidence-success-bound} and \eqref{proof:eq:outside-confidence-failure-bound} into \eqref{proof:eq:outside-confidence-decomposition} with $\delta=T^{-1}$ and multiplying by $2\normmax{A^\star}$, we obtain
  \begin{equation}
    \begin{aligned}
       & 2\normmax{A^\star}\sum_{t=1}^H\Pp(\mathsf{Near}_t\cap\mathsf{Supp}_t^c) \\
       & \quad\le\frac{16\normmax{A^\star}\sigma^2}{\rho^2}
      \{2+\log(4d(m+n-d))\}d(m+n-2d)C_{\rm occ}\log(eT)
      \\ &\quad+2\normmax{A^\star}(1+B_{\rm occ}).
    \end{aligned}
    \label{proof:eq:regret-term-switch}
  \end{equation}
  For full support, this bound also holds, since the left-hand side is zero.

  Substituting \eqref{proof:eq:regret-term-nonlocal}, \eqref{proof:eq:regret-term-marked}, and \eqref{proof:eq:regret-term-switch} into \eqref{proof:eq:expected-regret-decomposition} gives
  \[
    \E[\mathcal R_T]\le C_{\rm mix}\sqrt T+C_{\rm out}\log(eT)+C_0,
  \]
  where
  \begin{equation}
    C_{\rm mix}=20\sigma d\min\{\chi_{\rm core}/p^2,2/p\}\sqrt{\log(ed^2)}.
    \label{proof:eq:mix-constant}
  \end{equation}
  \begin{equation}
    C_{\rm out}=\frac{16\normmax{A^\star}\sigma^2}{\rho^2}
    \{2+\log(4d(m+n-d))\}d(m+n-2d)C_{\rm occ}.
    \label{proof:eq:out-constant}
  \end{equation}
  \begin{equation}
    C_0=2\normmax{A^\star}(mn+B_{\rm loc})
    +2\normmax{A^\star}(1+B_{\rm occ}).
    \label{proof:eq:additive-constant}
  \end{equation}
  All three constants are finite and depend only on $(A^\star,\sigma,m,n)$.
  For a pure equilibrium, $d=1$ and Lemma~\ref{proof:lem:core-gap-bound} gives $\chi_{\rm core}=0$, so $C_{\rm mix}=0$.
  For a fully mixed equilibrium, $m=n=d$ and $C_{\rm out}=0$.
\end{proof}

\subsection{Proof of Lemma~\ref{proof:lem:general-support-nonlocal-occupation}}
\begin{proof}
  For each entry $a\in[m]\times[n]$ and $c>0$, define
  \[
    q_{a,N}(c)=\Pp\!
    \left(\left|\frac{S_{a,N}}
    N+\frac G{\sqrt N}\right|\le c\,\middle|\,(\xi_{b,r})_{b\in[m]\times[n],\,r\ge1}\right),
  \]
  where $G\sim\mathcal N(0,1)$ is independent of $(\xi_{b,r})_{b\in[m]\times[n],\,r\ge1}$.
  By the definition of $\mathcal L_\rho$,
  \[
    \mathsf{Near}_t^c=\bigcup_{a\in I^\star\times J^\star}
    \{|\widetilde A_{a,t}-A_a^\star|>\rho\}.
  \]
  Lemma~\ref{proof:lem:conditional-posterior-law} and \eqref{proof:eq:posterior-error-decomposition} give
  \[
    \Pp(|\widetilde A_{a,t}-A_a^\star|>\rho
    \mid(\xi_{b,r})_{b\in[m]\times[n],\,r\ge1},(I_r,J_r)_{r<t})
    =1-q_{a,N_a(t)}(\rho/\sigma).
  \]
  Hence, the union bound gives
  \[
    \Pp(\mathsf{Near}_t^c\mid(\xi_{b,r})_{b\in[m]\times[n],\,r\ge1},(I_r,J_r)_{r<t})
    \le\sum_{a\in I^\star\times J^\star}(1-q_{a,N_a(t)}(\rho/\sigma)).
  \]
  For each entry $a$, partitioning by the observation count gives
  \[
    \sum_{t=1}^H(1-q_{a,N_a(t)}(\rho/\sigma))
    =\sum_{N=1}^{\infty}(1-q_{a,N}(\rho/\sigma))
    \sum_{t=1}^H\1_{\{N_a(t)=N\}}.
  \]
  Each $q_{a,N}(\rho/\sigma)$ is a function of $(\xi_{b,r})_{b\in[m]\times[n],\,r\ge1}$, and all summands are nonnegative.
  Taking conditional expectations therefore yields
  \begin{equation}
    \begin{aligned}
       & \sum_{t=1}^H\Pp\left(\mathsf{Near}_t^c\,\middle|\,(\xi_{b,r})_{b\in[m]\times[n],\,r\ge1}\right)
      \le\sum_{a\in I^\star\times J^\star}
      \E\left[\sum_{t=1}^H(1-q_{a,N_a(t)}(\rho/\sigma))\,\middle|\,(\xi_{b,r})_{b\in[m]\times[n],\,r\ge1}\right] \\
       & \quad=\sum_{a\in I^\star\times J^\star}\sum_{N=1}^{\infty}(1-q_{a,N}(\rho/\sigma))
      \sum_{t=1}^H\Pp\left(N_a(t)=N\,\middle|\,(\xi_{b,r})_{b\in[m]\times[n],\,r\ge1}\right).
    \end{aligned}
    \label{proof:eq:nonlocal-count-decomposition}
  \end{equation}

  Write $\cS_{\rm rel}=([m]\times J^\star)\cup(I^\star\times[n])$.
  We next lower-bound probabilities of observing core entries using Lemma~\ref{proof:lem:relevant-equilibrium-embedding}.
  By Lemma~\ref{proof:lem:conditional-posterior-law},
  \begin{align*}
     & \Pp\Bigl(\max_{b\in\cS_{\rm rel}}|\widetilde A_{b,t}-A_b^\star|\le\rho
          \,\Bigm|\,(\xi_{b,r})_{b\in[m]\times[n],\,r\ge1},(I_r,J_r)_{r<t}\Bigr)
    =\prod_{b\in\cS_{\rm rel}}q_{b,N_b(t)}(\rho/\sigma)                          \\
     & \quad\ge\prod_{b\in\cS_{\rm rel}}\inf_{N\ge1}q_{b,N}(\rho/\sigma).
  \end{align*}
  For $a=(i,j)\in I^\star\times J^\star$, Lemma~\ref{proof:lem:relevant-equilibrium-embedding} gives
  \begin{align*}
     & \Pp((I_t,J_t)=a\mid(\xi_{b,r})_{b\in[m]\times[n],\,r\ge1},(I_r,J_r)_{r<t})
    =\E[x_{t,i}y_{t,j}\mid(\xi_{b,r})_{b\in[m]\times[n],\,r\ge1},(I_r,J_r)_{r<t}]                                                                                                                                                 \\
     & \quad\ge\E\Bigl[\1_{\{\max_{b\in\cS_{\rm rel}}|\widetilde A_{b,t}-A_b^\star|\le\rho\}}x_{t,i}y_{t,j}\,\Bigm|\,(\xi_{b,r})_{b\in[m]\times[n],\,r\ge1},(I_r,J_r)_{r<t}\Bigr]                                                 \\
     & \quad\ge p^2\Pp\Bigl(\max_{b\in\cS_{\rm rel}}|\widetilde A_{b,t}-A_b^\star|\le\rho\,\Bigm|\,(\xi_{b,r})_{b\in[m]\times[n],\,r\ge1},(I_r,J_r)_{r<t}\Bigr) \ge p^2\prod_{b\in\cS_{\rm rel}}\inf_{N\ge1}q_{b,N}(\rho/\sigma).
  \end{align*}
  For every core entry $a$ and count $N$, observing $a$ increases its count, so
  \[
    \sum_{t=1}^H\1_{\{N_a(t)=N\}\cap\{(I_t,J_t)=a\}}\le1.
  \]
  Since $N_a(t)$ is determined by $(I_r,J_r)_{r<t}$, the tower property and the preceding probability bound give
  \begin{align*}
    1 & \ge\sum_{t=1}^H\Pp\left(\{N_a(t)=N\}\cap\{(I_t,J_t)=a\}\,\middle|\,(\xi_{b,r})_{b\in[m]\times[n],\,r\ge1}\right) \\
      & =\E\Bigl[\sum_{t=1}^H\1_{\{N_a(t)=N\}}
             \Pp\bigl((I_t,J_t)=a\mid(\xi_{b,r})_{b\in[m]\times[n],\,r\ge1},(I_r,J_r)_{r<t}\bigr)
             \,\Bigm|\,(\xi_{b,r})_{b\in[m]\times[n],\,r\ge1}\Bigr]                         \\
      & \ge p^2\prod_{b\in\cS_{\rm rel}}\inf_{k\ge1}q_{b,k}(\rho/\sigma)
    \sum_{t=1}^H\Pp\left(N_a(t)=N\,\middle|\,(\xi_{b,r})_{b\in[m]\times[n],\,r\ge1}\right).
  \end{align*}
  Lemma~\ref{proof:lem:all-count-localization} with $c=\rho/\sigma$ gives $\inf_{k\ge1}q_{b,k}(\rho/\sigma)>0$ almost surely for every $b\in\cS_{\rm rel}$.
  Consequently,
  \[
    \sum_{t=1}^H\Pp\left(N_a(t)=N\,\middle|\,(\xi_{b,r})_{b\in[m]\times[n],\,r\ge1}\right)
    \le\frac1{p^2\prod_{b\in\cS_{\rm rel}}\inf_{k\ge1}q_{b,k}(\rho/\sigma)}.
  \]
  Substituting this count bound into \eqref{proof:eq:nonlocal-count-decomposition}, we obtain
  \[
    \sum_{t=1}^H\Pp\left(\mathsf{Near}_t^c\,\middle|\,(\xi_{b,r})_{b\in[m]\times[n],\,r\ge1}\right)
    \le\frac1{p^2\prod_{b\in\cS_{\rm rel}}\inf_{N\ge1}q_{b,N}(\rho/\sigma)}\sum_{a\in I^\star\times J^\star}\sum_{N=1}^{\infty}(1-q_{a,N}(\rho/\sigma)).
  \]
  Independence of the complete entry-noise sequences gives
  \begin{align*}
    \sum_{t=1}^H\Pp(\mathsf{Near}_t^c)
     & \le\frac1{p^2}\sum_{a\in I^\star\times J^\star}
    \E\left[\frac{\sum_{N=1}^{\infty}(1-q_{a,N}(\rho/\sigma))}{\inf_{N\ge1}q_{a,N}(\rho/\sigma)}
        \prod_{b\in\cS_{\rm rel}\setminus\{a\}}\frac1{\inf_{N\ge1}q_{b,N}(\rho/\sigma)}\right]    \\
     & =\frac1{p^2}\sum_{a\in I^\star\times J^\star}
    \E\left[\frac{\sum_{N=1}^{\infty}(1-q_{a,N}(\rho/\sigma))}{\inf_{N\ge1}q_{a,N}(\rho/\sigma)}\right]
    \prod_{b\in\cS_{\rm rel}\setminus\{a\}}\E\left[\frac1{\inf_{N\ge1}q_{b,N}(\rho/\sigma)}\right] \\
     & =:B_{\rm loc}.
  \end{align*}
  Applying Lemma~\ref{proof:lem:all-count-localization} with $c=\rho/\sigma$ gives $B_{\rm loc}<\infty$.
\end{proof}

\subsection{Proof of Lemma~\ref{proof:lem:conditional-posterior-law}}
\begin{proof}
  By \eqref{eq:posterior}, represent the draws using independent $G_{a,t}\sim\mathcal N(0,1)$, independent of observation noises and past sampling randomness, as
  \begin{equation}
    \begin{aligned}
      \widehat A_{a,N} & =A_a^\star+\frac{\sigma S_{a,N}}
      N,                                                                             \\
      \widetilde A_{a,t}
                       & =\widehat A_{a,N_a(t)}+\frac{\sigma G_{a,t}}{\sqrt{N_a(t)}}
      =A_a^\star+\frac{\sigma S_{a,N_a(t)}}{N_a(t)}
      +\frac{\sigma G_{a,t}}{\sqrt{N_a(t)}}.
    \end{aligned}
    \label{proof:eq:posterior-error-decomposition}
  \end{equation}
  Conditional on $(\xi_{a,r})_{a\in[m]\times[n],\,r\ge1},(I_r,J_r)_{r<t}$, the current draws remain independent with
  \[
    G_{a,t}\mid(\xi_{a,r})_{a\in[m]\times[n],\,r\ge1},(I_r,J_r)_{r<t}\sim\mathcal N(0,1),\qquad a\in[m]\times[n].
  \]
  This conditioning fixes $N_a(t)$ and $\widehat A_{a,N_a(t)}$, so \eqref{proof:eq:posterior-error-decomposition} gives the asserted Gaussian distributions.
  The fixed equilibrium selection rule determines $(x_t,y_t)$ from $\widetilde A_t$, and independent action sampling gives
  \begin{equation*}
    \begin{aligned}
       & \Pp((I_t,J_t)=(i,j)\mid(\xi_{a,r})_{a\in[m]\times[n],\,r\ge1},(I_r,J_r)_{r<t},\widetilde A_t)       \\
       & \qquad=\Pp(I_t=i\mid(\xi_{a,r})_{a\in[m]\times[n],\,r\ge1},(I_r,J_r)_{r<t},\widetilde A_t)          \\
       & \qquad\quad\cdot\Pp(J_t=j\mid(\xi_{a,r})_{a\in[m]\times[n],\,r\ge1},(I_r,J_r)_{r<t},\widetilde A_t) \\
       & \qquad=x_{t,i}y_{t,j}.
    \end{aligned}
  \end{equation*}
\end{proof}

\subsection{Proof of Lemma~\ref{proof:lem:selected-posterior-deviation}}
\begin{proof}
  Let $L=\sum_{a\in\mathcal A}\sum_{t=1}^H\omega_{a,t}$.
  By \eqref{proof:eq:posterior-error-decomposition},
  \begin{equation}
    \sum_{a\in\mathcal A}\sum_{t=1}^H
    \omega_{a,t}|\widetilde A_{a,t}-A_a^\star|
    \le\sigma\sum_{a\in\mathcal A}\sum_{t=1}^H
    \omega_{a,t}\frac{|S_{a,N_a(t)}|}{N_a(t)}
    +\sigma\sum_{a\in\mathcal A}\sum_{t=1}^H
    \omega_{a,t}\frac{|G_{a,t}|}{\sqrt{N_a(t)}}.
    \label{proof:eq:core-two-sums}
  \end{equation}
  If $L=0$, both sums in \eqref{proof:eq:core-two-sums} vanish, so assume $L>0$.
  Since $\omega_{a,t}\le\1_{\{(I_t,J_t)=a\}}$, the observation counts with $\omega_{a,t}=1$ are distinct for each entry, and
  \[
    \begin{gathered}
      1\le N_a(t)\le t\le H,\qquad
      L\le H,\\
      \sum_{a\in\mathcal A}\sum_{t=1}^H
      \omega_{a,t}\1_{\{N_a(t)\le r\}}
      \le Dr\qquad(r\ge1).
    \end{gathered}
  \]
  Let $N_{(j)}^\uparrow$ be the $j$th smallest element of
  \[
    \{N_a(t):a\in\mathcal A,\ 1\le t\le H,\ \omega_{a,t}=1\},
  \]
  with equal counts from different pairs included separately.
  For $1\le j\le L$, the preceding count bound gives
  \[
    \begin{gathered}
      j\le\sum_{a\in\mathcal A}\sum_{t=1}^H
      \omega_{a,t}\1_{\{N_a(t)\le N_{(j)}^\uparrow\}}
      \le DN_{(j)}^\uparrow,\\
      \frac1{\sqrt{N_{(j)}^\uparrow}}\le\sqrt{\frac Dj}.
    \end{gathered}
  \]
  Enumerate the pairs with $\omega_{a,t}=1$ as $(a_1,t_1),\ldots,(a_L,t_L)$ so that $N_{a_j}(t_j)=N_{(j)}^\uparrow$.
  Let $W_{(j)}^{S,\downarrow}$ be the $j$th largest element of
  \[
    \{|S_{a,N}|/\sqrt N:a\in\mathcal A,\ 1\le N\le H\},
  \]
  and let $W_{(j)}^{G,\downarrow}$ be the $j$th largest element of
  \[
    \{|G_{a,t}|:a\in\mathcal A,\ 1\le t\le H\}.
  \]
  Each array has $DH$ elements, with equal values included separately.
  The selected indices $(a_j,N_{(j)}^\uparrow)$ and $(a_j,t_j)$ do not repeat.
  To bound the two sums in \eqref{proof:eq:core-two-sums}, we apply the rearrangement inequality to each array and obtain
  \begin{equation}
    \begin{aligned}
      \sum_{a\in\mathcal A}\sum_{t=1}^H
      \omega_{a,t}\frac{|S_{a,N_a(t)}|}{N_a(t)}
       & =\sum_{j=1}^{L}
      \frac{|S_{a_j,N_{(j)}^\uparrow}|}{N_{(j)}^\uparrow}
      =\sum_{j=1}^{L}
      \frac{|S_{a_j,N_{(j)}^\uparrow}|}{\sqrt{N_{(j)}^\uparrow}}
      \frac1{\sqrt{N_{(j)}^\uparrow}}                                                  \\
       & \le\sum_{j=1}^{L}
      \frac{W_{(j)}^{S,\downarrow}}{\sqrt{N_{(j)}^\uparrow}}
      \le\sqrt D\sum_{j=1}^H\frac{W_{(j)}^{S,\downarrow}}{\sqrt j},                    \\
      \sum_{a\in\mathcal A}\sum_{t=1}^H
      \omega_{a,t}\frac{|G_{a,t}|}{\sqrt{N_a(t)}}
       & =\sum_{j=1}^{L}|G_{a_j,t_j}|\frac1{\sqrt{N_{(j)}^\uparrow}} \le\sum_{j=1}^{L}
      \frac{W_{(j)}^{G,\downarrow}}{\sqrt{N_{(j)}^\uparrow}}
      \le\sqrt D\sum_{j=1}^H\frac{W_{(j)}^{G,\downarrow}}{\sqrt j}.
    \end{aligned}
    \label{proof:eq:selected-rearrangement-bounds}
  \end{equation}

  By \eqref{eq:observation} and~\eqref{proof:eq:posterior-error-decomposition}, every element of both arrays has the distribution of $|\mathcal N(0,1)|$.
  Apply Lemma~\ref{proof:lem:gaussian-order-statistics} to each array with $M=DH$ and $L=H$.
  The lemma does not require independence within an array.
  Combining it with \eqref{proof:eq:selected-rearrangement-bounds} and~\eqref{proof:eq:core-two-sums} gives
  \[
    \E\!
    \left[\sum_{a\in\mathcal A}\sum_{t=1}^H
      \omega_{a,t}|\widetilde A_{a,t}-A_a^\star|\right]
    \le\sigma\sqrt D\sum_{Q\in\{S,G\}}
    \E\!\left[\sum_{j=1}^H\frac{W_{(j)}^{Q,\downarrow}}{\sqrt j}\right] \le20\sigma\sqrt{DH\log(eD)}.
  \]
\end{proof}

\subsection{Proof of Lemma~\ref{proof:lem:finite-horizon-confidence}}
\begin{proof}
  Write $\cS_{\rm rel}=([m]\times J^\star)\cup(I^\star\times[n])$, so $|\cS_{\rm rel}|=d(m+n-d)$.
  To cover adaptive counts $N_a(t)$, define
  \[
    \begin{aligned}
      \mathcal A_T
       & =\bigcap_{a\in\cS_{\rm rel},1\le N\le T}
      \{|S_{a,N}|/\sqrt N\le\sqrt{2\log(4d(m+n-d)T/\delta)}\} \\
       & \quad\cap\bigcap_{a\in\cS_{\rm rel},1\le t\le H}
      \{|G_{a,t}|\le\sqrt{2\log(4d(m+n-d)T/\delta)}\}.
    \end{aligned}
  \]
  Since $S_{a,N}/\sqrt N\sim\mathcal N(0,1)$ and $G_{a,t}\sim\mathcal N(0,1)$ for fixed $N,t$, the Gaussian tail bound and a union bound give
  \[
    \Pp(\mathcal A_T^c)
    \le d(m+n-d)(T+H)\frac{\delta}{2d(m+n-d)T}
    =\frac{T+H}{2T}\delta\le\delta,
  \]
  where the last inequality uses $H\le T$.
  On $\mathcal A_T$, the bound $1\le N_a(t)\le t\le H\le T$ and \eqref{proof:eq:posterior-error-decomposition} yield, for every $a\in\cS_{\rm rel}$ and $1\le t\le H$,
  \[
    |\widetilde A_{a,t}-A_a^\star|
    \le\frac{\sigma}{\sqrt{N_a(t)}}
    \left(\frac{|S_{a,N_a(t)}|}{\sqrt{N_a(t)}}+|G_{a,t}|\right)
    \le2\sigma\sqrt{\frac{2\log(4d(m+n-d)T/\delta)}{N_a(t)}}.
  \]
  Thus, $\mathcal A_T\subseteq\mathcal E_T(\delta)$ and $\Pp(\mathcal E_T(\delta)^c)\le\delta$.

  Next, fix $1\le t\le H$ and a realization in $\mathcal E_T(\delta)\cap\mathsf{Near}_t$.
  For $a\in\cS_{\rm rel}$ with $N_a(t)\ge n_T^\star(\delta)$, the threshold definition gives
  \begin{equation}
    |\widetilde A_{a,t}-A_a^\star|
    \le2\sigma\sqrt{\frac{2\log(4d(m+n-d)T/\delta)}{N_a(t)}} \le2\sigma\sqrt{\frac{2\log(4d(m+n-d)T/\delta)}
      {(8\sigma^2/\rho^2)\log(4d(m+n-d)T/\delta)}}
    =\rho.
    \label{proof:eq:confidence-threshold-error}
  \end{equation}
  Set $C_t=(\widetilde A_t)_{I^\star,J^\star}$.
  For an outside row $i\notin I^\star$ with $\min_{j\in J^\star}N_{ij}(t)\ge n_T^\star(\delta)$, the row margin of Lemma~\ref{proof:lem:local-geometry-from-uniqueness} gives
  \begin{align*}
    U_i(\widetilde A_t)
     & =\sum_{j\in J^\star}A^\star_{ij}y_j^{C_t}-v^{C_t}
    +\sum_{j\in J^\star}y_j^{C_t}(\widetilde A_{ij,t}-A^\star_{ij}) \\
     & \le-\Delta_{\rm row}
    +\sum_{j\in J^\star}y_j^{C_t}|\widetilde A_{ij,t}-A^\star_{ij}|
    \le-\Delta_{\rm row}+\rho\sum_{j\in J^\star}y_j^{C_t}
    =-\Delta_{\rm row}+\rho<0,
  \end{align*}
  by \eqref{proof:eq:confidence-threshold-error}, $\sum_{j\in J^\star}y_j^{C_t}=1$ and $\Delta_{\rm row}\ge2\rho$.
  Similarly, for an outside column $j\notin J^\star$ with $N_{ij}(t)\ge n_T^\star(\delta)$ for every $i\in I^\star$,
  \begin{align*}
    V_j(\widetilde A_t)
     & =v^{C_t}-\sum_{i\in I^\star}x_i^{C_t}A^\star_{ij}
    +\sum_{i\in I^\star}x_i^{C_t}(A^\star_{ij}-\widetilde A_{ij,t}) \\
     & \le-\Delta_{\rm col}
    +\sum_{i\in I^\star} x_i^{C_t}|\widetilde A_{ij,t}-A^\star_{ij}|
    \le-\Delta_{\rm col}+\rho\sum_{i\in I^\star}x_i^{C_t}
    =-\Delta_{\rm col}+\rho<0,
  \end{align*}
  by the column margin, \eqref{proof:eq:confidence-threshold-error}, $\sum_{i\in I^\star}x_i^{C_t}=1$ and $\Delta_{\rm col}\ge2\rho$.
  Taking contrapositives, we obtain
  \[
    \begin{aligned}
      U_i(\widetilde A_t)>0
       & \ \Longrightarrow\ \min_{j\in J^\star}N_{ij}(t)<n_T^\star(\delta), \\
      V_j(\widetilde A_t)>0
       & \ \Longrightarrow\ \min_{i\in I^\star}N_{ij}(t)<n_T^\star(\delta).
    \end{aligned}
  \]
  Thus, on $\mathcal E_T(\delta)\cap\mathsf{Near}_t$, each outside row satisfies $U_i(\widetilde A_t)\le0$ or $\min_{j\in J^\star}
    N_{ij}(t)<n_T^\star(\delta)$, and each outside column satisfies the analogous condition.
  Intersecting with $\mathsf{Supp}_t^c$ gives $\mathcal E_T(\delta)\cap\mathsf{Near}_t\cap\mathsf{Supp}_t^c\subseteq\mathsf D_t(n_T^\star(\delta))$.

  Finally, for $\delta=T^{-1}$, the ceiling bound $\lceil z\rceil-1\le z$ gives
  \begin{align*}
    n_T^\star(T^{-1})-1
     & \le\frac{8\sigma^2}{\rho^2}\{\log(4d(m+n-d))+2\log T\}    \\
     & \le\frac{8\sigma^2\{2+\log(4d(m+n-d))\}}{\rho^2}\log(eT).
  \end{align*}
\end{proof}

\subsection{Proof of Lemma~\ref{proof:lem:outside-action-occupation}}
\begin{proof}
  Fix integers $H,k\ge1$.
  Each observation increases $N_{ij}$ by one, starting from $N_{ij}(1)=1$, so the counts at successive observations are $1,2,\ldots,N_{ij}(H+1)-1$.
  Hence, for every $(i,j)$,
  \[
    \sum_{t=1}^H\1_{\{(I_t,J_t)=(i,j)\}}\1_{\{N_{ij}(t)<k\}}
    =\sum_{r=1}^{N_{ij}(H+1)-1}\1_{\{r<k\}}
    =\min\{N_{ij}(H+1)-1,k-1\}\le k-1.
  \]
  Summing over the $d(m+n-2d)$ entries in the two disjoint outside blocks gives
  \begin{equation}
    \begin{aligned}
      \sum_{t=1}^H\1_{\{(I_t,J_t)\in\mathcal V_t^\times(k)\}}
       & =\sum_{(i,j)\in(I^{\star c}\times J^\star)\cup(I^\star\times J^{\star c})}
      \sum_{t=1}^H\1_{\{(I_t,J_t)=(i,j)\}}\1_{\{N_{ij}(t)<k\}}                      \\
       & \le d(m+n-2d)(k-1).
    \end{aligned}
    \label{proof:eq:outside-count-budget}
  \end{equation}

  We first compare rounds and observations when exactly one player has outside actions.
  \begin{lemma}
    \label{proof:lem:one-sided-observation-comparison}
    Suppose that exactly one player has actions outside its equilibrium support, that is, $m>d$ and $n=d$, or $m=d$ and $n>d$.
    There exist finite nonnegative constants $C_{\rm st}^{\rm 1s}$ and $B_{\rm 1s}$, depending only on $(A^\star,\sigma,m,n)$, such that, for every $H,k\ge1$,
    \[
      \sum_{t=1}^H\Pp(\mathsf D_t(k))
      \le C_{\rm st}^{\rm 1s}\sum_{t=1}^H\Pp((I_t,J_t)\in\mathcal V_t^\times(k))+B_{\rm 1s}.
    \]
  \end{lemma}
  When either $m>d$ and $n=d$ or $m=d$ and $n>d$, Lemma~\ref{proof:lem:one-sided-observation-comparison} and \eqref{proof:eq:outside-count-budget} give
  \[
    \sum_{t=1}^H\Pp(\mathsf D_t(k))
    \le C_{\rm st}^{\rm 1s}\sum_{t=1}^H\Pp((I_t,J_t)\in\mathcal V_t^\times(k))+B_{\rm 1s} \le d(m+n-2d)(k-1)C_{\rm st}^{\rm 1s}+B_{\rm 1s}.
  \]
  Thus, we set $C_{\rm occ}=C_{\rm st}^{\rm 1s}$ and $B_{\rm occ}=B_{\rm 1s}$.

  When both players have outside actions, we use the following comparison.
  \begin{lemma}
    \label{proof:lem:two-sided-observation-comparison}
    Suppose that $A^\star$ has a unique Nash equilibrium and both players have actions outside their equilibrium supports, that is, $m>d$ and $n>d$.
    Let $\rho,p,\Delta_{\rm row},\Delta_{\rm col}$ and $\kappa_x,\kappa_y$ be the constants of Lemmas~\ref{proof:lem:local-geometry-from-uniqueness} and~\ref{proof:lem:gradient-margins}.
    There exist a finite nonnegative constant $B_{\rm bad}$, depending only on $(A^\star,\sigma,m,n)$, and a nonnegative random variable $K_{\rm reg}$, independent of the horizon $T$ and depending only on $(A^\star,\sigma,m,n)$ and $(\xi_{a,r})_{a\in[m]\times[n],\,r\ge1}$, such that, for every $H,k\ge1$,
    \[
      \sum_{t=1}^H\Pp(\mathsf D_t(k))
      \le\E\left[K_{\rm reg}\sum_{t=1}^H\1_{\{(I_t,J_t)\in\mathcal V_t^\times(k)\}}\right]+B_{\rm bad}.
    \]
    Moreover, there exists a finite nonnegative constant $C_{\rm reg}^{\rm geom}$, depending only on $(\normmax{A^\star},\rho,p,\sigma,\Delta_{\rm row},\Delta_{\rm col},\kappa_x,\kappa_y)$, such that
    \[
      \E[K_{\rm reg}]\le C_{\rm reg}^{\rm geom}(1+\log(2mn))^2<\infty.
    \]
  \end{lemma}
  When $m>d$ and $n>d$, Lemma~\ref{proof:lem:two-sided-observation-comparison} and \eqref{proof:eq:outside-count-budget} give
  \[
    \begin{aligned}
      \sum_{t=1}^H\Pp(\mathsf D_t(k))
       & \le\E\left[K_{\rm reg}\sum_{t=1}^H\1_{\{(I_t,J_t)\in\mathcal V_t^\times(k)\}}\right]+B_{\rm bad} \\
       & \le d(m+n-2d)(k-1)\E[K_{\rm reg}]+B_{\rm bad}.
    \end{aligned}
  \]
  Thus, we set $C_{\rm occ}=\E[K_{\rm reg}]$ and $B_{\rm occ}=B_{\rm bad}$.
\end{proof}

\subsection{Proof of Lemma~\ref{proof:lem:one-sided-observation-comparison}}
\begin{proof}
  Fix $H,k\ge1$.
  We bound $\sum_{t=1}^H\Pp(\mathsf D_t(k))$ by comparing positive-margin probabilities with observation probabilities and bounding the contribution of empirical errors.
  The following lemma compares outside-margin events with observations of entries whose counts are below $k$.
  \begin{lemma}
    \label{proof:lem:common-observation-comparison}
    Suppose that $m>d$ or $n>d$.
    Fix $t,k\ge1$, $C\in\mathcal L_\rho$, and sets $\mathcal I\subseteq[m]\setminus I^\star$ and $\mathcal J\subseteq[n]\setminus J^\star$.
    For each $i\in\mathcal I$ and $j\in\mathcal J$, choose $j_i\in J^\star$ and $i_j\in I^\star$ with $N_{i j_i}(t)<k$ and $N_{i_j j}(t)<k$ using a fixed rule based on the counts.
    Then,
    \begin{align*}
       & \Pp_t^C\!\left(\mathsf D_t(k)
      \cap\bigcap_{u\notin I^\star\cup\mathcal I}\{U_u(\widetilde A_t)\le0\}
      \cap\bigcap_{v\notin J^\star\cup\mathcal J}\{V_v(\widetilde A_t)\le0\}\right)    \\
       & \quad\le\prod_{u\notin I^\star\cup\mathcal I}\Pp_t^C(U_u(\widetilde A_t)\le0)
      \prod_{v\notin J^\star\cup\mathcal J}\Pp_t^C(V_v(\widetilde A_t)\le0)
      \\ &\qquad\cdot\left[\sum_{i\in\mathcal I}\Pp_t^C(U_i(\widetilde A_t)>0)
                       +\sum_{j\in\mathcal J}\Pp_t^C(V_j(\widetilde A_t)>0)\right].
    \end{align*}
    For any $L_{\rm row},L_{\rm col}>0$ and $g_{\rm row},g_{\rm col}\ge0$, we also have
    \begin{align*}
       & \Pp_t^C((I_t,J_t)\in\mathcal V_t^\times(k))
      \\ &\quad\ge\sum_{i\in\mathcal I}\E_t^C\!\left[
                                                 x_{t,i}y_{t,j_i}\1_{\{0<U_i(\widetilde A_t)\le L_{\rm row}\}}
                                                 \prod_{u\notin I^\star\cup\{i\}}\1_{\{U_u(\widetilde A_t)\le-g_{\rm row}\}}
                                                 \prod_{v\notin J^\star}\1_{\{V_v(\widetilde A_t)\le-g_{\rm col}\}}\right]              \\
       & \qquad+\sum_{j\in\mathcal J}\E_t^C\!\left[
                                               x_{t,i_j}y_{t,j}\1_{\{0<V_j(\widetilde A_t)\le L_{\rm col}\}}
                                               \prod_{u\notin I^\star}\1_{\{U_u(\widetilde A_t)\le-g_{\rm row}\}}
                                               \prod_{v\notin J^\star\cup\{j\}}\1_{\{V_v(\widetilde A_t)\le-g_{\rm col}\}}\right].
    \end{align*}
  \end{lemma}

  We use the conditioning of Lemma~\ref{proof:lem:conditional-posterior-law}, with counts and empirical means fixed.
  First, consider localized cores satisfying
  \begin{equation}
    \begin{aligned}
      \Pp_t^C(U_i(\widetilde A_t)\le-\Delta_{\rm row}/4) & \ge1-\frac1{8(m+n-2d)}
                                                         &                        & (i\notin I^\star), \\
      \Pp_t^C(V_j(\widetilde A_t)\le-\Delta_{\rm col}/4) & \ge1-\frac1{8(m+n-2d)}
                                                         &                        & (j\notin J^\star).
    \end{aligned}.
    \label{proof:eq:outside-margin-probabilities}
  \end{equation}
  Omit empty index sets and let $\mathsf B_t$ be the event that $C_t$ is localized and at least one inequality in \eqref{proof:eq:outside-margin-probabilities} fails.
  Since $\mathsf D_t(k)\subseteq\mathsf{Near}_t$, the tower property gives
  \begin{equation}
    \sum_{t=1}^H\Pp(\mathsf D_t(k))
    \le\E\left[\sum_{t=1}^H\1_{\mathsf{Near}_t\cap\mathsf B_t^c}
      \Pp_t^{C_t}(\mathsf D_t(k))\right]
    +\sum_{t=1}^H\Pp(\mathsf B_t).
    \label{proof:eq:outside-good-bad-decomposition}
  \end{equation}
  For the first term in \eqref{proof:eq:outside-good-bad-decomposition}, we work with cores satisfying \eqref{proof:eq:outside-margin-probabilities}.
  Under these inequalities, the margins have negative means, and all except any chosen margin lie below their respective thresholds with probability at least $7/8$ by a union bound.

  Suppose $n=d<m$.
  Fix $C\in\mathcal L_\rho$ satisfying \eqref{proof:eq:outside-margin-probabilities} at round $t$ and $i\notin I^\star$.
  Let $B\in\R^{m\times n}$ satisfy $B_{I^\star,J^\star}=C$ and let $(x,y)\in\NE(B)$.
  Suppose that $\theta=\max_{h\notin I^\star}U_h(B)>0$, choose a maximizing row $i_0$, and set $M_1=\max_{j\in J^\star}|B_{i_0j}|$.
  Testing against $y^C$, a core row from Lemma~\ref{proof:lem:gradient-margins}, and row $i_0$ gives
  \[
    \begin{aligned}
      v(B)-v^C & \le\theta,                \\
      v(B)-v^C & \ge\kappa_y\|y-y^C\|_1,   \\
      v(B)-v^C & \ge\theta-M_1\|y-y^C\|_1.
    \end{aligned}
  \]
  Combining the last two inequalities and then testing against $y^C$ yields
  \[
    \frac{\kappa_y\theta}{\kappa_y+M_1}
    \le v(B)-v^C\le x^\trans By^C-v^C
    =\sum_{h\notin I^\star}x_hU_h(B)
    \le\theta\sum_{h:U_h(B)>0}x_h.
  \]
  Consequently,
  \[
    \|y-y^C\|_1\le\frac\theta{\kappa_y},\qquad
    \sum_{h:U_h(B)>0}x_h\ge\frac{\kappa_y}{\kappa_y+M_1}.
  \]
  If $i_0$ is the only outside row with a positive margin and $\theta\le p\kappa_y/2$, these bounds imply
  \begin{equation}
    \begin{gathered}
      x_{i_0}\ge\frac{\kappa_y}{\kappa_y+M_1},\qquad
      y_j\ge y_j^C-\|y-y^C\|_1\ge p/2,\\
      x_{i_0}y_j\ge\frac{p\kappa_y}{2(\kappa_y+M_1)}
      \qquad(j\in J^\star).
    \end{gathered}
    \label{proof:eq:one-sided-observation-bound}
  \end{equation}
  Apply \eqref{proof:eq:one-sided-observation-bound} with $B=\widetilde A_t$ and $i_0=i$ when $0<U_i(\widetilde A_t)\le p\kappa_y/2$ and all other outside row margins satisfy $U_h(\widetilde A_t)\le-\Delta_{\rm row}/4<0$.
  This gives
  \[
    x_{t,i}y_{t,j}\ge\frac{p\kappa_y}{2(\kappa_y+\max_{h\in J^\star}|\widetilde A_{ih,t}|)}
    \qquad(j\in J^\star).
  \]
  These conditions on the other margins are independent of this row and hold with probability at least $1/2$.
  To average over the sampled row, apply Lemma~\ref{proof:lem:window-regression} with $w=y^C,h=v^C,M=\normmax{A^\star}+\rho$, using $y_j^C\ge p$, $|v^C|\le \normmax{A^\star}+\rho$, and $N_{ij}(t)\ge1$.
  Put $L=p\kappa_y/2$.
  The conditional coordinate variances are $\sigma^2/N_{ih}(t)\le\sigma^2$.
  Conditional Jensen's inequality gives
  \[
    \begin{aligned}
      \Pp_t^C((I_t,J_t)=(i,j))
       & \ge\frac{p\kappa_y}{4}\E_t^C\!\left[
                                         \frac{\1_{\{0<U_i(\widetilde A_t)\le L\}}}{\kappa_y+\max_{h\in J^\star}|\widetilde A_{ih,t}|}\right] \\
       & \ge\frac{p\kappa_y}{4}
      \frac{\Pp_t^C(0<U_i(\widetilde A_t)\le L)}
      {\kappa_y+\E_t^C[\max_{h\in J^\star}|\widetilde A_{ih,t}|\mid0<U_i(\widetilde A_t)\le L]}.
    \end{aligned}
  \]
  Here, the factor $1/2$ accounts for the independent other-margin event.
  For cores satisfying \eqref{proof:eq:outside-margin-probabilities}, the conditional mean of $U_i(\widetilde A_t)$ is negative and its standard deviation is at most $\sigma$.
  Thus, Lemma~\ref{proof:lem:gaussian-interval-probabilities} compares the interval probability with $\Pp_t^C(U_i(\widetilde A_t)>0)$.
  If $\max_{h\in J^\star}|\widehat A_{ih,N_{ih}(t)}-A^\star_{ih}|\le\Delta_{\rm row}/2$, these bounds give the comparison coefficient
  \begin{equation}
    \begin{aligned}
      C_{\rm st}^{\rm 1s} & =\frac{4}{p\kappa_y(1-e^{-\min\{1,p\kappa_y/(2\sigma)\}^2/2})}        \\
                          & \qquad\cdot\left[\kappa_y+\frac{\normmax{A^\star}+\rho+p\kappa_y/2}{p}
                                         +(1+p^{-1})(\normmax{A^\star}+\Delta_{\rm row}/2)
                                         +\sigma\sqrt{2\log(2d)}\right].
    \end{aligned}
    \label{proof:eq:one-sided-stable-coefficient}
  \end{equation}
  Hence, for every $j\in J^\star$,
  \begin{equation}
    \Pp_t^C(U_i(\widetilde A_t)>0)
    \le C_{\rm st}^{\rm 1s}\Pp_t^C((I_t,J_t)=(i,j)).
    \label{proof:eq:one-sided-stable-comparison}
  \end{equation}
  Choose one entry below count $k$ in each outside row where one exists, using the fixed count-based rule.
  Applying \eqref{proof:eq:one-sided-stable-comparison} to rows whose empirical errors are at most $\Delta_{\rm row}/2$, and using \eqref{proof:eq:sampled-embedding-equivalence} for the remaining rows, a union bound gives
  \begin{equation}
    \begin{aligned}
      \Pp_t^C(\mathsf D_t(k))
       & \le C_{\rm st}^{\rm 1s}\Pp_t^C((I_t,J_t)\in\mathcal V_t^\times(k)) \\
       & \quad+\sum_{i\notin I^\star}\sum_{j\in J^\star}
      \1_{\{|\widehat A_{ij,N_{ij}(t)}-A^\star_{ij}|>\Delta_{\rm row}/2\}}
      \Pp_t^C(U_i(\widetilde A_t)>0).
    \end{aligned}
    \label{proof:eq:one-sided-error-decomposition}
  \end{equation}

  For the second term in \eqref{proof:eq:one-sided-error-decomposition}, we control the empirical errors uniformly over counts using
  \[
    \max_{h\in J^\star}|\widehat A_{ih,N_{ih}(t)}|
    \le \normmax{A^\star}+\Delta_{\rm row}/4
    +\sigma\max_{h\in J^\star}\sup_{r\ge1}\frac{(|S_{ih,r}|-\Delta_{\rm row}r/(4\sigma))_+}{\sqrt r}.
  \]
  Replacing $\normmax{A^\star}+\Delta_{\rm row}/2$ in \eqref{proof:eq:one-sided-stable-coefficient} by this empirical-mean bound gives, for every $j\in J^\star$,
  \begin{equation}
    \begin{aligned}
      \Pp_t^C(U_i(\widetilde A_t)>0)
       & \le\Biggl[C_{\rm st}^{\rm 1s}
              +\frac{4\sigma(1+p^{-1})}{p\kappa_y(1-e^{-\min\{1,p\kappa_y/(2\sigma)\}^2/2})}
              \\[-1mm]
              &\qquad\cdot\max_{h\in J^\star}\sup_{r\ge1}\frac{(|S_{ih,r}|-\Delta_{\rm row}r/(4\sigma))_+}{\sqrt r}\Biggr]
      \Pp_t^C((I_t,J_t)=(i,j)).
    \end{aligned}
    \label{proof:eq:common-rowwise-comparison}
  \end{equation}
  Each entry has distinct counts at successive observations, so pathwise
  \begin{equation}
    \begin{aligned}
       & \sum_{t=1}^H\1_{\{(I_t,J_t)=(i,j)\}}
      \1_{\{|S_{ij,N_{ij}(t)}|>\Delta_{\rm row}
        N_{ij}(t)/(2\sigma)\}} \le\sum_{N\ge1}\1_{\{|S_{ij,N}|>\Delta_{\rm row}N/(2\sigma)\}}                                                \\
       & \quad\le\frac{16\sigma^2}{\Delta_{\rm row}^2}\left(\sup_{r\ge1}\frac{(|S_{ij,r}|-\Delta_{\rm row}r/(4\sigma))_+}{\sqrt r}\right)^2.
    \end{aligned}
    \label{proof:eq:erroneous-observation-count}
  \end{equation}
  The last inequality follows by subtracting $\Delta_{\rm row}N/(4\sigma)$ from $|S_{ij,N}|$ at every indicated count.
  The second term in \eqref{proof:eq:outside-good-bad-decomposition} is bounded by the following lemma.
  \begin{lemma}
    \label{proof:prop:two-sided-bad-center-occupation}
    Suppose that $A^\star$ has a unique Nash equilibrium and at least one player has actions outside its equilibrium support, that is, $m>d$ or $n>d$.
    There is a finite nonnegative constant $B_{\rm bad}$, depending only on $(A^\star,\sigma,m,n)$, such that, for every $H\ge1$,
    \[
      \sum_{t=1}^H\Pp(\mathsf B_t)\le B_{\rm bad}.
    \]
  \end{lemma}

  Substitute \eqref{proof:eq:common-rowwise-comparison} into \eqref{proof:eq:one-sided-error-decomposition}, integrate over the good cores in \eqref{proof:eq:outside-good-bad-decomposition}, and use the tower property and \eqref{proof:eq:erroneous-observation-count}.
  Together with Lemma~\ref{proof:prop:two-sided-bad-center-occupation}, this gives
  \[
    B_{\rm 1s}=B_{\rm bad}+c\sum_{i\notin I^\star}\sum_{j\in J^\star}
    \E\left[\left(1+\sup_{r\ge1}\frac{(|S_{ij,r}|-\Delta_{\rm row}r/(4\sigma))_+}{\sqrt r}\right)^3\right]<\infty,
  \]
  where $c$ depends only on $(A^\star,\sigma,m,n)$.
  The cubic bound uses $ab^2\le(a^3+2b^3)/3$, and finiteness follows from Lemma~\ref{proof:lem:gaussian-drift-moments} with $r=3$.
  Thus, \eqref{proof:eq:outside-good-bad-decomposition} and \eqref{proof:eq:one-sided-error-decomposition} yield
  \[
    \sum_{t=1}^H\Pp(\mathsf D_t(k))
    \le C_{\rm st}^{\rm 1s}\sum_{t=1}^H\Pp((I_t,J_t)\in\mathcal V_t^\times(k))+B_{\rm 1s}.
  \]
  For $m=d<n$, apply the proof to $-A^{\star\trans}$, negating noises and draws and swapping players.
  Use the selector that swaps the original output on $-B^\trans$; counts and action probabilities are preserved, with $(\Delta_{\rm row},\kappa_y)$ replaced by $(\Delta_{\rm col},\kappa_x)$.
\end{proof}

\subsection{Proof of Lemma~\ref{proof:lem:common-observation-comparison}}
\begin{proof}
  Fix $t,k\ge1$ and $C\in\mathcal L_\rho$, as in the statement.
  On $\mathsf{Near}_t$, Lemma~\ref{proof:lem:local-geometry-from-uniqueness} and the definitions of the margins give
  \begin{align*}
    e_i^\trans\widetilde A_t\bar y^{C_t}-v^{C_t}
     & =\begin{cases}
          0,                   & i\in I^\star,    \\
          U_i(\widetilde A_t), & i\notin I^\star,
        \end{cases} \\
    v^{C_t}-(\bar x^{C_t})^\trans\widetilde A_te_j
     & =\begin{cases}
          0,                   & j\in J^\star,    \\
          V_j(\widetilde A_t), & j\notin J^\star.
        \end{cases}
  \end{align*}
  Conditional on the complete observation noises, past actions, and a localized core $C_t=C$, each outside margin is a nondegenerate Gaussian variable.
  For example, the variance of $U_i(\widetilde A_t)$ is $\sigma^2\sum_{j\in J^\star}(y_j^C)^2/N_{ij}(t)>0$.
  Hence, with probability one, no outside margin is zero on $\mathsf{Near}_t$, simultaneously for all $t\ge1$.

  If some outside margin is positive, the equilibrium conditions fail, so the embedded pair is not an equilibrium of $\widetilde A_t$.
  Otherwise, all outside margins are strictly negative, and the embedded pair is an equilibrium with value $v^{C_t}$.
  For any other equilibrium $(x,y)$, optimality gives
  \begin{align*}
    0 & =v^{C_t}-x^\trans\widetilde A_t\bar y^{C_t}
    =-\sum_{i\notin I^\star}x_iU_i(\widetilde A_t),   \\
    0 & =(\bar x^{C_t})^\trans\widetilde A_ty-v^{C_t}
    =-\sum_{j\notin J^\star}y_jV_j(\widetilde A_t).
  \end{align*}
  Strict negativity forces $x_i=0$ outside $I^\star$ and $y_j=0$ outside $J^\star$.
  Uniqueness of the restricted equilibrium in Lemma~\ref{proof:lem:local-geometry-from-uniqueness} gives $(x,y)=(\bar x^{C_t},\bar y^{C_t})$.
  Thus, on $\mathsf{Near}_t$ and almost surely,
  \begin{equation}
    \mathsf{Supp}_t^c
    \quad\Longleftrightarrow\quad
    \bigcup_{i\notin I^\star}\{U_i(\widetilde A_t)>0\}
    \cup\bigcup_{j\notin J^\star}\{V_j(\widetilde A_t)>0\}.
    \label{proof:eq:sampled-embedding-equivalence}
  \end{equation}

  By Lemma~\ref{proof:lem:conditional-posterior-law}, the outside rows restricted to $J^\star$ and outside columns restricted to $I^\star$ are conditionally independent.
  Equation~\eqref{proof:eq:sampled-embedding-equivalence} and a union bound give
  \begin{align*}
     & \Pp_t^C\!\left(\mathsf D_t(k)
    \cap\bigcap_{u\notin I^\star\cup\mathcal I}\{U_u(\widetilde A_t)\le0\}
    \cap\bigcap_{v\notin J^\star\cup\mathcal J}\{V_v(\widetilde A_t)\le0\}\right)    \\
     & \quad\le\prod_{u\notin I^\star\cup\mathcal I}\Pp_t^C(U_u(\widetilde A_t)\le0)
    \prod_{v\notin J^\star\cup\mathcal J}\Pp_t^C(V_v(\widetilde A_t)\le0)            \\
     & \qquad\cdot\left[1-\prod_{i\in\mathcal I}\Pp_t^C(U_i(\widetilde A_t)\le0)
                    \prod_{j\in\mathcal J}\Pp_t^C(V_j(\widetilde A_t)\le0)\right]    \\
     & \quad\le\prod_{u\notin I^\star\cup\mathcal I}\Pp_t^C(U_u(\widetilde A_t)\le0)
    \prod_{v\notin J^\star\cup\mathcal J}\Pp_t^C(V_v(\widetilde A_t)\le0)            \\
     & \qquad\cdot\left[\sum_{i\in\mathcal I}\Pp_t^C(U_i(\widetilde A_t)>0)
                    +\sum_{j\in\mathcal J}\Pp_t^C(V_j(\widetilde A_t)>0)\right].
  \end{align*}
  The events with exactly one positive outside margin are disjoint.
  By the definition of $\mathcal V_t^\times(k)$, the inequalities $N_{i j_i}(t)<k$ and $N_{i_j j}(t)<k$ imply $(i,j_i),(i_j,j)\in\mathcal V_t^\times(k)$.
  Lemma~\ref{proof:lem:conditional-posterior-law} therefore gives
  \begin{align*}
     & \Pp_t^C((I_t,J_t)\in\mathcal V_t^\times(k))
    =\E_t^C\!\left[\sum_{(h,c)\in\mathcal V_t^\times(k)}x_{t,h}y_{t,c}\right]                                                   \\
     & \ge\sum_{i\in\mathcal I}\E_t^C\!\left[x_{t,i}y_{t,j_i}
                                         \1_{\{0<U_i(\widetilde A_t)\le L_{\rm row}\}}
                                         \prod_{u\notin I^\star\cup\{i\}}\1_{\{U_u(\widetilde A_t)\le-g_{\rm row}\}}
                                         \prod_{v\notin J^\star}\1_{\{V_v(\widetilde A_t)\le-g_{\rm col}\}}\right]             \\
     & \quad+\sum_{j\in\mathcal J}\E_t^C\!\left[x_{t,i_j}y_{t,j}
                                            \1_{\{0<V_j(\widetilde A_t)\le L_{\rm col}\}}
                                            \prod_{u\notin I^\star}\1_{\{U_u(\widetilde A_t)\le-g_{\rm row}\}}
                                            \prod_{v\notin J^\star\cup\{j\}}\1_{\{V_v(\widetilde A_t)\le-g_{\rm col}\}}\right].
  \end{align*}
  By conditional independence, each row summand factors as
  \begin{align*}
     & \E_t^C\!\left[x_{t,i}y_{t,j_i}\1_{\{0<U_i(\widetilde A_t)\le L_{\rm row}\}}
                 \prod_{u\notin I^\star\cup\{i\}}\1_{\{U_u(\widetilde A_t)\le-g_{\rm row}\}}
                 \prod_{v\notin J^\star}\1_{\{V_v(\widetilde A_t)\le-g_{\rm col}\}}\right]       \\
     & \quad=\Pp_t^C(0<U_i(\widetilde A_t)\le L_{\rm row})
    \prod_{u\notin I^\star\cup\{i\}}\Pp_t^C(U_u(\widetilde A_t)\le-g_{\rm row})
    \prod_{v\notin J^\star}\Pp_t^C(V_v(\widetilde A_t)\le-g_{\rm col})                            \\
     & \qquad\cdot\E_t^C\!\left[x_{t,i}y_{t,j_i}\,\middle|\,
                            \substack{0<U_i(\widetilde A_t)\le L_{\rm row}\\
                              U_u(\widetilde A_t)\le-g_{\rm row}\ (u\notin I^\star\cup\{i\}),\quad
                              V_v(\widetilde A_t)\le-g_{\rm col}\ (v\notin J^\star)}\right].
  \end{align*}
  For a column $j\in\mathcal J$, the corresponding factorization is
  \begin{equation*}
    \begin{aligned}
       & \E_t^C\!\left[x_{t,i_j}y_{t,j}\1_{\{0<V_j(\widetilde A_t)\le L_{\rm col}\}}
                   \prod_{u\notin I^\star}\1_{\{U_u(\widetilde A_t)\le-g_{\rm row}\}}
                   \prod_{v\notin J^\star\cup\{j\}}\1_{\{V_v(\widetilde A_t)\le-g_{\rm col}\}}\right]   \\
       & \quad=\Pp_t^C(0<V_j(\widetilde A_t)\le L_{\rm col})
      \prod_{u\notin I^\star}\Pp_t^C(U_u(\widetilde A_t)\le-g_{\rm row})
      \prod_{v\notin J^\star\cup\{j\}}\Pp_t^C(V_v(\widetilde A_t)\le-g_{\rm col})                       \\
       & \qquad\cdot\E_t^C\!\left[x_{t,i_j}y_{t,j}\,\middle|\,
                              \substack{0<V_j(\widetilde A_t)\le L_{\rm col}\\
                                U_u(\widetilde A_t)\le-g_{\rm row}\ (u\notin I^\star),\quad
                                V_v(\widetilde A_t)\le-g_{\rm col}\ (v\notin J^\star\cup\{j\})}\right].
    \end{aligned}
  \end{equation*}
\end{proof}

\subsection{Proof of Lemma~\ref{proof:lem:two-sided-observation-comparison}}
\begin{proof}
  We first compare $\Pp_t^C(\mathsf D_t(k))$ with the probability of observing an entry in $\mathcal V_t^\times(k)$ for cores satisfying \eqref{proof:eq:outside-margin-probabilities}, then use Lemma~\ref{proof:prop:two-sided-bad-center-occupation} for the remaining rounds.
  Fix $H,k\ge1$ and the complete observation sequences and past actions, and use the product-Gaussian law conditional on $C\in\mathcal L_\rho$.

  We write $U_{mn,1}=\max_{a\in[m]\times[n]}\sup_{N\ge1}|S_{a,N}|/N$.
  For $0<L\le\Delta_{\rm col}/16$, write the row event from Lemma~\ref{proof:lem:common-observation-comparison} as
  \[
    \begin{aligned}
      \mathsf R_{i,t}(L)={} & \{0<U_i(\widetilde A_t)\le L\}
      \cap\bigcap_{u\notin I^\star\cup\{i\}}\{U_u(\widetilde A_t)\le-\Delta_{\rm row}/4\}                 \\
                            & \cap\bigcap_{v\notin J^\star}\{V_v(\widetilde A_t)\le-\Delta_{\rm col}/4\}.
    \end{aligned}
  \]
  Lemma~\ref{proof:lem:window-regression} with $M=\normmax{A^\star}+\rho$, Lemma~\ref{proof:lem:gaussian-maxima}, and the core bound give
  \begin{align}
     & \E_t^C[\max\{\normmax{A^\star}+\rho,\normmax{\widetilde A_t}\}\mid\mathsf R_{i,t}(L)]\le\overline M,
    \label{proof:eq:two-sided-conditional-entry-bound}                                                                \\
     & \overline M=\normmax{A^\star}+\rho+\frac{\normmax{A^\star}+\rho+\max(\Delta_{\rm row},\Delta_{\rm col})/16}{p}
    \notag                                                                                                            \\
     & \qquad+(3+p^{-1})(\normmax{A^\star}+\sigma U_{mn,1})
    +3\sigma\sqrt{2\log(2mn)}.\notag
  \end{align}
  Here, the Gaussian coordinate variances are at most $\sigma^2$, and conditioning on the independent other-margin event costs at most a factor two because its probability is at least $1/2$.
  The same bound holds for the column events.
  We use \eqref{proof:eq:two-sided-conditional-entry-bound} to apply the following observation bound.
  \begin{lemma}
    \label{proof:lem:unified-reveal}
    Suppose that $A^\star$ has a unique Nash equilibrium and use $\rho,p,\kappa_y$ from Lemmas~\ref{proof:lem:local-geometry-from-uniqueness} and~\ref{proof:lem:gradient-margins}.
    Let $B\in\R^{m\times n}$, $n>d$, and $C=B_{I^\star,J^\star}\in\mathcal L_\rho$.
    Suppose that $U_i(B)=\theta>0$ for $i\notin I^\star$, $U_u(B)\le0$ for $u\notin I^\star\cup\{i\}$, and $V_j(B)\le-g$ for $j\notin J^\star$, where $g>0$.
    Let $\max_{j\in J^\star}|B_{ij}|\le M_1$ and $|B_{hc}|\le M_2$ for $h\in I^\star\cup\{i\}$, $c\notin J^\star$, where $M_1,M_2\ge0$.
    For $M_1,M_2\ge0$, $g>0$, and $\kappa>0$, define
    \[
      \begin{gathered}
        \lambda(M_2;g,\kappa)=\min\left\{\frac g4,
        \frac{3p\kappa}{16[1+(M_2+\normmax{A^\star}+\rho)/g]}\right\},\\
        \eta(M_1,M_2;g,\kappa)=\frac{3p}{8[1+(M_2+\normmax{A^\star}+\rho)/g][1+4M_1/(3\kappa)]}.
      \end{gathered}
    \]
    For $g=+\infty$, the formulas use $g^{-1}=0$.
    If $\theta\le \lambda(M_2;g,\kappa_y)$, every $(x,y)\in\NE(B)$ satisfies, for $j\in J^\star$,
    \[
      x_i\ge\frac2p\eta(M_1,M_2;g,\kappa_y),\qquad y_j\ge p/2,\qquad
      x_iy_j\ge\eta(M_1,M_2;g,\kappa_y).
    \]
  \end{lemma}

  Set $L=\lambda(2\overline M;\Delta_{\rm col}/4,\kappa_y)\le\Delta_{\rm col}/16$.
  Equation~\eqref{proof:eq:two-sided-conditional-entry-bound} and conditional Markov's inequality give
  \[
    \Pp_t^C(\normmax{\widetilde A_t}\le2\overline M\mid\mathsf R_{i,t}(L))\ge\tfrac12,
    \qquad
    \Pp_t^C(\mathsf R_{i,t}(L))\ge\tfrac12\Pp_t^C(0<U_i(\widetilde A_t)\le L),
  \]
  where the second bound uses independence and \eqref{proof:eq:outside-margin-probabilities}.
  Apply Lemma~\ref{proof:lem:unified-reveal} on $\mathsf R_{i,t}(L)\cap\{\normmax{\widetilde A_t}\le2\overline M\}$ with $(M_1,M_2,g)=(2\overline M,2\overline M,\Delta_{\rm col}/4)$.
  Multiplying these probabilities and using Lemma~\ref{proof:lem:gaussian-interval-probabilities}, we obtain, for every $j\in J^\star$,
  \begin{equation}
    \begin{aligned}
       & \E_t^C[x_{t,i}y_{t,j}\1_{\mathsf R_{i,t}(L)}]
      \ge\frac{\eta(2\overline M,2\overline M;\Delta_{\rm col}/4,\kappa_y)}4
      \Pp_t^C(0<U_i(\widetilde A_t)\le L)                                            \\
       & \quad\ge\frac{\eta(2\overline M,2\overline M;\Delta_{\rm col}/4,\kappa_y)}4
      (1-e^{-\min\{1,L/\sigma\}^2/2})\Pp_t^C(U_i(\widetilde A_t)>0).
    \end{aligned}
    \label{proof:eq:two-sided-row-tail-comparison}
  \end{equation}
  Here, \eqref{proof:eq:outside-margin-probabilities} ensures the negative mean required for the Gaussian tail comparison.
  Applying the same argument to $-\widetilde A_t^\trans$ with equilibrium $(y_t,x_t)$ gives the column version of \eqref{proof:eq:two-sided-row-tail-comparison}, with $(\Delta_{\rm col},\kappa_y)$ replaced by $(\Delta_{\rm row},\kappa_x)$.
  Taking the reciprocal of the smaller row or column coefficient in \eqref{proof:eq:two-sided-row-tail-comparison}, set
  \[
    \begin{aligned}
      K_{\rm reg}=4\max_{(\Delta,\kappa)\in\{(\Delta_{\rm col},\kappa_y),(\Delta_{\rm row},\kappa_x)\}}
      \Biggl\{ & \frac1{\eta(2\overline M,2\overline M;\Delta/4,\kappa)}                                 \\
               & \cdot\frac1{1-e^{-\min\{1,\lambda(2\overline M;\Delta/4,\kappa)/\sigma\}^2/2}}\Biggr\}.
    \end{aligned}
  \]
  In Lemma~\ref{proof:lem:common-observation-comparison}, take $\mathcal I,\mathcal J$ to be the outside rows and columns containing entries below count $k$, and choose one such entry per row or column by the fixed rule.
  Substituting \eqref{proof:eq:two-sided-row-tail-comparison} and its column version gives the first bound below, while integration and Lemma~\ref{proof:prop:two-sided-bad-center-occupation} give the second.
  \[
    \begin{aligned}
      \Pp_t^C(\mathsf D_t(k))
       & \le K_{\rm reg}\Pp_t^C((I_t,J_t)\in\mathcal V_t^\times(k)),                                       \\
      \sum_{t=1}^H\Pp(\mathsf D_t(k))
       & \le\E\left[K_{\rm reg}\sum_{t=1}^H\1_{\{(I_t,J_t)\in\mathcal V_t^\times(k)\}}\right]+B_{\rm bad}.
    \end{aligned}
  \]
  The second line uses the tower property with the nonnegative $K_{\rm reg}$ fixed by the complete observation sequences.

  Finally, the definitions of $\lambda$ and $\eta$ and the scalar bound $(1-e^{-\min\{1,L/\sigma\}^2/2})^{-1}\le4(1+\sigma^2/L^2)$ give
  \[
    K_{\rm reg}\le c(1+\overline M)^4,
    \qquad
    \E[K_{\rm reg}]\le c\E[(1+\overline M)^4]
    \le C_{\rm reg}^{\rm geom}(1+\log(2mn))^2,
  \]
  where $c$ has only the stated geometric dependence, and the last inequality uses Lemma~\ref{proof:lem:complete-table-envelope} with $r=4$ and $(a+b)^4\le8(a^4+b^4)$.
\end{proof}

\subsection{Proof of Lemma~\ref{proof:prop:two-sided-bad-center-occupation}}
\begin{proof}
  By the definition of $\mathsf B_t$,
  \[
    \begin{aligned}
      \1_{\mathsf B_t}\le\1_{\mathsf{Near}_t}\Bigg[
                                               &\sum_{i\notin I^\star}
                                               \1_{\{\Pp_t^{C_t}(U_i(\widetilde A_t)\le-\Delta_{\rm row}/4)<1-1/[8(m+n-2d)]\}}\\
                                               &+\sum_{j\notin J^\star}
                                               \1_{\{\Pp_t^{C_t}(V_j(\widetilde A_t)\le-\Delta_{\rm col}/4)<1-1/[8(m+n-2d)]\}}\Bigg].
    \end{aligned}
  \]
  The following lemma controls a row or column whenever its margin inequality fails.
  \begin{lemma}
    \label{proof:lem:two-sided-bad-row-envelope}
    Suppose that $A^\star$ has a unique Nash equilibrium and at least one player has actions outside its equilibrium support ($m>d$ or $n>d$).
    Fix $C\in\mathcal L_\rho$ and use the conditional notation of Lemma~\ref{proof:lem:common-observation-comparison}.
    For each $i\notin I^\star$, there are nonnegative random variables $F_{i,N}^R$, depending only on $(A^\star,\sigma,m,n)$ and the observation sequences of $\{i\}\times J^\star$, such that $\E[\sum_{N\ge1}F_{i,N}^R]<\infty$.
    If $\Pp_t^C(U_i(\widetilde A_t)\le-\Delta_{\rm row}/4)<1-1/[8(m+n-2d)]$, then, uniformly over $t\ge1$ and $C\in\mathcal L_\rho$,
    \[
      \begin{aligned}
        \Bigl(\E_t^C\bigl[&\eta(\max_{j\in J^\star}|\widetilde A_{ij,t}|,\normmax{A^\star}+\rho;\Delta_{\rm col}-\rho,\kappa_y)\\
                      &\cdot\1_{\{0<U_i(\widetilde A_t)\le\min\{\lambda(\normmax{A^\star}+\rho;\Delta_{\rm col}-\rho,\kappa_y),\Delta_{\rm row}/2\}\}}\bigr]\Bigr)^{-1}
        \le F_{i,\min_{j\in J^\star}
                N_{ij}(t)}^R.
      \end{aligned}
    \]
    For columns, replace $(i,J^\star,U_i,\Delta_{\rm row},\Delta_{\rm col},\kappa_y,F_{i,N}^R)$ by $(j,I^\star,V_j,\Delta_{\rm col},\Delta_{\rm row},\kappa_x,F_{j,N}^C)$, take the maximum and minimum over $i\in I^\star$, and use the entries in $I^\star\times\{j\}$.
    The variables $F_{j,N}^C$ depend only on $(A^\star,\sigma,m,n)$ and the observation sequences of those entries, and satisfy $\E[\sum_{N\ge1}F_{j,N}^C]<\infty$.
  \end{lemma}

  To ensure the required inequalities for the other entries, define
  \[
    \begin{aligned}
      \cS_i^{\rm row} & =((I^{\star c}\setminus\{i\})\times J^\star)
      \cup((I^\star\cup\{i\})\times J^{\star c}),                    \\
      \cS_j^{\rm col} & =(I^\star\times(J^{\star c}\setminus\{j\}))
      \cup(I^{\star c}\times(J^\star\cup\{j\})).
    \end{aligned}
  \]
  Let $B\in\R^{m\times n}$ have $C=B_{I^\star,J^\star}\in\mathcal L_\rho$ and let $(x,y)\in\NE(B)$.
  Fix $i\notin I^\star$ and suppose that $|B_{uv}-A^\star_{uv}|\le\rho$ for every $(u,v)\in\cS_i^{\rm row}$ and $0<U_i(B)\le\lambda(\normmax{A^\star}+\rho;\Delta_{\rm col}-\rho,\kappa_y)$.
  Entrywise accuracy and Lemma~\ref{proof:lem:local-geometry-from-uniqueness} give
  \[
    \begin{aligned}
      U_h(B)   & \le-(\Delta_{\rm row}-\rho)<0 &  & (h\notin I^\star\cup\{i\}),                \\
      V_j(B)   & \le-(\Delta_{\rm col}-\rho)<0 &  & (j\notin J^\star),                         \\
      |B_{hj}| & \le \normmax{A^\star}+\rho    &  & (h\in I^\star\cup\{i\},\ j\notin J^\star).
    \end{aligned}
  \]
  If $n>d$, apply Lemma~\ref{proof:lem:unified-reveal} with $M_1=\max_{j\in J^\star}|B_{ij}|$, $M_2=\normmax{A^\star}+\rho$, and $g=\Delta_{\rm col}-\rho>0$, yielding
  \[
    x_iy_j\ge\eta(\max_{h\in J^\star}|B_{ih}|,\normmax{A^\star}+\rho;\Delta_{\rm col}-\rho,\kappa_y)
    \qquad(j\in J^\star).
  \]
  If $n=d$, then $\Delta_{\rm col}=+\infty$ and $U_i(B)\le3p\kappa_y/16\le p\kappa_y/2$.
  All other outside row margins are negative, so \eqref{proof:eq:one-sided-observation-bound} applies with $i_0=i$ and $M_1=\max_{h\in J^\star}|B_{ih}|$, giving
  \[
    x_iy_j\ge\frac{p\kappa_y}{2(\kappa_y+\max_{h\in J^\star}|B_{ih}|)}
    \ge\eta(\max_{h\in J^\star}|B_{ih}|,\normmax{A^\star}+\rho;+\infty,\kappa_y)
    \qquad(j\in J^\star).
  \]
  For an outside column $j$, apply the row argument to $-B^\trans$ with equilibrium $(y,x)$.
  If $|B_{uv}-A^\star_{uv}|\le\rho$ for every $(u,v)\in\cS_j^{\rm col}$ and $0<V_j(B)\le\lambda(\normmax{A^\star}+\rho;\Delta_{\rm row}-\rho,\kappa_x)$, it yields
  \[
    x_iy_j\ge\eta(\max_{h\in I^\star}|B_{hj}|,\normmax{A^\star}+\rho;\Delta_{\rm row}-\rho,\kappa_x)
    \qquad(i\in I^\star).
  \]

  Fix $i\notin I^\star$ and choose $j_t$ minimizing $N_{ij}(t)$ over $J^\star$ by the fixed count-based rule.
  Under $\Pp_t^C$, accuracy within $\rho$ on $\cS_i^{\rm row}$ has probability at least $\prod_{a\in\cS_i^{\rm row}}\inf_{N\ge1}q_{a,N}(\rho/\sigma)$ and is independent of the row entries.
  Combining the observation-probability bound above with Lemma~\ref{proof:lem:two-sided-bad-row-envelope} therefore gives
  \[
    \1_{\{\Pp_t^C(U_i(\widetilde A_t)\le-\Delta_{\rm row}/4)<1-1/[8(m+n-2d)]\}}
    \le\frac{F_{i,\min_{j\in J^\star}N_{ij}(t)}^R}
    {\prod_{a\in\cS_i^{\rm row}}\inf_{N\ge1}q_{a,N}(\rho/\sigma)}
    \Pp_t^C((I_t,J_t)=(i,j_t)).
  \]
  The multiplier and $j_t$ are fixed before integrating over localized cores.
  The tower property and the fact that each entry is observed at most once at each count give
  \[
    \sum_{t=1}^H F_{i,\min_{j\in J^\star}N_{ij}(t)}^R\1_{\{(I_t,J_t)=(i,j_t)\}}
    \le d\sum_{N\ge1}F_{i,N}^R.
  \]
  The observation sequences in $\cS_i^{\rm row}$ and $\{i\}\times J^\star$ are independent, so the row contribution is at most
  \[
    d\left(\prod_{a\in\cS_i^{\rm row}}
    \E\left[\frac1{\inf_{N\ge1}q_{a,N}(\rho/\sigma)}\right]\right)
    \E\left[\sum_{N\ge1}
      F_{i,N}^R\right]<\infty.
  \]
  Finiteness follows from Lemmas~\ref{proof:lem:all-count-localization} and~\ref{proof:lem:two-sided-bad-row-envelope}, with empty products equal to one.
  Adding these bounds over rows and the analogous column bounds proves the assertion with a finite $B_{\rm bad}$ depending only on $(A^\star,\sigma,m,n)$.
\end{proof}

\subsection{Proof of Lemma~\ref{proof:lem:unified-reveal}}
\begin{proof}
  Fix $(x,y)\in\NE(B)$ and put $\gamma=\sum_{j\notin J^\star}y_j$.
  Testing against the embedded core strategies gives
  \[
    0\le g\gamma\le v(B)-v^C\le\theta\le g/4,\qquad 1-\gamma\ge3/4.
  \]
  Write $\widehat y=y_{J^\star}/(1-\gamma)$.
  A maximizing core row from Lemma~\ref{proof:lem:gradient-margins} and row $i$ give
  \[
    \begin{aligned}
      \tfrac34\kappa_y\|\widehat y-y^C\|_1
       & \le v(B)-v^C+(M_2+\normmax{A^\star}+\rho)\gamma,                              \\
      v(B)-v^C
       & \ge\tfrac34\theta-M_1\|\widehat y-y^C\|_1-(M_2+\normmax{A^\star}+\rho)\gamma.
    \end{aligned}
  \]
  Here, the outside-column payoff is at least $-M_2\gamma$ in either row.
  Eliminating $\gamma\le(v(B)-v^C)/g$ gives
  \begin{equation}
    \begin{aligned}
      \|\widehat y-y^C\|_1
       & \le\frac4{3\kappa_y}\left(1+\frac{M_2+\normmax{A^\star}+\rho}{g}\right)(v(B)-v^C), \\
      v(B)-v^C
       & \ge\frac{3\theta}{4[1+(M_2+\normmax{A^\star}+\rho)/g][1+4M_1/(3\kappa_y)]}
      =\frac{2\theta}{p}\eta(M_1,M_2;g,\kappa_y).
    \end{aligned}
    \label{proof:eq:reveal-distance-and-value}
  \end{equation}
  Combining the reverse bound $v(B)-v^C\le x^\trans B\bar y^C-v^C\le x_i\theta$ with \eqref{proof:eq:reveal-distance-and-value} gives
  \[
    x_i\ge\frac{v(B)-v^C}{\theta}\ge\frac2p\eta(M_1,M_2;g,\kappa_y).
  \]
  Moreover, $v(B)-v^C\le\theta\le\lambda(M_2;g,\kappa_y)$ and the first line of \eqref{proof:eq:reveal-distance-and-value} give $\|\widehat y-y^C\|_1\le p/4$.
  Thus, for every $j\in J^\star$,
  \[
    y_j=(1-\gamma)\widehat y_j\ge\tfrac34(p-p/4)\ge p/2,
    \qquad x_iy_j\ge\eta(M_1,M_2;g,\kappa_y).
  \]
\end{proof}

\subsection{Proof of Lemma~\ref{proof:lem:two-sided-bad-row-envelope}}
\begin{proof}
  Fix $i\notin I^\star$ and set
  \[
    L=\min\{\lambda(\normmax{A^\star}+\rho;\Delta_{\rm col}-\rho,\kappa_y),\Delta_{\rm row}/2\}>0.
  \]
  The convention for $\Delta_{\rm col}=+\infty$ gives $L=\min\{3p\kappa_y/16,\Delta_{\rm row}/2\}$.
  Apply Lemma~\ref{proof:lem:gaussian-negative-mean-interval} to the observation sequences of $\{i\}\times J^\star$, with
  \[
    \mu=(A^\star_{ij})_{j\in J^\star},\quad
    w=y^C,\quad h=v^C,\quad\Delta=\Delta_{\rm row},\quad
    \varepsilon=\frac1{8(m+n-2d)}.
  \]
  Lemma~\ref{proof:lem:local-geometry-from-uniqueness} gives $w^\trans\mu-h\le-\Delta_{\rm row}$.
  Let $H_*$ denote the variable $H_\Delta$ from the Gaussian lemma for these sequences and parameters.
  Its simultaneous bounds apply to all current count vectors and all $C\in\mathcal L_\rho$.
  Under the hypothesis in the statement, they give
  \begin{equation}
    \begin{gathered}
      \frac{3\Delta_{\rm row}}{8\sigma}\sqrt{\min_{j\in J^\star}
        N_{ij}(t)}
      <H_*+\Phi^{-1}(1-1/[8(m+n-2d)]),\\
      \Pp_t^C(0<U_i(\widetilde A_t)\le L)^{-1}
      \le c_0^{-1}(1+H_*)e^{H_*^2/2}.
    \end{gathered}
    \label{proof:eq:bad-row-count-and-drift}
  \end{equation}
  Put $M_{i,t}=\max_{j\in J^\star}|\widetilde A_{ij,t}|$.
  Since $\max_{j\in J^\star}|\widehat A_{ij,N_{ij}(t)}|\le \normmax{A^\star}+3\Delta_{\rm row}/8+\sigma H_*$, Lemma~\ref{proof:lem:window-regression} with $M=\normmax{A^\star}+\rho$ gives
  \begin{equation}
    \begin{aligned}
       & \E_t^C[M_{i,t}\mid0<U_i(\widetilde A_t)\le L]
      \le\frac{\normmax{A^\star}+\rho+L}{p}
      +(1+p^{-1})(\normmax{A^\star}+3\Delta_{\rm row}/8+\sigma H_*) \\
       & \qquad+\sigma\sqrt{2\log(2d)}\le c(1+H_*),
    \end{aligned}
    \label{proof:eq:bad-row-conditional-maximum}
  \end{equation}
  where $c>0$ depends only on $(A^\star,\sigma,m,n)$.
  Conditional Jensen's inequality, \eqref{proof:eq:bad-row-conditional-maximum}, and the probability bound in \eqref{proof:eq:bad-row-count-and-drift} yield
  \[
    \begin{aligned}
       & \Bigl(\E_t^C[\eta(M_{i,t},\normmax{A^\star}+\rho;\Delta_{\rm col}-\rho,\kappa_y)
                        \1_{\{0<U_i(\widetilde A_t)\le L\}}]\Bigr)^{-1}
      \le\frac{8}{3p\Pp_t^C(0<U_i(\widetilde A_t)\le L)}                                   \\
       & \qquad\cdot\left(1+\frac{2(\normmax{A^\star}+\rho)}{\Delta_{\rm col}-\rho}\right)
      \left(1+\frac4{3\kappa_y}\E_t^C[M_{i,t}\mid0<U_i(\widetilde A_t)\le L]\right)        \\
       & \quad\le K_{\rm bad}(1+H_*)^2e^{H_*^2/2},
    \end{aligned}
  \]
  where the convention for $\Delta_{\rm col}=+\infty$ uses $(\Delta_{\rm col}-\rho)^{-1}=0$ and
  \[
    K_{\rm bad}=\frac8{3pc_0}
    \left(1+\frac{2(\normmax{A^\star}+\rho)}{\Delta_{\rm col}-\rho}\right)
    \left(1+\frac{4c}{3\kappa_y}\right).
  \]
  Thus the required variables can be chosen as
  \[
    F_{i,N}^R=K_{\rm bad}(1+H_*)^2e^{H_*^2/2}
    \1_{\{3\Delta_{\rm row}\sqrt N/(8\sigma)<H_*+\Phi^{-1}(1-1/[8(m+n-2d)])\}}.
  \]
  Equation~\eqref{proof:eq:bad-row-count-and-drift} proves the comparison uniformly in $t,C$.
  Counting the integers that satisfy the condition in $F_{i,N}^R$ gives
  \[
    \begin{aligned}
      \sum_{N\ge1}
      F_{i,N}^R
       & \le K_{\rm bad}\left(\frac{8\sigma}{3\Delta_{\rm row}}\right)^2
      (1+H_*)^2e^{H_*^2/2}
      \bigl(H_*+\Phi^{-1}(1-1/[8(m+n-2d)])\bigr)_+^2                     \\
       & \le C_{\rm env}(1+H_*)^4e^{H_*^2/2},
    \end{aligned}
  \]
  where $C_{\rm env}$ is finite and depends only on $(A^\star,\sigma,m,n)$.
  Lemma~\ref{proof:lem:gaussian-negative-mean-interval} with $r=4$ consequently gives $\E[\sum_N F_{i,N}^R]<\infty$.
  For columns, apply the same construction to $(-\widetilde A_{ij,t})_{i\in I^\star}$ with weights $x^C$, value $-v^C$, and $(\Delta_{\rm row},\Delta_{\rm col},\kappa_y)$ replaced by $(\Delta_{\rm col},\Delta_{\rm row},\kappa_x)$.
\end{proof}

\subsection{Dimension Dependence of the Cumulative Duality Gap Bound}
Equation~\eqref{proof:eq:one-sided-stable-coefficient}, with the corresponding column parameters when $m=d<n$, bounds $C_{\rm occ}$ by a constant times $1+\sqrt{\log(2d)}$ in either one-sided case.
The constant depends only on $(\normmax{A^\star},\rho,p,\sigma)$ and the relevant margin and gradient bounds.
For two-sided supports, Lemma~\ref{proof:lem:two-sided-observation-comparison} gives $C_{\rm occ}\le C_{\rm reg}^{\rm geom}(1+\log(2mn))^2$.
Using $d(m+n-d)\le mn$ in the explicit coefficients of the cumulative duality gap bound yields
\[
  \begin{aligned}
    C_{\rm mix}\sqrt T
     & \le20\sqrt2\sigma\min\{\chi_{\rm core}/p^2,2/p\}\,d\sqrt{T\log(ed)}, \\
    C_{\rm out}\log(eT)
     & \le\frac{16\normmax{A^\star}\sigma^2}{\rho^2}(2+\log4)
    d(m+n-2d)(1+\log(mn))C_{\rm occ}\log(eT).
  \end{aligned}
\]
These bounds give the asserted dimension dependence with the geometric quantities and $\sigma$ fixed.
Products of reciprocal probabilities contribute only to $C_0$, whose dimension dependence is unrestricted.

\section{Proofs for Exact Asymptotics}
\label{proof:sec:proof-exact-core-asymptotics}
\subsection{Proof of Theorem~\ref{thm:exact-core-asymptotics}}
\begin{proof}
  Let $H=T-mn$.
  We run the horizon-independent recursion indefinitely.
  Since $\mathsf{Supp}_t\subseteq\mathsf{Near}_t$ and $(x_t,y_t)=(\bar x^{C_t},\bar y^{C_t})$ on $\mathsf{Supp}_t$, we decompose the cumulative duality gap as
  \begin{equation}
    \begin{aligned}
      \mathcal R_T
       & =\mathcal R_{\rm init}
      +\sum_{t=1}^H\1_{\mathsf{Near}_t}\Gap_{A^\star}(\bar x^{C_t},\bar y^{C_t})
      +\sum_{t=1}^H\1_{\mathsf{Near}_t^c}\Gap_{A^\star}(x_t,y_t)     \\
       & \quad+\sum_{t=1}^H\1_{\mathsf{Near}_t\cap\mathsf{Supp}_t^c}
      \left[\Gap_{A^\star}(x_t,y_t)-\Gap_{A^\star}(\bar x^{C_t},\bar y^{C_t})\right],
    \end{aligned}
    \label{proof:eq:exact-regret-decomposition}
  \end{equation}
  where $\mathcal R_{\rm init}$ is the cumulative duality gap of the initial $mn$ rounds and $0\le\mathcal R_{\rm init}\le2\normmax{A^\star} mn$.
  For the first sum in \eqref{proof:eq:exact-regret-decomposition}, we use the following expectation limit.
  \begin{lemma}
    \label{proof:lem:main-general-core-proxy-limit}
    Suppose that $A^\star$ has a unique Nash equilibrium.
    For the horizon-independent recursion of Algorithm~\ref{alg:ts-ne}, we have, as $t\to\infty$,
    \[
      \sqrt t\,\E\left[\1_{\mathsf{Near}_t}\Gap_{A^\star}(\bar x^{C_t},\bar y^{C_t})\right]\to\Lambda_{\rm core}.
    \]
  \end{lemma}

  By Lemma~\ref{proof:lem:main-general-core-proxy-limit}, for every $\varepsilon>0$, there is an integer $t_0\ge1$ such that, for every $t\ge t_0$,
  \[
    \left|\sqrt t\,\E\left[\1_{\mathsf{Near}_t}\Gap_{A^\star}(\bar x^{C_t},\bar y^{C_t})\right]-\Lambda_{\rm core}\right|\le\varepsilon.
  \]
  Splitting the sum at $t_0$ and comparing $\sum_{t=t_0}^Ht^{-1/2}$ with its integral, we obtain, for every $H\ge t_0$,
  \[
    \begin{aligned}
       & \frac1{\sqrt H}\left|\sum_{t=1}^Ht^{-1/2}\left[\sqrt t\,\E\left[\1_{\mathsf{Near}_t}\Gap_{A^\star}(\bar x^{C_t},\bar y^{C_t})\right]-\Lambda_{\rm core}\right]\right| \\
       & \quad\le\frac1{\sqrt H}\sum_{t=1}^{t_0-1}
      \left|\E\left[\1_{\mathsf{Near}_t}\Gap_{A^\star}(\bar x^{C_t},\bar y^{C_t})\right]-\Lambda_{\rm core}t^{-1/2}\right|
      +\frac{\varepsilon}{\sqrt H}\sum_{t=t_0}^Ht^{-1/2}                                                                                                                       \\
       & \quad\le\frac1{\sqrt H}\sum_{t=1}^{t_0-1}
      \left|\E\left[\1_{\mathsf{Near}_t}\Gap_{A^\star}(\bar x^{C_t},\bar y^{C_t})\right]-\Lambda_{\rm core}t^{-1/2}\right|+2\varepsilon.
    \end{aligned}
  \]
  For fixed $\varepsilon$, the finite sum on the right is independent of $H$.
  Letting $H\to\infty$ therefore gives
  \[
    \limsup_{H\to\infty}\frac1{\sqrt H}
    \left|\sum_{t=1}^H\E\left[\1_{\mathsf{Near}_t}\Gap_{A^\star}(\bar x^{C_t},\bar y^{C_t})\right]
    -\Lambda_{\rm core}\sum_{t=1}^Ht^{-1/2}\right|\le2\varepsilon.
  \]
  Since $\varepsilon>0$ is arbitrary and the quantity inside the upper limit is nonnegative, that quantity converges to zero.
  Thus, by the definition of $o(\sqrt H)$,
  \[
    \sum_{t=1}^Ht^{-1/2}\left[\sqrt t\,\E\left[\1_{\mathsf{Near}_t}\Gap_{A^\star}(\bar x^{C_t},\bar y^{C_t})\right]-\Lambda_{\rm core}\right]
    =o(\sqrt H).
  \]
  Since integral comparison also gives $\sum_{t=1}^Ht^{-1/2}=2\sqrt H+\mathcal O(1)$ and $H/T\to1$, we obtain
  \begin{equation}
    \begin{aligned}
      \sum_{t=1}^H\E\left[\1_{\mathsf{Near}_t}\Gap_{A^\star}(\bar x^{C_t},\bar y^{C_t})\right]
       & =\Lambda_{\rm core}\sum_{t=1}^Ht^{-1/2}
      +\sum_{t=1}^Ht^{-1/2}\left[\sqrt t\,\E\left[\1_{\mathsf{Near}_t}\Gap_{A^\star}(\bar x^{C_t},\bar y^{C_t})\right]-\Lambda_{\rm core}\right] \\
       & =\Lambda_{\rm core}\sum_{t=1}^Ht^{-1/2}+o(\sqrt H)                                                                                      \\
       & =2\Lambda_{\rm core}\sqrt H+o(\sqrt H)
      =2\Lambda_{\rm core}\sqrt T+o(\sqrt T).
    \end{aligned}
    \label{proof:eq:expected-gap-sum}
  \end{equation}

  For the second sum in \eqref{proof:eq:exact-regret-decomposition}, Lemma~\ref{proof:lem:general-support-nonlocal-occupation} gives
  \begin{equation}
    \E\left[\sum_{t=1}^H\1_{\mathsf{Near}_t^c}\Gap_{A^\star}(x_t,y_t)\right]
    \le2\normmax{A^\star}\sum_{t=1}^H\Pp(\mathsf{Near}_t^c)
    \le2\normmax{A^\star}
    B_{\rm loc}.
    \label{proof:eq:asymptotic-second-sum}
  \end{equation}
  For the third sum in \eqref{proof:eq:exact-regret-decomposition}, the bound $\Gap_{A^\star}\in[0,2\normmax{A^\star}]$ gives
  \begin{equation}
    \begin{aligned}
       & \left|\E\left[\sum_{t=1}^H\1_{\mathsf{Near}_t\cap\mathsf{Supp}_t^c}
                   \left[\Gap_{A^\star}(x_t,y_t)-\Gap_{A^\star}(\bar x^{C_t},\bar y^{C_t})\right]\right]\right|
      \le2\normmax{A^\star}\sum_{t=1}^H\Pp(\mathsf{Near}_t\cap\mathsf{Supp}_t^c).
    \end{aligned}
    \label{proof:eq:asymptotic-third-sum}
  \end{equation}
  The bounds \eqref{proof:eq:outside-confidence-success-bound} and~\eqref{proof:eq:outside-confidence-failure-bound}, applied at horizon $H+mn$, give
  \[
    \begin{aligned}
      \sum_{t=1}^H\Pp(\mathsf{Near}_t\cap\mathsf{Supp}_t^c)
       & \le\1_{\{d(m+n-2d)>0\}}+B_{\rm occ}+\frac{8\sigma^2\{2+\log(4d(m+n-d))\}}{\rho^2}
      d(m+n-2d)C_{\rm occ}\log(eT)                                                         \\
       & =\mathcal O(\log T),
    \end{aligned}
  \]
  where $C_{\rm occ}=B_{\rm occ}=0$ for full support.
  For the upper bound, we apply \eqref{proof:eq:asymptotic-second-sum}, \eqref{proof:eq:asymptotic-third-sum}, and the bound above to the decomposition \eqref{proof:eq:exact-regret-decomposition}.
  Together with $0\le\mathcal R_{\rm init}\le2\normmax{A^\star} mn$ and \eqref{proof:eq:expected-gap-sum}, this gives
  \[
    \begin{aligned}
      \E[\mathcal R_T]
       & \le2\normmax{A^\star} mn
      +\sum_{t=1}^H\E\left[\1_{\mathsf{Near}_t}\Gap_{A^\star}(\bar x^{C_t},\bar y^{C_t})\right]
      +2\normmax{A^\star}
      B_{\rm loc}+\mathcal O(\log T)             \\
       & =2\Lambda_{\rm core}\sqrt T+o(\sqrt T).
    \end{aligned}
  \]
  For the lower bound, $\mathcal R_{\rm init}$ and the second sum in \eqref{proof:eq:exact-regret-decomposition} are nonnegative.
  Hence, \eqref{proof:eq:asymptotic-third-sum} and \eqref{proof:eq:expected-gap-sum} give
  \[
    \begin{aligned}
      \E[\mathcal R_T]
       & \ge\sum_{t=1}^H\E\left[\1_{\mathsf{Near}_t}\Gap_{A^\star}(\bar x^{C_t},\bar y^{C_t})\right]
      +\E\left[\sum_{t=1}^H\1_{\mathsf{Near}_t\cap\mathsf{Supp}_t^c}
           \left[\Gap_{A^\star}(x_t,y_t)-\Gap_{A^\star}(\bar x^{C_t},\bar y^{C_t})\right]\right]     \\
       & \ge\sum_{t=1}^H\E\left[\1_{\mathsf{Near}_t}\Gap_{A^\star}(\bar x^{C_t},\bar y^{C_t})\right] \\
       & \quad-2\normmax{A^\star}\sum_{t=1}^H\Pp(\mathsf{Near}_t\cap\mathsf{Supp}_t^c)               \\
       & =2\Lambda_{\rm core}\sqrt T+o(\sqrt T).
    \end{aligned}
  \]
  Combining the upper and lower bounds, we conclude that
  \[
    \E[\mathcal R_T]=2\Lambda_{\rm core}\sqrt T+o(\sqrt T).
  \]

  Within each family, the projections $(x^\star)_{I^\star}^\trans\mathsf W^{\rm core}e_j$ and $e_i^\trans\mathsf W^{\rm core}(y^\star)_{J^\star}$ are independent centered Gaussians with variances $2\sigma^2/y_j^\star$ and $2\sigma^2/x_i^\star$, respectively.
  For $d\ge2$ and $j\ne j'$, symmetry gives
  \[
    \E\left[\max_{j\in J^\star}(x^\star)_{I^\star}^\trans\mathsf W^{\rm core}e_j\right]
    \ge\tfrac12\E\left[\left|(x^\star)_{I^\star}^\trans\mathsf W^{\rm core}(e_j-e_{j'})\right|\right]>0.
  \]
  The other expected maximum in $\Lambda_{\rm core}$ is nonnegative, whereas for $d=1$ its two scalar terms cancel.
\end{proof}

\subsection{Proof of Lemma~\ref{proof:lem:main-general-core-proxy-limit}}
\begin{proof}
  We split the target expectation according to the observation counts.
  \begin{equation}
    \begin{aligned}
      \sqrt t\,\E\left[\1_{\mathsf{Near}_t}\Gap_{A^\star}(\bar x^{C_t},\bar y^{C_t})\right]
       & =\E\left[\sqrt t\,\1_{\mathsf{Near}_t\cap\{\min_{a\in I^\star\times J^\star}
                N_a(t)\ge p^2t/8\}}\Gap_{A^\star}(\bar x^{C_t},\bar y^{C_t})\right]                                                                        \\
       & \quad+\E\left[\sqrt t\,\1_{\mathsf{Near}_t\cap\{\min_{a\in I^\star\times J^\star}N_a(t)<p^2t/8\}}\Gap_{A^\star}(\bar x^{C_t},\bar y^{C_t})\right].
    \end{aligned}
    \label{proof:eq:gap-expectation-decomposition}
  \end{equation}
  For the first term in \eqref{proof:eq:gap-expectation-decomposition}, we use the following bounds on the observation counts and Gaussian limit.
  \begin{lemma}
    \label{proof:lem:main-general-core-allocation-limit}
    Suppose that $A^\star$ has a unique Nash equilibrium.
    Let $\rho,p$ be the constants from Lemma~\ref{proof:lem:local-geometry-from-uniqueness}.
    For the horizon-independent recursion of Algorithm~\ref{alg:ts-ne}, we have, for every $a=(i,j)\in I^\star\times J^\star$ and $t\ge1$,
    \[
      \Pp\left(N_a(t)<p^2t/8,\ \sum_{r=1}^t\1_{\mathsf{Supp}_r^c}\le t/4\right)\le e^{-p^2t/8},\qquad
      \Pp\left(\sum_{r=1}^t\1_{\mathsf{Supp}_r^c}>t/4\right)\le\frac4t\sum_{r=1}^t\Pp(\mathsf{Supp}_r^c).
    \]
    Moreover, as $t\to\infty$, we have
    \[
      \begin{gathered}
        C_t\longrightarrow A^\star_{I^\star,J^\star}\text{ in probability},\\
        N_{ij}(t)/t\longrightarrow x_i^\star y_j^\star
        \text{ in }L^1\text{ and in probability}\quad((i,j)\in I^\star\times J^\star),\\
        \sqrt t(C_t-A^\star_{I^\star,J^\star})\Rightarrow\mathsf W^{\rm core}.
      \end{gathered}
    \]
  \end{lemma}

  Lemma~\ref{proof:lem:main-general-core-allocation-limit}, continuity of the core equilibrium, and Lemma~\ref{proof:lem:core-gap-bound} give
  \begin{equation}
    \sqrt t\,\1_{\mathsf{Near}_t}\Gap_{A^\star}(\bar x^{C_t},\bar y^{C_t})\Rightarrow
    \max_{j\in J^\star}(x^\star)_{I^\star}^\trans\mathsf W^{\rm core}e_j
    -\min_{i\in I^\star}e_i^\trans\mathsf W^{\rm core}(y^\star)_{J^\star},
    \label{proof:eq:core-gap-distributional-limit}
  \end{equation}
  since $\Pp(\mathsf{Near}_t)\to1$.
  To control the probabilities in Lemma~\ref{proof:lem:main-general-core-allocation-limit}, we use $\mathsf{Supp}_r\subseteq\mathsf{Near}_r$ to write, for every $h\ge0$,
  \begin{equation}
    \begin{aligned}
      \sum_{r=1}^h\Pp(\mathsf{Supp}_r^c)
       & =\sum_{r=1}^h\Pp(\mathsf{Near}_r^c)
      +\sum_{r=1}^h\Pp(\mathsf{Near}_r\cap\mathsf{Supp}_r^c)                   \\
       & \le B_{\rm loc}+\sum_{r=1}^h\Pp(\mathsf{Near}_r\cap\mathsf{Supp}_r^c)
      =\mathcal O(\log(e(h+mn))),
    \end{aligned}
    \label{proof:eq:embedding-failure-count}
  \end{equation}
  where the inequality uses Lemma~\ref{proof:lem:general-support-nonlocal-occupation}, and the last equality follows from \eqref{proof:eq:outside-confidence-success-bound} and~\eqref{proof:eq:outside-confidence-failure-bound} at horizon $h+mn$.
  A union bound and the two probability bounds in Lemma~\ref{proof:lem:main-general-core-allocation-limit} now give
  \begin{equation}
    \Pp\left(\min_{a\in I^\star\times J^\star}
    N_a(t)<p^2t/8\right)
    \le d^2e^{-p^2t/8}+\frac4t\sum_{r=1}^t\Pp(\mathsf{Supp}_r^c)\longrightarrow0.
    \label{proof:eq:good-count-tail}
  \end{equation}
  Thus, the random variable inside the first expectation in \eqref{proof:eq:gap-expectation-decomposition} has the limit in \eqref{proof:eq:core-gap-distributional-limit}.
  On $\{\min_{a\in I^\star\times J^\star}N_a(t)\ge p^2t/8\}$, Lemma~\ref{proof:lem:core-gap-bound} and \eqref{proof:eq:posterior-error-decomposition} give
  \[
    \begin{aligned}
      \sqrt t\,\1_{\mathsf{Near}_t}\Gap_{A^\star}(\bar x^{C_t},\bar y^{C_t}) & \le2\sqrt t\,\normmax{C_t-A^\star_{I^\star,J^\star}}                                                                  \\
                                                                             & \le2\sigma\sqrt t\left[
                                                                                                  \max_{a\in I^\star\times J^\star}\sup_{N\ge\max\{1,\lfloor p^2t/8\rfloor\}}\frac{|S_{a,N}|}{N}
                                                                                                  +\frac{\max_{a\in I^\star\times J^\star}|G_{a,t}|}{\sqrt{\max\{1,\lfloor p^2t/8\rfloor\}}}\right].
    \end{aligned}
  \]
  Since $t/\max\{1,\lfloor p^2t/8\rfloor\}\le16/p^2$, the second-moment bounds of Lemmas~\ref{proof:lem:complete-table-envelope} and~\ref{proof:lem:gaussian-maxima}, with $D=d^2$ and $k=\max\{1,\lfloor p^2t/8\rfloor\}$ in the former, imply
  \begin{equation}
    \begin{aligned}
       & \E\left[t\1_{\mathsf{Near}_t\cap\{\min_{a\in I^\star\times J^\star}
               N_a(t)\ge p^2t/8\}}
             \Gap_{A^\star}(\bar x^{C_t},\bar y^{C_t})^2\right]                               \\
       & \quad\le8\sigma^2t\,\E\left[\max_{a\in I^\star\times J^\star}
                                 \sup_{N\ge\max\{1,\lfloor p^2t/8\rfloor\}}\frac{|S_{a,N}|^2}{N^2}\right]
      +\frac{128\sigma^2}{p^2}\E\left[\max_{a\in I^\star\times J^\star}|G_{a,t}|^2\right]                 \\
       & \quad\le\frac{512\sigma^2}{p^2}\{\log(2d^2)+1\},
    \end{aligned}
    \label{proof:eq:localized-gap-second-moment}
  \end{equation}
  where the first inequality uses $(u+v)^2\le2u^2+2v^2$.
  The second-moment bound \eqref{proof:eq:localized-gap-second-moment} gives uniform integrability.
  Hence, the first term in \eqref{proof:eq:gap-expectation-decomposition} tends to $\Lambda_{\rm core}$.
  For the second term in \eqref{proof:eq:gap-expectation-decomposition}, the bound $\1_{\mathsf{Near}_t}\Gap_{A^\star}(\bar x^{C_t},\bar y^{C_t})\le2\normmax{A^\star}$ and \eqref{proof:eq:good-count-tail} give
  \[
    \begin{aligned}
       & \E\left[\sqrt t\,\1_{\mathsf{Near}_t\cap\{\min_{a\in I^\star\times J^\star}
               N_a(t)<p^2t/8\}}\Gap_{A^\star}(\bar x^{C_t},\bar y^{C_t})\right] \\
       & \quad\le2\normmax{A^\star} d^2\sqrt t\,e^{-p^2t/8}
      +\frac{8\normmax{A^\star}}{\sqrt t}\sum_{r=1}^t\Pp(\mathsf{Supp}_r^c)\longrightarrow0.
    \end{aligned}
  \]
  Combining the two limits in \eqref{proof:eq:gap-expectation-decomposition}, we obtain the asserted expectation limit.
\end{proof}

\subsection{Proof of Lemma~\ref{proof:lem:main-general-core-allocation-limit}}
\begin{proof}
  By \eqref{proof:eq:posterior-error-decomposition}, for $a=(i,j)$,
  \[
    \sqrt t((C_t)_{ij}-A_{ij}^\star)
    =\sigma\left\{\sqrt t\,\frac{S_{a,N_a(t)}}{N_a(t)}
    +\sqrt{\frac t{N_a(t)}}G_{a,t}\right\}.
  \]
  The construction in Lemma~\ref{proof:lem:conditional-posterior-law} makes the current Gaussian draws independent of the complete observation sequences and current counts.
  The second limit in Lemma~\ref{proof:lem:gaussian-random-index-limit} reduces the Gaussian limit to $N_{ij}(t)/t\to x_i^\star y_j^\star>0$ jointly in probability.
  To identify the observation frequencies, we decompose
  \begin{equation}
    \frac{N_a(t)}t-x_i^\star y_j^\star =\frac{1-x_i^\star y_j^\star}{t} +\frac1t\sum_{r=1}^{t-1} \bigl(\1_{\{(I_r,J_r)=a\}}-x_{r,i}y_{r,j}\bigr)+\frac1t\sum_{r=1}^{t-1} (x_{r,i}y_{r,j}-x_i^\star y_j^\star).
    \label{proof:eq:core-frequency-decomposition}
  \end{equation}
  The first term in \eqref{proof:eq:core-frequency-decomposition} is bounded in absolute value by $t^{-1}$.
  For the second term in \eqref{proof:eq:core-frequency-decomposition}, the centered action indicators have conditional mean zero and orthogonal increments.
  Hence,
  \begin{equation}
    \E\left[\left(\sum_{r=1}^{t-1}
      (\1_{\{(I_r,J_r)=a\}}-x_{r,i}y_{r,j})\right)^2\right]
    =\sum_{r=1}^{t-1}\E\left[x_{r,i}y_{r,j}(1-x_{r,i}y_{r,j})\right]\le t-1.
    \label{proof:eq:core-action-second-moment}
  \end{equation}
  We apply Cauchy--Schwarz with \eqref{proof:eq:core-action-second-moment} to the second term in \eqref{proof:eq:core-frequency-decomposition}.
  Splitting the third term in \eqref{proof:eq:core-frequency-decomposition} over $\mathsf{Supp}_r$ and its complement, we obtain
  \begin{equation}
    \begin{aligned}
      \E\left[|N_a(t)/t-x_i^\star y_j^\star|\right]
       & \le t^{-1}+t^{-1/2}+\frac1t\sum_{r=1}^{t-1}\Pp(\mathsf{Supp}_r^c)
      +\frac1t\sum_{r=1}^{t-1}
      \E\left[\1_{\mathsf{Supp}_r}|x_i^{C_r}y_j^{C_r}-x_i^\star y_j^\star|\right] \\
       & \le t^{-1}+t^{-1/2}+\frac1t\sum_{r=1}^{t-1}\Pp(\mathsf{Supp}_r^c)
      +\frac1t\sum_{r=1}^{t-1}
      \E\left[\1_{\mathsf{Near}_r}|x_i^{C_r}y_j^{C_r}-x_i^\star y_j^\star|\right].
    \end{aligned}
    \label{proof:eq:core-frequency-convergence}
  \end{equation}
  where the first inequality uses $x_{r,i}y_{r,j}=x_i^{C_r}y_j^{C_r}$ on $\mathsf{Supp}_r$.
  The second inequality uses $\mathsf{Supp}_r\subseteq\mathsf{Near}_r$, and \eqref{proof:eq:embedding-failure-count} shows that the first three terms tend to zero.

  To send the last term in \eqref{proof:eq:core-frequency-convergence} to zero, we establish core consistency.
  On $\{\min_{a\in I^\star\times J^\star}N_a(t)\ge p^2t/8\}$, \eqref{proof:eq:posterior-error-decomposition} gives
  \[
    \normmax{C_t-A^\star_{I^\star,J^\star}}
    \le\sigma\max_{a\in I^\star\times J^\star}
    \sup_{N\ge\max\{1,\lfloor p^2t/8\rfloor\}}\frac{|S_{a,N}|}{N}
    +\frac{\sigma}{\sqrt{\max\{1,\lfloor p^2t/8\rfloor\}}}
    \max_{a\in I^\star\times J^\star}|G_{a,t}|.
  \]
  We apply Lemma~\ref{proof:lem:complete-table-envelope} with $D=d^2$ and $k=\max\{1,\lfloor p^2t/8\rfloor\}$ to the first term and Gaussian tails to the second.
  For every $\varepsilon>0$, these bounds give
  \[
    \Pp\left(\normmax{C_t-A^\star_{I^\star,J^\star}}>\varepsilon\right)
    \le\Pp\left(\min_{a\in I^\star\times J^\star}N_a(t)<p^2t/8\right)+4d^2\exp\!
    \left(-\frac{\max\{1,\lfloor p^2t/8\rfloor\}\varepsilon^2}{8\sigma^2}\right).
  \]

  For the remaining probability, a union bound gives
  \[
    \Pp\left(\min_{a\in I^\star\times J^\star}N_a(t)<p^2t/8\right) \le\Pp\left(\sum_{r=1}^t\1_{\mathsf{Supp}_r^c}>t/4\right)+\sum_{a\in I^\star\times J^\star}\Pp\left(N_a(t)<p^2t/8,\ \sum_{r=1}^t\1_{\mathsf{Supp}_r^c}\le t/4\right).
  \]
  Markov's inequality and \eqref{proof:eq:embedding-failure-count} give $\Pp(\sum_{r=1}^t\1_{\mathsf{Supp}_r^c}>t/4)\le\frac4t\sum_{r=1}^t\Pp(\mathsf{Supp}_r^c)$.
  For a core entry $a=(i,j)$, we set
  \[
    L_h=\exp\left\{-(\log2)\sum_{r=1}^h\1_{\mathsf{Supp}_r}\1_{\{(I_r,J_r)=a\}}
    +\frac{p^2}2\sum_{r=1}^h\1_{\mathsf{Supp}_r}\right\},\qquad L_0=1.
  \]
  Conditional on the complete observation sequences, past matrices and actions, and $\widetilde A_r$, the action law and $x_{r,i}y_{r,j}\ge p^2$ on $\mathsf{Supp}_r$ give
  \[
    \begin{aligned}
       & \E\left[L_r/L_{r-1}\,\middle|\,(\xi_{b,s})_{b,s},(\widetilde A_s,I_s,J_s)_{s<r},\widetilde A_r\right]
      =e^{p^2\1_{\mathsf{Supp}_r}/2}
      (1-\1_{\mathsf{Supp}_r}x_{r,i}y_{r,j}/2)                                                                 \\
       & \quad\le e^{\1_{\mathsf{Supp}_r}(p^2-x_{r,i}y_{r,j})/2}\le1,
    \end{aligned}
  \]
  where $1-u\le e^{-u}$ and $x_{r,i}y_{r,j}\ge p^2$ on $\mathsf{Supp}_r$.
  Iterating conditional expectations gives $\E[L_h]\le1$.
  For $t\ge4$, on $\{N_a(t)<p^2t/8,\ \sum_{r=1}^t\1_{\mathsf{Supp}_r^c}\le t/4\}$,
  \[
    \begin{aligned}
      \sum_{r=1}^{t-1}\1_{\mathsf{Supp}_r}
       & =t-1-\sum_{r=1}^{t-1}\1_{\mathsf{Supp}_r^c}\ge3t/4-1\ge t/2, \\
      \sum_{r=1}^{t-1}\1_{\mathsf{Supp}_r}\1_{\{(I_r,J_r)=a\}}
       & \le N_a(t)-1<p^2t/8.
    \end{aligned}
  \]
  Consequently, $L_{t-1}\ge\exp\{p^2t/4-(\log2)p^2t/8\}\ge e^{p^2t/8}$.
  Markov's inequality gives
  \[
    \Pp\left(N_a(t)<p^2t/8,\ \sum_{r=1}^t\1_{\mathsf{Supp}_r^c}\le t/4\right)
    \le e^{-p^2t/8}\E[L_{t-1}]
    \le e^{-p^2t/8}.
  \]
  For $t\le3$, the event is empty because $N_a(t)\ge1>p^2t/8$.
  Substituting these bounds into the union bound yields
  \[
    \Pp\left(\min_{a\in I^\star\times J^\star}N_a(t)<p^2t/8\right)
    \le d^2e^{-p^2t/8}+\frac4t\sum_{r=1}^t\Pp(\mathsf{Supp}_r^c)\longrightarrow0.
  \]
  Hence, the posterior-error probability tends to zero, proving core consistency.
  By continuity of the equilibrium, bounded convergence in probability, and Ces\`aro averaging, the last term in \eqref{proof:eq:core-frequency-convergence} tends to zero.
  By \eqref{proof:eq:core-frequency-convergence}, $N_{ij}(t)/t\to x_i^\star y_j^\star>0$ jointly in probability.
  Thus, the second limit in Lemma~\ref{proof:lem:gaussian-random-index-limit}, with $D=d^2$ and $\nu_a=x_i^\star y_j^\star$ for $a=(i,j)$, gives the joint limit $\mathsf W^{\rm core}$, with independent entries of variance $2\sigma^2/(x_i^\star y_j^\star)$.
\end{proof}

\subsection{Proof of Corollary~\ref{cor:core-coefficient-bound}}
\begin{proof}
  We write the coefficient as the sum of two expected maxima,
  \[
    \Lambda_{\rm core}=\E\left[\max_{j\in J^\star}(x^\star)_{I^\star}^\trans\mathsf W^{\rm core}e_j\right]
    +\E\left[\max_{i\in I^\star}(-e_i^\trans\mathsf W^{\rm core}(y^\star)_{J^\star})\right].
  \]
  The centered Gaussian projections in $\Lambda_{\rm core}$ have variances
  \[
    \operatorname{Var}((x^\star)_{I^\star}^\trans\mathsf W^{\rm core}e_j)=2\sigma^2/y_j^\star,
    \qquad
    \operatorname{Var}(e_i^\trans\mathsf W^{\rm core}(y^\star)_{J^\star})=2\sigma^2/x_i^\star.
  \]
  Applying Lemma~\ref{proof:lem:gaussian-maxima} separately to the two families gives
  \[
    \begin{aligned}
      \Lambda_{\rm core}
       & \le\sigma\sqrt{\frac2{\min_{j\in J^\star}y_j^\star}}\sqrt{2\log d}
      +\sigma\sqrt{\frac2{\min_{i\in I^\star}x_i^\star}}\sqrt{2\log d}      \\
       & =2\sigma\sqrt{\log d}
      \left[(\min_{j\in J^\star}y_j^\star)^{-1/2}
        +(\min_{i\in I^\star}x_i^\star)^{-1/2}\right],
    \end{aligned}
  \]
  where $-\min_i Z_i=\max_i(-Z_i)$ and centered Gaussian variables are symmetric.
  For the uniform equilibrium, each family consists of independent $\mathcal N(0,2\sigma^2d)$ variables.
  The lower bound in Lemma~\ref{proof:lem:gaussian-maxima} therefore applies to each family.

  Consequently,
  \[
    \Lambda_{\rm core}=2\sigma\sqrt{2d}\,m_d,
    \qquad
    \frac{3\sqrt2}{8}\sigma\sqrt{d\log d}
    \le\Lambda_{\rm core}\le4\sigma\sqrt{d\log d}.
  \]
\end{proof}

\section{Proofs for the Minimax Lower Bound}
\label{proof:sec:unique-minimax-lower}
\begin{proof}[Proof of Theorem~\ref{thm:unique-minimax-lower}]
  Fix $\pi,T$ and consider $A^\mu=(R/2)I_m+\mu\1^\trans$.
  For $\|\mu\|_\infty\le R/(8m)$, the strategies and value
  \[
    x^\mu=\frac\1m,\qquad
    y^\mu=\frac\1m-\frac2R\left(\mu-\frac\1m\sum_i\mu_i\right),\qquad
    v^\mu=\frac R{2m}+\frac1m\sum_i\mu_i
  \]
  satisfy $A^\mu y^\mu=v^\mu\1$, $(x^\mu)^\trans A^\mu=v^\mu\1^\trans$, $1/(2m)\le y_i^\mu\le3/(2m)$, and $\normmax{A^\mu}<R$.
  Positivity forces every optimal strategy to equalize all opposing payoffs: $(R/2)y_i+\mu_i=v^\mu$ and $(R/2)x_j+x^\trans\mu=v^\mu$ determine $y=y^\mu$ and $x=x^\mu$.
  Thus the equilibrium is unique and fully mixed.
  The nonnegative loss
  \[
    \ell_\mu(y)=\max_i\left\{\frac R2(y_i-1/m)+\mu_i-\frac1m\sum_h\mu_h\right\}
  \]
  is defined for all $\mu$, since the expressions inside the maximum sum to zero, and on the bounded parameter set satisfies
  \begin{equation}
    \Gap_{A^\mu}(x,y)\ge\max_i e_i^\trans A^\mu y-v^\mu=\ell_\mu(y).
    \label{proof:eq:main-row-offset-gap-lower}
  \end{equation}

  Draw independent $\mu_i\sim\mathcal N(0,\tau^2)$ with
  \[
    \tau=\min\left\{\sigma\sqrt{m/T},\frac{R}{32m\sqrt{\log(em)}}\right\},
  \]
  and write $\E_\Pi$ for expectation over this prior and the learning process.
  The observations satisfy $Y_s-(R/2)\1_{\{I_s=J_s\}}=\mu_{I_s}+\xi_s$.
  Conditional on the history and the policy's current independent randomization, its parameter-independent sampling kernels cancel from the likelihood, leaving independent Gaussian coordinates with means $\widehat\mu_{i,t}$ and variances
  \[
    \left(\tau^{-2}+\frac{N_i(t-1)}{\sigma^2}\right)^{-1},\qquad
    N_i(t-1)=\sum_{s<t}\1_{\{I_s=i\}}.
  \]
  At least $\lfloor m/2\rfloor$ rows have $N_i(t-1)\le2T/m$, hence posterior variance at least $\tau^2/3$.
  Apply Lemma~\ref{proof:lem:main-shifted-gaussian-maximum} to the centered posterior coordinates and shifts $(R/2)(y_{t,i}-1/m)+\widehat\mu_{i,t}-m^{-1}\sum_h\widehat\mu_{h,t}$.
  These shifts sum to zero, and subtracting the average centered coordinate does not change the expectation.
  Since $\lfloor m/2\rfloor/m\ge1/3$ and $\log\lfloor m/2\rfloor\ge(\log m)/4$ for $m\ge4$,
  \begin{equation}
    \E_\Pi\left[\sum_{t=1}^T\ell_\mu(y_t)\right]
    \ge\frac{T\tau\sqrt{\log m}}{32\sqrt3}.
    \label{proof:eq:main-row-offset-bayes-loss-explicit}
  \end{equation}

  To restrict the prior to bounded games, use $\ell_\mu(y)\le R/2+2\|\mu\|_\infty$ and the Gaussian union bound to obtain
  \[
    \begin{aligned}
      \E_\Pi\left[\1_{\{\|\mu\|_\infty>R/(8m)\}}\sum_{t=1}^T\ell_\mu(y_t)\right]
       & \le2mT e^{-R^2/(128m^2\tau^2)}
      \left(\frac R2+\frac{R}{4m}+\frac{16m\tau^2}{R}\right) \\
       & \le\frac54mRT e^{-R^2/(128m^2\tau^2)}.
    \end{aligned}
  \]
  Here, $m\ge4$ and $\tau\le R/(32m\sqrt{\log(em)})$.
  Since $R/(8m\tau)\ge4\sqrt{\log(em)}$ and $u e^{-u^2/2}$ decreases for $u\ge1$,
  \[
    \frac{(5/4)mRT e^{-R^2/(128m^2\tau^2)}}{T\tau\sqrt{\log m}/(32\sqrt3)}
    \le1280\sqrt6\,e^{-8}m^{-6}<\frac12.
  \]
  Subtracting this contribution from \eqref{proof:eq:main-row-offset-bayes-loss-explicit} and applying \eqref{proof:eq:main-row-offset-gap-lower} gives
  \[
    \E_\Pi[\1_{\{\|\mu\|_\infty\le R/(8m)\}}\mathcal R_T(A^\mu)]
    \ge\frac{T\tau\sqrt{\log m}}{64\sqrt3}.
  \]
  Some admissible $\mu$ attains this lower bound.
  Finally,
  \[
    T\tau\sqrt{\log m}
    =\min\left\{\sigma\sqrt{mT\log m},\frac{RT}{32m}\sqrt{\frac{\log m}{\log(em)}}\right\}
    \ge\frac1{64}\min\{\sigma\sqrt{mT\log m},RT/m\}.
  \]
  Taking $A_{\pi,T}=A^\mu$ and using $4096\sqrt3\le8192$ proves the claim.
\end{proof}

\section{Proofs for the Instance-Dependent Lower Bound}
\label{proof:sec:pure-separable-lower-proof}
\begin{proof}[Proof of Theorem~\ref{thm:pure-separable-lower}]
  Count all rounds, including initialization, by
  \[
    N_i^{\rm r}(T)=\sum_{t=1}^T\1_{\{I_t=i\}},\qquad
    N_j^{\rm c}(T)=\sum_{t=1}^T\1_{\{J_t=j\}}.
  \]
  Separability gives $\Gap_{A^\star}(x,y)=\sum_{i>1}(u_1-u_i)x_i+\sum_{j>1}(v_j-v_1)y_j$.
  The action law therefore implies
  \begin{equation}
    \E_{A^\star}^\pi[\mathcal R_T]
    =\sum_{i=2}^m(u_1-u_i)\E_{A^\star}^\pi[N_i^{\rm r}(T)]
    +\sum_{j=2}^n(v_j-v_1)\E_{A^\star}^\pi[N_j^{\rm c}(T)].
    \label{proof:eq:pure-regret-counts}
  \end{equation}
  Fix $i>1$ and sufficiently small $\varepsilon>0$.
  Under $A'=A^\star+(u_1-u_i+\varepsilon)e_i\1^\trans$, row $i$ is uniquely optimal and every other row has gap at least $\varepsilon$.
  Uniform efficiency implies, for every $\alpha>0$,
  \[
    \E_{A^\star}^\pi[N_i^{\rm r}(T)]=o(T^\alpha),\qquad
    \E_{A'}^\pi[T-N_i^{\rm r}(T)]=o(T^\alpha).
  \]
  Let $p_T=\Pp_{A^\star}^\pi(N_i^{\rm r}(T)\le T/2)$ and $q_T=\Pp_{A'}^\pi(N_i^{\rm r}(T)\le T/2)$.
  For $0<\alpha<1$, Markov's inequality gives
  \[
    1-p_T\le2\E_{A^\star}^\pi[N_i^{\rm r}(T)]/T=o(T^{\alpha-1}),\qquad
    q_T\le2\E_{A'}^\pi[T-N_i^{\rm r}(T)]/T=o(T^{\alpha-1}).
  \]
  Treating payoff entries as arms, apply \citet[Lemma~1]{kaufmann2016complexity} at the fixed horizon $T$.
  The same policy, including its randomization, is used under both models.
  Only row $i$ changes, with divergence $(u_1-u_i+\varepsilon)^2/(2\sigma^2)$ per observation, so eventually
  \[
    \begin{aligned}
      \frac{(u_1-u_i+\varepsilon)^2}{2\sigma^2}\E_{A^\star}^\pi[N_i^{\rm r}(T)]
       & \ge p_T\log\frac{p_T}{q_T}+(1-p_T)\log\frac{1-p_T}{1-q_T} \\
       & \ge p_T\log(1/q_T)-\log2                                  \\
       & \ge p_T(1-\alpha)\log T-\log2.
    \end{aligned}
  \]
  If $p_T$ or $q_T$ is zero or one, use the limiting values of binary relative entropy.
  Finite Gaussian divergence excludes an infinite right-hand side.
  Divide by $\log T$ and let $T\to\infty$, $\alpha\downarrow0$, and $\varepsilon\downarrow0$ in this order to get
  \[
    \liminf_{T\to\infty}\frac{\E_{A^\star}^\pi[N_i^{\rm r}(T)]}{\log T}
    \ge\frac{2\sigma^2}{(u_1-u_i)^2}.
  \]
  Applying the same argument to $A'=A^\star-(v_j-v_1+\varepsilon)\1e_j^\trans$ gives
  \[
    \liminf_{T\to\infty}\frac{\E_{A^\star}^\pi[N_j^{\rm c}(T)]}{\log T}
    \ge\frac{2\sigma^2}{(v_j-v_1)^2}\qquad(j>1).
  \]
  Substitute both bounds into the finite nonnegative sum in \eqref{proof:eq:pure-regret-counts} to conclude.
\end{proof}

\section{Experimental Details and Additional Results}
\label{app:experimental-details}

\subsection{Computational Setup}
\label{app:computing-resources}

All experiments were conducted on a Mac running macOS with an Apple M2 Max chip and 32\,GB of RAM.
For Nash-TS and symmetric UCB, we use \texttt{highspy}, the Python interface to HiGHS \citep{huangfu2018parallelizing}, to compute equilibrium strategies for the sampled and confidence-adjusted payoff matrices, respectively.
For Figure~\ref{fig:scaled-regret}, we compute the coefficient $2\Lambda_{\rm core}$ from the equilibrium strategies by numerical integration of the Gaussian-maximum expectations in \eqref{eq:core-asymptotic-constant} using SciPy \citep{virtanen2020scipy}.

\subsection{Game Instances}
\label{app:benchmark-games}

Table~\ref{tab:games} gives the payoff matrices and equilibria of the four structured games in Section~\ref{sec:experiments}.
The first three games have unique pure, fully mixed, and boundary equilibria, respectively.
In the fourth game, the row player's equilibrium strategies form a line segment, whereas the column player's equilibrium strategy is unique.
For the random payoff games, we draw each matrix entry independently from the uniform distribution on $[-1,1]$ and use $25$ or $50$ actions per player.

\begin{table}[H]
  \caption{Payoff matrices and Nash equilibria of the four structured games.
    The row player maximizes and the column player minimizes.}
  \label{tab:games}
  \centering
  \begin{tabular}{lll}
    \toprule
    Game & $A^\star$                                                                                                                     & Nash equilibria $(x^\star,y^\star)$ \\
    \midrule
    Unique pure
         & $\frac{1}{7}\left(\begin{smallmatrix}2 & 5 & 4\\ 1 & -3 & 6\\ 0 & 7 & -5\end{smallmatrix}\right)$
         & $x^\star=y^\star=(1,0,0)$                                                                                                                                           \\[1.6ex]
    Unique fully mixed
         & $\frac{1}{3}\left(\begin{smallmatrix}0 & -1 & 3\\ 1 & 0 & -1\\ -3 & 1 & 0\end{smallmatrix}\right)$
         & $x^\star=y^\star=(\nicefrac{1}{5},\nicefrac{3}{5},\nicefrac{1}{5})$                                                                                                 \\[1.6ex]
    Unique boundary mixed
         & $\frac{1}{4}\left(\begin{smallmatrix}1 & -1 & 4\\ -1 & 3 & 4\\ -2 & -2 & 0\end{smallmatrix}\right)$
         & $x^\star=y^\star=(\nicefrac{2}{3},\nicefrac{1}{3},0)$                                                                                                               \\[1.6ex]
    Multiple equilibria
         & $\frac{1}{10}\left(\begin{smallmatrix}1 & -2\\ -4 & 3\\ -10 & 9\end{smallmatrix}\right)$
         & $\begin{gathered}
               x^\star\in\{\lambda x_1^\star+(1-\lambda)x_2^\star:
               \lambda\in[0,1]\},\\
               y^\star=(\nicefrac{1}{2},\nicefrac{1}{2}),\\
               \text{where }x_1^\star=(\nicefrac{7}{10},\nicefrac{3}{10},0)\text{ and }x_2^\star=(\nicefrac{19}{22},0,\nicefrac{3}{22}).
             \end{gathered}$                                        \\
    \bottomrule
  \end{tabular}
\end{table}

\subsection{Baseline Algorithms and Hyperparameters}

\subsubsection{PMO-LB}

We implement the update rule of \citet{hait2026near} with doubling epoch lengths.
We start from uniform strategies without observing all payoff entries first.
At the end of each epoch, we compute the empirical payoff matrix using only observations from that epoch and solve its log-barrier-regularized game to obtain the next strategies.
We set entries not observed during the epoch to zero.
At epoch $s$, the regularization parameter is
\[
  \gamma_s
  =c_\gamma\max\{m,n\}2^{-s/2}
  \sqrt{\log\!
    \left(8\max\{m,n\}^2s^2/\delta\right)},
\]
where we fix $\delta=1/20$.
We use the Gaussian observations without clipping.
Since the original analysis assumes observations in $[-1,1]$, its guarantee does not directly apply to this implementation.

For each game, we select $c_\gamma$ from $\{2^{k/2}:k=-24,\ldots,16\}$ by minimizing the mean cumulative duality gap at $T=10^5$ over $20$ validation runs with $\sigma=1/5$.
For the structured games, these runs use trial indices $1000,\ldots,1019$.
For the random games, they use independent matrices with seeds $1000,\ldots,1019$ and one run per matrix.
Thus, the validation and evaluation runs use disjoint seeds.
Table~\ref{tab:pmo-hyperparameters} lists the selected values.

\begin{table}[ht]
  \caption{Selected regularization scales for PMO-LB.}
  \label{tab:pmo-hyperparameters}
  \centering
  \begin{tabular}{lc}
    \toprule
    Game                  & $c_\gamma$  \\
    \midrule
    Unique pure           & $1/64$      \\
    Unique fully mixed    & $2^{-11/2}$ \\
    Unique boundary mixed & $1/64$      \\
    Multiple equilibria   & $1/64$      \\
    Random $25\times25$   & $2^{-9}$    \\
    Random $50\times50$   & $2^{-21/2}$ \\
    \bottomrule
  \end{tabular}
\end{table}

To solve each regularized game, we parameterize the strategies by softmax logits with the last logit fixed to zero and solve the stationarity equations using SciPy's \texttt{least\_squares}.
We initialize the solver at the preceding solution.
The function, step, and gradient tolerances are $10^{-12}$, with at most $10^4$ function evaluations per solve.
We require the maximum absolute stationarity residual to be at most $10^{-8}\max\{1,\normmax{\widehat A},\gamma_s\max\{m,n\}\}$, where $\widehat A$ is the epoch's empirical payoff matrix.

\subsubsection{Symmetric UCB}
\label{app:symmetric-ucb}

We apply the UCB rule of \citet{odonoghue2021matrix} to both players using the shared empirical payoff matrix and observation counts.
As in Nash-TS, we first observe every payoff entry once in row-major order and include these rounds in the observation budget and cumulative duality gap.
Before online round $s$, let $N_{ij}(s)$ denote the observation count of entry $(i,j)$ and let $\widehat A(s)$ denote the empirical payoff matrix.
We define the confidence radius by
\[
  B_{ij,s}=c\sigma\sqrt{\frac{2\log(2T^2mn)}{N_{ij}(s)}}
\]
and fix $c=1$ for all games without tuning.
The players select
\[
  \begin{aligned}
    x_s & \in\argmax_{x\in\simplex_m}\min_{y\in\simplex_n}
    x^\trans(\widehat A(s)+B_s)y,                          \\
    y_s & \in\argmin_{y\in\simplex_n}\max_{x\in\simplex_m}
    x^\trans(\widehat A(s)-B_s)y.
  \end{aligned}
\]
We then sample $I_s\sim x_s$ and $J_s\sim y_s$ independently and update the empirical mean and count of the observed entry.
The column-player update is the UCB rule applied to $-(A^\star)^\trans$.
The single-learner guarantee of \citet{odonoghue2021matrix} does not directly give a cumulative duality gap bound for this symmetric implementation.
The initialization and scaling by $\sigma$ are specific to our implementation.

\FloatBarrier
\subsection{Additional Experimental Results}

\subsubsection{Comparison with VMG}
\label{app:vmg}

We compare Nash-TS with the value-incentivized matrix-game algorithm (VMG) of \citet[Algorithm~1]{yang2025incentivize} using the same games and experimental settings as in Section~\ref{sec:experiments}.
We use tabular one-hot features, the model class $[-1,1]^{m\times n}$, and a zero initial payoff estimate.

We set the KL regularization parameter to $\beta=0$ to evaluate the unregularized game.
For $\alpha_{\rm VMG}$, we use the coefficient in their Theorem~1 with feature dimension $mn$, failure probability $1/20$, and horizon $10^5$ iterations, without tuning.

At each iteration, VMG computes equilibrium strategies for its current payoff estimate and collects two payoff observations using different strategy pairs.
Following \citet{yang2025incentivize}, we report the cumulative gap of these equilibrium strategies, summed once per iteration, as a function of the number of payoff observations.
Figure~\ref{fig:vmg-cumulative-regret} shows that Nash-TS attains a smaller final mean cumulative gap on all six games under this evaluation.

\begin{figure}[htbp]
  \centering
  \includegraphics[width=\textwidth]{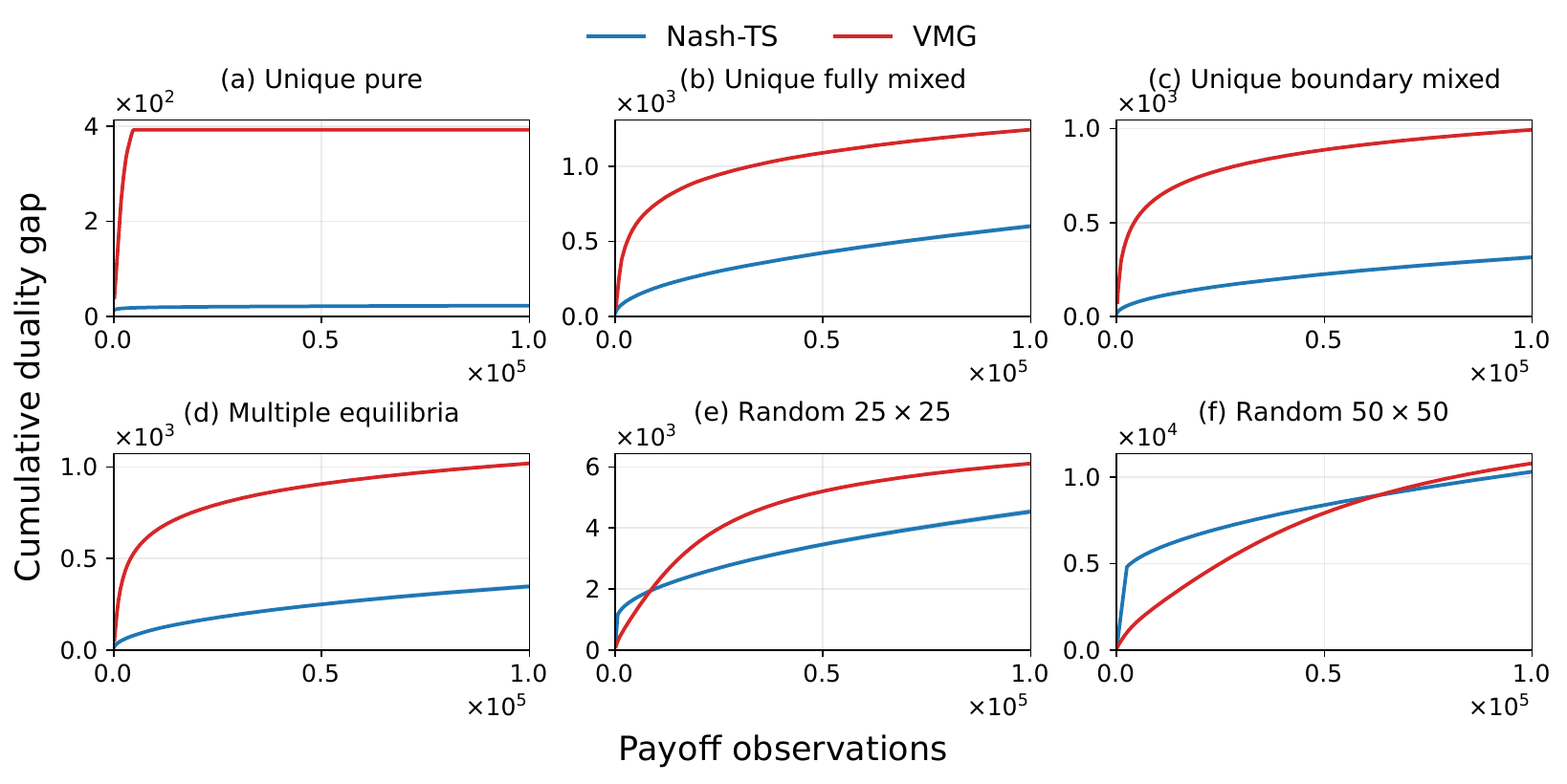}
  \caption{Cumulative duality gap of Nash-TS and VMG as a function of the number of payoff observations.
    Each curve is the mean over $100$ trials, and the shaded areas represent standard errors.
    Nash-TS sums the gaps of the played strategy pairs.
    VMG sums the gaps of its equilibrium strategies once per iteration, with two observations per iteration.
    These equilibrium strategies differ from the strategy pairs used to collect observations.}
  \label{fig:vmg-cumulative-regret}
\end{figure}

\FloatBarrier
\subsubsection{Comparison with Greedy}
\label{app:greedy}

To examine the effect of posterior sampling, we compare Nash-TS with a Greedy baseline.
We use the same one-observation-per-entry initialization as Nash-TS, including these rounds in the observation budget and cumulative duality gap.
After initialization, Greedy computes an equilibrium of the empirical payoff matrix at each round.
At each online round, Greedy samples the two actions independently from this equilibrium and updates the observed entry.
It uses no explicit exploration after initialization.
We use the same games and experimental settings as in Section~\ref{sec:experiments}.

\begin{figure}[htbp]
  \centering
  \includegraphics[width=\textwidth]{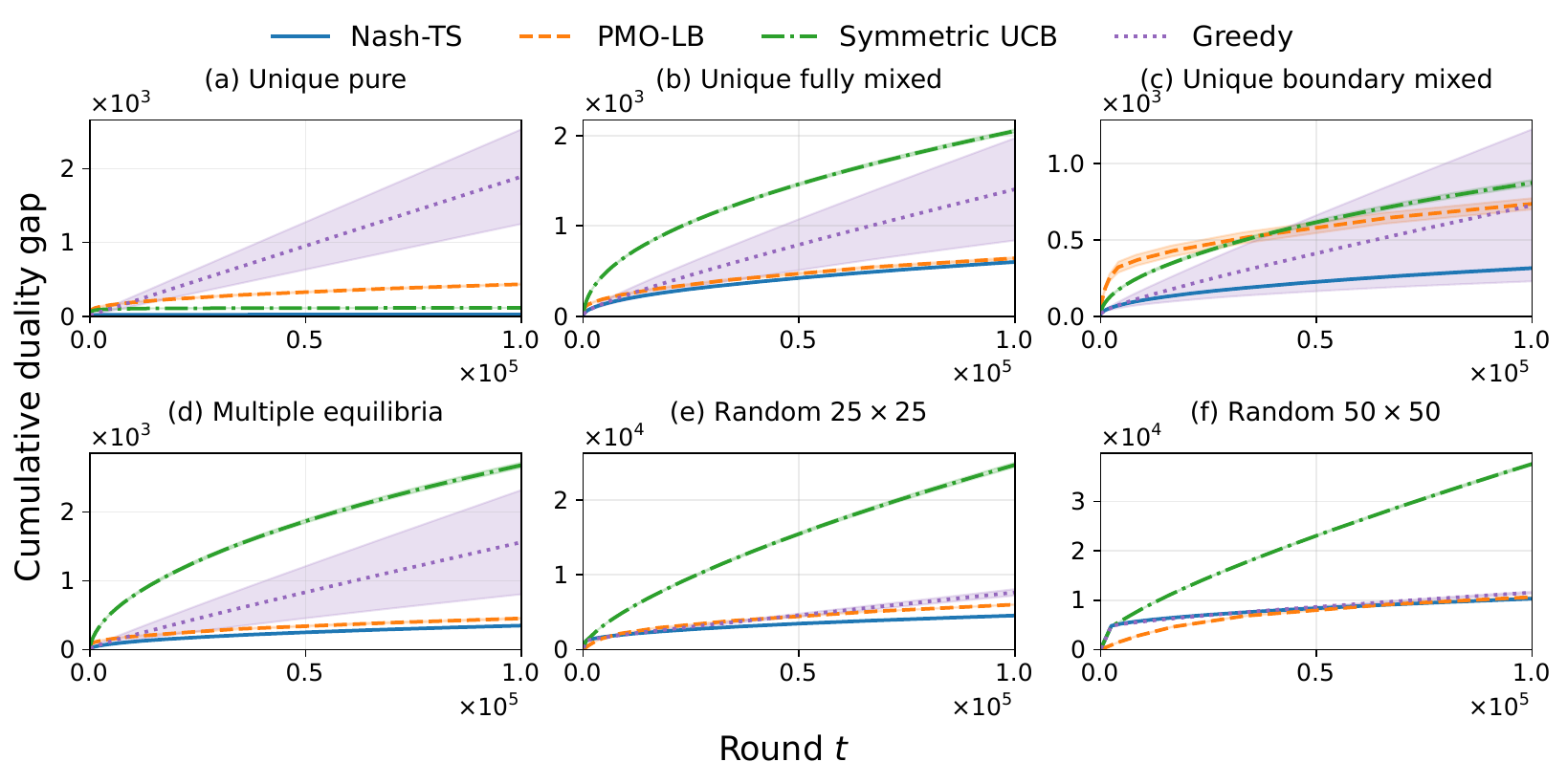}
  \caption{Cumulative duality gap of the Greedy baseline and the three methods in Figure~\ref{fig:regret}.
    Each curve is the mean over $100$ trials or independently generated matrices, with one-standard-error bands.}
  \label{fig:regret-greedy}
\end{figure}

Figure~\ref{fig:regret-greedy} compares Greedy with Nash-TS, PMO-LB, and symmetric UCB.
Greedy has substantially greater variability, consistent with premature commitment to noisy initial estimates.

\end{document}